\documentclass[11pt]{amsart}
\usepackage{amsmath,amssymb,bm}
\usepackage{a4wide}
\newtheorem{theorem}{Theorem}[section]
\newtheorem{lemma}[theorem]{Lemma}
\newtheorem{proposition}[theorem]{Proposition}

\newtheorem{claim}{Claim}[section]
\theoremstyle{remark}
\newtheorem{remark}{Remark}[section]
\numberwithin{equation}{section}
\newcommand{\R}{\mathbb{R}}
\newcommand{\C}{\mathbb{C}}
\newcommand{\N}{\mathbb{N}}
\newcommand{\Z}{\mathbb{Z}}
\newcommand{\T}{\mathbb{T}}
\newcommand{\la}{\langle}
\newcommand{\ra}{\rangle}
\newcommand{\pd}{\partial}

\newcommand{\mgamma}{\bm{\gamma}}
\newcommand{\mk}{\bm{k}}
\newcommand{\eps}{\varepsilon}
\newcommand{\bM}{\mathbb{M}}
\newcommand{\wR}{\widetilde{R}}
\newcommand{\pai}{\psi_{a,i}}
\newcommand{\tpa}{\tilde{\psi}_a}
\newcommand{\tpai}{\tilde{\psi}_{a,i}}
\DeclareMathOperator{\sech}{sech}
\DeclareMathOperator{\supp}{supp}
\DeclareMathOperator{\diag}{diag}

\DeclareMathOperator{\spann}{span}

\begin{document}
\title{Asymptotic stability of N-solitons of the FPU lattices}
\author{Tetsu Mizumachi}
\thanks{Faculty of Mathematics, Kyushu University,
6-10-1 Hakozaki, Fukuoka 812-8581 Japan,
mizumati@math.kyushu-u.ac.jp}
\begin{abstract}
We study stability of $N$-soliton solutions of the FPU lattice equation.
Solitary wave solutions of FPU cannot be characterized as a critical point of
conservation laws due to the lack of infinitesimal invariance in
the spatial variable. In place of standard variational arguments for
Hamiltonian systems, we use an exponential stability property of
the linearized FPU equation in a weighted space which is biased
in the direction of motion.
\par
The dispersion of the linearized FPU equation balances the potential term
for low frequencies, whereas the dispersion is superior for high frequencies.
We approximate the low frequency part of a solution of the linearized FPU
equation by a solution to the linearized KdV equation around an $N$-soliton.

We prove an exponential stability property of the linearized KdV equation
around $N$-solitons by using the linearized B\"acklund transformation and
use the result to analyze the linearized FPU equation.
\end{abstract}
\maketitle
\section{Introduction}
\label{sec:intro}
In this paper, we study stability of multi-pulse solutions of lattice equations
which describe motion of infinite particles
connected by nonlinear springs:
\begin{equation}
\label{eq:Toda}
\ddot{q}(t,n)=V'(q(t,n)-q(t,n-1))-V'(q(t,n+1)-q(t,n))
\quad\text{for $(t,n)\in\R\times\Z$,}
\end{equation}
where $q(t,n)$ denotes the displacement of the $n$-th particle at time $t$,
$V(r)$ denotes a kinetic potential and $\;\dot{}\;$ denotes differentiation
with respect to $t$. Making use of the change of variables
$p(t,n)=\dot{q}(t,n)$, $r(t,n)=q(t,n+1)-q(t,n)$ and
$u(t,n)={}^t\!(r(t,n),p(t,n))$,
we can translate \eqref{eq:Toda} into a Hamiltonian system
\begin{equation}
  \label{eq:FPU}
  \frac{du}{dt}=JH'(u),
\end{equation}
where $J=\begin{pmatrix} 0 & e^\pd-1 \\ 1-e^{-\pd} & 0
\end{pmatrix}$, $e^{\pm\pd}$ are the shift operators defined
by $(e^{\pm\pd})f(n)=f(n\pm1)$ and
\begin{gather*}
H(u(t))=\sum_{n\in\Z}\left(\frac12p(t,n)^2+V(r(t,n))\right)
\quad\text{(Hamiltonian).}
\end{gather*}
Typical examples of \eqref{eq:Toda} are the $\alpha$-FPU equation
($V(r)=\frac12r^2+\frac{1}{6}r^3$) and the Toda lattice equation
($V(r)=e^{-r}-1+r$).
\par
Originally, Fermi-Pasta and Ulam \cite{FPU} studied the FPU lattice
numerically to observe the equipartition of the energy among all Fourier modes
and found an almost recurrence phenomena contrary to their expectation.
Zabusky and Kruskal \cite{ZK} numerically found multi-solitons of KdV
that was known to describe the long wave solutions of FPU and interpreted
their result as an explanation of the FPU recurrent phenomena.
For recent development of metastability results on solitary waves of the
finite FPU lattice, see \cite{Ba-P} and the references therein.

The FPU lattice equation has solitary wave solutions due to a balance of
nonlinearity and dispersion induced by discreteness.
This was indicated by \cite{EF} by numerics before being proved by
Friesecke and Wattis \cite{FW} by using a concentration compactness theorem.
See also \cite{To} for the Toda lattice equation that is integrable and
has explicit $N$-soliton solutions.
\par

Eq. \eqref{eq:FPU} has two parameter family of solitary wave
solutions
$\{u_c(n-ct-\gamma) : c\in(-\infty,-1)\cup(1,\infty),\,\gamma\in\R\}$,
where $u_c=\begin{pmatrix}  r_c\\ p_c\end{pmatrix}$ is a solution of
\begin{equation}
  \label{eq:boundst}
c\pd_xu_c+JH'(u_c)=0.  
\end{equation}
In the case where $c$ is close to $1$ or $-1$,
Friesecke and Pego \cite{FP1} prove that solitary wave solutions are unique up
to translation and their shape are similar to KdV $1$-solitons.
We remark that a solitary wave solution $u_c(\cdot-ct)$ is small
if $c$ is close to $1$ or $-1$ and $\lim_{c\to\pm1}H(u_c)=0$.
\par
Friesecke and Pego also prove in \cite{FP2,FP3,FP4} that small solitary waves
of FPU are asymptotically stable in an exponentially weighted space.
Their idea is to compare spectral property of the linearized FPU equation and
the linearized KdV equation and to make use of the phenomena that the main
solitary wave moves fastest to the right (or to the left)
and it outruns from the rest of the solution as Pego and Weinstein \cite{PW}
did for KdV.
See also Mizumachi and Pego \cite{MP} that prove stability of Toda lattice
$1$-solitons of any size. More recently, Mizumachi \cite{Mi1} has proved
stability of $1$-soliton solutions of FPU in the energy space and
Hoffman and Wayne proved stability of two solitary waves which propagate
to the opposite directions.

\par
Our goal is to prove stability of $N$-solitons in the energy space.
In this paper, we assume
\begin{equation}
  \label{eq:H1}
V\in C^\infty(\R;\R),\quad V(0)=V'(0)=0, \quad V''(0)=1, \quad
 V'''(0)=\tfrac16,
\tag{H1}
\end{equation}
and use the following properties of solitary wave solutions proved by \cite{FP1}.
\begin{itemize}
\item[(P1)]
Let $c_*>1$ be a constant sufficiently close to $1$ and let $a\in[0,2)$.
For any $c\in(1,c_*]$, there exists a unique single hump
solution of \eqref{eq:boundst} in $l^2$ up to translation in $x$.
Moreover, 
$\sqrt{6(c-1)}=:\eps \mapsto \eps^{-2}u_c(\frac{\cdot}{\eps})
\in H^5(\R;e^{2a|x|})$ is $C^2$.
\item[(P2)] There exists an open interval $I$ such that 
$V''(r)>0$ for every $r\in I$ and that
$\overline{\{r_c(x):x\in\R\}}\subset I$ for every $c\in(1,1+c_*]$.
\item[(P3)] 
The solitary wave energy $H(u_c)$ satisfies
$dH(u_c)/dc\ne 0$ for $c\in(1,c_*]$.
\item[(P4)]
As $c$ tends to $1$, a shape of solitary wave solution becomes similar to
that of a KdV $1$-soliton. More precisely,
\begin{align*}
\sum_{j=0}^2\eps^j\left\|\pd_\varepsilon^j\left(
\eps^{-2}r_c\left(\frac{\cdot}{\eps}\right)-\sech^2x\right)
\right\|_{H^5(\R;e^{2a|x|}dx)}=O(\eps^{2}).
\end{align*}
\end{itemize}
\par
Now we state our main result.
\begin{theorem}
  \label{thm:1}
Let $0<k_1<\cdots<k_N$ and $c_{i,0}=1+\frac{k_i^2\eps^2}{6}$ $(1\le i\le N)$.
There exist positive numbers $\eps_0$, $\gamma_0$, 
$A_0$ and $\delta_0$ satisfying the following: Suppose
$\eps\in(0,\eps_0)$ and that $u(t)$ is a solution to
\eqref{eq:FPU} such that $\|v_0\|_{l^2}<\delta_0\eps^2$,
\begin{gather}
\label{eq:ic}
u(\cdot,0)=\sum_{i=1}^Nu_{c_{i,0}}(\cdot-x_{i,0})+v_0,
\\
\label{eq:dist}
L:=\min_{2\le i\le N}\eps(x_{i,0}-x_{i-1,0})\ge \frac{1}{k_1}
|\log(\delta_0\eps)|.
\end{gather}
Then there exist $C^1$-functions $x_i(t)$ $(i=1,\cdots,N)$ such that
\begin{equation}
  \label{eq:orbital-stability}
\sup_{t\ge0}\left\|u(\cdot,t)-\sum_{i=1}^Nu_{c_{i,0}}(\cdot-x_i(t))
\right\|_{l^2}< A_0(\|v_0\|_{l^2}+\eps^{\frac32}e^{-\gamma_0L}).
\end{equation}
\par
Furthermore, there exist $c_{N,+}>\cdots>c_{1,+}>1$ and
$c_*\in(1,(1+c_{1,0})/2)$ such that
\begin{gather}
  \label{eq:asympstability}
\lim_{t\to\infty}\left\|u(\cdot,t)-\sum_{i=1}^Nu_{c_{i,+}}
(\cdot-x_i(t))\right\|_{l^2(n\ge c_*t)}=0,
\\
\label{eq:spphconv}
\lim_{t\to\infty}\dot{x}_i(t)=c_{i,+}\quad{and}\quad
|c_{i,+}-c_{i,0}|< A_0(\eps^{-1}\|v_0\|_{l^2}^2+\eps^2e^{-\gamma_0L})
\quad\text{for $1\le i\le N$.}
\end{gather}
\end{theorem}

\begin{remark}
Eq. \eqref{eq:orbital-stability} implies orbital stability
of FPU co-propagating $N$-solitons since by (P4),
\begin{align*}
\|u_{c_{i,0}}\|_{l^2}^2=& 2\int_\R r_{c_{i,0}}(x)^2dx(1+o(1))  
= \frac{8k_i^3\eps^3}{3}(1+o(1)).
\end{align*}
\end{remark}
\begin{remark}
  The solitary waves moving to the same direction interact more
  strongly than counter-propagating solitary waves because they
  interact each other through their tails for a longer period.  Noting
  that the relative speeds between solitary waves are of $O(\eps^2)$,
  we see that the total impulse caused by the interaction of solitary
  waves is of $O(\eps^{\frac32}e^{-k_1L})=O(\eps^{\frac52})$ in the setting of Theorem
  \ref{thm:1}, whereas the total impulse caused by the interaction
  among counter-propagating solitary waves is of $O(\eps^{\frac72})$
  (\cite{Hoff-Way}).
\end{remark}
Orbital stability of KdV multi-solitons was first studied by Maddocks and
Sachs \cite{Madd-S} (see Kapitula \cite{Kapit} for other integrable
systems). In the nonintegrable case, Perelman \cite{Per1,Per2},
Rodgnanski-Schlag-Soffer \cite{RSS} proved stability of multi-solitons of
nonlinear Schr\"odinger equations that have super critical nonlinearities
by using scattering theory.
\par
Martel-Merle-Tsai \cite{MMT1,MMT2} studied stability of multi-soliton solutions
of gKdV and NLS by combining a variational argument (\cite[Chapter 8]{Ca})
and some propagation estimates. Their approach seems more favorable
because FPU has a subcritical nonlinearity. However, a solitary wave
solution cannot be characterized as a local minimizer because FPU does not
have a conservation law corresponding to momentum for KdV because the spatial
variable is defined on $\Z$.
\par
Instead of using the positivity of the second variation of a
conservation law  as is done in \cite{MMT1,MMT2},
we will use exponential linear stability property of the multi-soliton.
The idea of using exponential linear stability property  was applied to FPU
by Friesecke and Pego \cite{FP1,FP2,FP3,FP4} and lately used by Mizumachi
\cite{Mi1} to prove orbital stability of $1$-soliton solutions of FPU.
\par

We remark that most of propagation estimates of linearized dispersive
equations around multi-solitons are obtained in the case where relative speed
between solitary waves are large (Perelman \cite{Per1,Per2},
Rodgnanski-Schlag-Soffer \cite{RSS}, Hoffman-Wayne \cite{Hoff-Way}) so that
a dispersive wave mostly interacts with one solitary wave and
virtually has no interaction with the others. In these cases,
the problem can be reduced to that of $1$-soliton solutions by using Fourier
analysis or cut-off functions.
The other extreme case is where the relative speed is small
(Mizumachi \cite{Mi2}). In that case, $2$-soliton solutions can be treated
as a multi-bump bound state for a sufficiently long time.
\par
In our problem, a dispersive wave effectively interacts with all
the solitary waves which locate behind the dispersive wave at initial time
because the group velocity of plane waves is $\pm\cos\frac\xi2\in[-1,1]$
and velocity of solitary waves are larger than $1$.
Therefore, we need to consider exponential linear stability of
$N$-solitons without using cut-off functions in the spatial 
variable. 
\par

To prove exponential linear stability of FPU $N$-solitons, we
translate the linearized equation into a system of a high frequency
part, a middle frequency part and a low frequency part.  The high
frequency part is governed by a linearized FPU equation around the
null solution and the middle and low frequencies are in the KdV
regime. The behavior of middle frequency modes is approximated by
$u_t+u_{xxx}=0$ because the potential term turns out be negligible in
this region.  For low frequency modes, the dispersion and
the potential term are of the same order and its behavior is
governed by a linearized KdV equation around $N$-soliton solutions.
\par
Haragus and Sattinger \cite{HS} proved exponential linear stability
of linearized KdV equations in a class of analytic functions.
In this paper, we show the exponential linear stability in weighted $L^2$
spaces.
\par
Before we state our result, let us introduce several notations.
Let $0<k_1<\cdots<k_N$, $\gamma_i\in\R$, 
$\theta_i=k_i(x-4k_i^2t-\gamma_i)$ for $i=1,\dots,N$ and let 
$\mk=(k_1,\cdots,k_N)$, $\mgamma=(\gamma_1,\ldots, \gamma_N)\in\R^N$ and
$$
C_N=\left[\frac{1}{k_i+k_j}e^{-(\theta_i+\theta_j)}\right
]_{\substack{i=1,\dots,N\downarrow\\ j=1,\cdots,N\rightarrow}}.$$
Then $\varphi_N(t,x;\mk,\mgamma):=\pd_x^2\log\det(I+C_N)$ is
an N-soliton solution of KdV
\begin{equation}
  \label{eq:KdV}
  \pd_tu+\pd_x(\pd_x^2u+6u^2)=0\quad\text{for $x\in\R$ and $t>0$.}
\end{equation}
Especially $\varphi_1(t,x;k,\gamma)=k^2\sech^2k(x-4k^2t-\gamma)$.
\par
Let $a>0$ and
\begin{gather*}
\mathcal{P}(t,\mk,\mgamma):L^2_a\to
\spann\{\pd_{\gamma_i}\varphi_{N}(t,y;\mk,\mgamma),\;
\pd_{k_i}\varphi_{N}(t,y;\mk,\mgamma)\,:\, 1\le i\le N\},
\\  \mathcal{Q}(t,\mk,\mgamma)=I-\mathcal{P}(t)
\end{gather*}
be projections associated with
\begin{equation}
  \label{eq:LKdV}
\pd_tv+\pd_x(\pd_x^2v+12\varphi_N(\mk,\mgamma)v)=0
\quad\text{for $x\in\R$ and $t>0$.}
\end{equation}
 such that for $v\in\mathcal{Q}(t,\mk,\mgamma)$ and $i=1,\cdots,N$,
\begin{align}
\label{eq:secKdV1}
& \int_\R v(x)\int^x_{-\infty} \pd_{\gamma_i}\varphi_N(t,y;\mk,\mgamma)dydx=0,\\
\label{eq:secKdV2}
& \int_\R v(x)\int^x_{-\infty}\pd_{k_i}\varphi_N(t,y;\mk,\mgamma)dydx=0.
\end{align}
If $v$ is a solution of \eqref{eq:LKdV} and 
$v(s)\in\mathcal{Q}(s)$, then $v(t)\in\mathcal{Q}(t)$ for every $t\ge s$.
\begin{theorem}
  \label{thm:linearizedKdV}
Let $0<k_1<\ldots<k_N$, $0<a<2k_1$, $\theta\ge0$, $\eta\in(0,1)$ and let
$v(t,x)$ be a solution of \eqref{eq:LKdV}. Then there exists a positive
constant $K$ such that for every $t>s$ and $c$, $x_0\in\R$,
\begin{gather*}
\|e^{a(\cdot-ct-x_0)}\mathcal{Q}(t)v(t)\|_{L^2}\le
Ke^{-a(c-a^2)(t-s)}
\|e^{a(\cdot-cs-x_0)}\mathcal{Q}(s)v(s)\|_{L^2},
\\
\|e^{a(\cdot-ct-x_0)}\mathcal{Q}(t)v(t)\|_{L^2}\le
K(t-s)^{-\frac{\theta}{2}}e^{-\eta a(c-a^2)(t-s)}
\|e^{a(\cdot-cs-x_0)}\mathcal{Q}(s)v(s)\|_{H^{-\theta}}.
\end{gather*}
\end{theorem}
\par
Our plan of the present paper is as follows.
In Section \ref{sec:decomp}, we decompose a solution that is close to a family
of $N$-solitons into a sum of an $N$-soliton part and several
remainder parts and derive modulation equations on parameters of
speed and phase shift of the $N$-soliton part. 
In Section \ref{sec:energy-virial}, we estimate the energy norm of
the remainder parts and prove virial identities for each remainder part.
In Section \ref{sec:prthm1},  we prove orbital and asymptotic stability
of $N$-solitons assuming exponential linear stability of $N$-solitons
of FPU.
In Section \ref{sec:linear}, we will prove exponential linear stability of
small $N$-soliton solutions of FPU  assuming exponential stability property of
KdV. In Section \ref{sec:backlund}, we will use a linearized B\"acklund
transformation to prove Theorem \ref{thm:linearizedKdV}
following the idea of Mizumachi and Pego \cite{MP}.
We will show that a linearized B\"acklund transformation determines
an isomorphism that connects solutions of
$u_t+u_{xxx}=0$ and solutions of \eqref{eq:LKdV} satisfying 
\eqref{eq:secKdV1} and \eqref{eq:secKdV2} whose operator norm is
uniformly bounded with respect to $t$.
\par
Finally, let us introduce some notations.
Let $\la u,v\ra:=\sum_{n\in\Z}(u_1(n)u_2(n)+v_1(n)v_2(n))$
for $\R^2$-sequences $u=(u_1,u_2)$ and $v=(v_1,v_2)$ and let
$\|u\|_{l^2}=(\la u,u\ra)^{\frac12}$ and 
$\|u\|_{l^2_a}=\|e^{an}u(n)\|_{l^2}$.
We use notations $\|u\|_{L^2_a(\R)}=\|e^{ax}u(x)\|_{L^2(\R)}$ and
$\|u\|_{H^k_a(\R)}=\|e^{ax}u(x)\|_{H^k(\R)}$.

For Banach spaces $X$ and $Y$, we denote by $B(X,Y)$ the space of 
all linear continuous operators from $X$ to $Y$ and abbreviate
$B(X,X)$ as $B(X)$.
We use $a\lesssim b$ and $a=O(b)$ to mean that there exists a positive constant
such that $a\le Cb$. For any $f\in l^2$,
$$(\mathcal{F}_nf)(\xi)=\tilde{f}(\xi)
=\frac{1}{\sqrt{2\pi}}\sum_{n\in\Z}f(n)e^{-in\xi},$$
and $(f_1*_\T f_2)(x)=\int_\T f_1(x-y)f_2(y)dy$ for $f_1$, $f_2\in L^2(\T)$,
where $\T=\R/2\pi\Z$.
We denote by $\tau_h$ a translation operator defined by $(\tau_hf)(x):=f(x+h)$.
\bigskip

\section{Decomposition of the solution}
\label{sec:decomp}
Let $u(t)$ be a solution to \eqref{eq:FPU} which lies in a tubular
neighborhood of
$$\mathcal{M}=\left\{\sum_{i=1}^Nu_{c_{i,0}}(\cdot-y_i):
y_{i+1}-y_i>L\quad\text{for $i=1,\cdots,N-1$}\right\},$$
where $L$ is sufficiently large.

We decompose a solution around $\mathcal{M}$ as
\begin{equation}
  \label{eq:decomp}
u(t)=\sum_{1\le i\le N}u_{c_i(t)}(\cdot-x_i(t))+v(t),
\end{equation}
where $u_{c_i(t)}(\cdot-x_i(t))$ $(i=1,\cdots,N)$ denote solitary waves
and $c_i(t)$ and $x_i(t)$ are modulation parameters of the speed
and the phase shift of each solitary wave, respectively.
Let $U_N(t)=\sum_{i=1}^Nu_{c_i(t)}(\cdot-x_i(t)).$
 Substituting \eqref{eq:decomp} into
\eqref{eq:FPU}, we  have
\begin{equation}
  \label{eq:v}
\pd_tv=JH''(U_N)v+l+R,  
\end{equation}
 where $R=R_1+R_2$ and
\begin{align*}
& R_1= JH'(U_N+v)-JH'(U_N)-JH''(U_N)v,\\
& R_2= JH'(U_N)-\sum_{i=1}^NJH'(u_{c_i(t)}(\cdot-x_i(t))),\\
& l=-\sum_{i=1}^N\left\{\dot{c}_i\pd_cu_{c_i}(\cdot-x_i(t))
-(\dot{x}_i-c_i)\pd_xu_{c_i}(\cdot-x_i(t))\right\}.
\end{align*}
Now we decompose $v(t)$ into the sum of a small solution $v_1(t)$ to
\eqref{eq:FPU} and a remainder term which belongs to $l^2_a$ and is localized
around solitary waves.
Let $v_1(t)$ be a solution to
\begin{equation}
  \label{eq:v1}
\left\{
  \begin{aligned}
&  \pd_tv_1=JH'(v_1),\\
&   v_1(0)=v_0,
  \end{aligned}\right.
\end{equation}
and $v_2(t)=v(t)-v_1(t)$.
By \cite[Proposition 3]{Mi1},
we see that $u(t)-v_1(t)$ remains in $l^2_a$ for every
$0\le a<2\min_{1\le i\le N}\kappa(c_{i,0})$ and $t\in\R$, where
$\kappa(c)$ is a positive root of $c=\sinh\kappa/\kappa$.
\par

Suppose $x_i(t)$ and $c_i(t)$ are of class $C^1$. 
Then if $u(t)=U_N(t)+v_1(t)+v_2(t)$ is a solution to \eqref{eq:FPU},
\begin{equation}
  \label{eq:v2}
\left\{
  \begin{aligned}
& \pd_tv_2= JH''(U_N(t))v_2+l(t)+\wR(t),
\\ & v_2(0)=0,
  \end{aligned}\right.
\end{equation}
where $\wR(t)=R(t)-JH'(v_1(t))+JH''(U_N(t))v_1$.
Our strategy is to derive modulation equations on $x_i(t)$ and $c_i(t)$ and
{\it \`a priori} estimates on $v_2$, $x_i$ and $c_i$ $(1\le i\le N)$
to prove that $u$ remains in a tubular neighborhood of $\mathcal{M}$ in $l^2$.
To prove convergence of speed parameters $c_i(t)$ $(1\le i\le N)$,
we need to estimate $v_2(t)$ in an exponential weighted space.
Since $e^{-k_1\eps x_1(t)}\|v_2(t)\|_{l^2_{k_1\eps}}$ may grow as $t\to\infty$
due to the interaction between $v_1(t)$ and 
solitary waves $u_{c_i}(\cdot-x_i(t))$ $(i\ge 2)$,
we will decompose $v_2(t)$ into a sum of $N$ functions $v_{2k}$ 
$(1\le k\le N)$ such that each $v_{2k}(t)$ remains small in a weighted space
$$X_k(t)=\left\{v\in l^2_{k_1\eps} : \|v\|_{X_k(t)}=
\left(\sum_{n\in\Z}e^{k_1\eps(n-x_{N+1-k}(t))}|v(n)|^2\right)^{\frac12}<\infty
\right\}.$$
\par
Let $Q_k(t)\colon l^2_a\to l^2_a$ be an operator defined by
$$Q_k(t)f=f-\sum_{N+1-k\le i\le N}(\alpha_i(f)\pd_xu_{c_i}(\cdot-x_i(t))
+\beta_i(f)\pd_cu_{c_i}(\cdot-x_i(t)))$$
for $a>0$, where $\alpha_i(f)$ and $\beta_i(f)$ $(i=1,\cdots,N)$ are real
numbers satisfying 
$$
\la Q_k(t)f,J^{-1}\pd_xu_{c_i}(\cdot-x_i(t))\ra
=\la Q_k(t)f,J^{-1}\pd_cu_{c_i}(\cdot-x_i(t))\ra=0
$$
for $N+1-k\le i\le N$ and let $P_k(t)=I-Q_k(t)$.
We remark that if $a>0$,
\begin{equation}\label{eq:J-1}
J^{-1}= \begin{pmatrix} 0 & \sum_{k=-\infty}^0 e^{k\pd}
\\ \sum_{k=-\infty}^{-1}e^{k\pd} & 0 \end{pmatrix}    
\end{equation}
is a bounded operator on $l^2_{-a}$  because
$\|e^{-\pd}u\|_{l^2_{-a}}=e^{-a}\|u\|_{l^2_{-a}}$ and that
$J^{-1}\pd_cu_c$ and  $J^{-1}\pd_xu_c$ belong to $l^2_{-a}$
for any $a\in(0,2\kappa(c))$.
\par

Let $v_{2k}(t)$ $(1\le k\le N-1)$ be a solution of 
\begin{equation}
  \label{eq:v2k}
\left\{
  \begin{aligned}
&  \pd_tv_{2k}=JH''(U_k)v_{2k}+l_k+Q_k(t)JR_k,\\
& v_{2k}(0)=0,
  \end{aligned}\right.
\end{equation}
where $w_0=v_1$, $w_k=v_1+\sum_{1\le i\le k}v_{2i}$ $(1\le k\le N)$,
\begin{align*}
& R_{k}=H'(U_k+w_k)-H'(u_{c_{N+1-k}})-H'(U_{k-1}+w_{k-1})-H''(U_k)v_{2k},
\\ & 
l_k=\sum_{N+1-k\le j\le N}(\alpha_{j,k}\pd_cu_{c_j}+\beta_{j,k}\pd_xu_{c_j}),
\end{align*}
and $\alpha_{j,k}$ and $\beta_{j,k}$ ($N+1-k\le j\le N$, $1\le k\le N-1$)
are continuous functions that will be defined later.
\par
Let $v_{2N}(t)=v_2(t)-\sum_{1\le i\le N-1}v_{2i}(t)$.
To fix the decomposition \eqref{eq:decomp}, we will define $c_i(t)$
and $x_i(t)$ $(1\le i\le N)$ so that
\begin{align}
\label{eq:orth1}
& \la v_{2N}(t), J^{-1}\pd_xu_{c_i(t)}(\cdot-x_i(t))\ra=0
\quad\text{for $i=1,\cdots,N$,}\\
\label{eq:orth2}
& \la v_{2N}(t),J^{-1}\pd_cu_{c_i(t)}(\cdot-x_i(t))\ra=0
\quad\text{for $i=1,\cdots, N$.}
\end{align}
\par

By \eqref{eq:v}, \eqref{eq:v1} and \eqref{eq:v2k},
\begin{equation}
  \label{eq:v2N}
  \begin{split}
 \pd_tv_{2N}=& J\left\{H'(U_N+v)-\sum_{k=1}^NH'(u_{c_k})-H'(v_1)\right\}+l
\\ & -\sum_{k=1}^{N-1}(JH''(U_k)v_{2k}+Q_k(t)JR_k+l_k)
\\ =& JH''(U_N)v_{2N}+JR_N+\sum_{k=1}^{N-1}(P_k(t)JR_k-l_k)+l.
  \end{split}
\end{equation}
\par
Let $\mathcal{A}_k=\begin{pmatrix}\mathcal{A}_{i,j}
\end{pmatrix}_{\substack{i=N+1-k,\dots,N\downarrow\\
 j=N+1-k,\dots,N\rightarrow}}$,
$F_{j,k}={}^t(F_{j,k}^1,F_{j,k}^2)$ and 
\begin{align*}
& \mathcal{A}_{i,j}=
\begin{pmatrix}
  \eps^{-1}\la \pd_cu_{c_j},J^{-1}\pd_xu_{c_i}\ra 
  & \eps^{-4}\la \pd_xu_{c_j},J^{-1}\pd_xu_{c_i}\ra \\ 
   \eps^2\la \pd_cu_{c_j},J^{-1}\pd_cu_{c_i}\ra
  & \eps^{-1}\la \pd_xu_{c_j},J^{-1}\pd_cu_{c_i}\ra
 \end{pmatrix},\\
& F_{j,k}^1=
  \eps^{-4}\la v_{2k},(H''(U_k)-H''(u_{c_j}))\pd_xu_{c_j}\ra
\\ & +\eps^{-4}\{(\dot{x}_j-c_j)\la v_{2k},J^{-1}\pd_x^2u_{c_j}\ra
  -\dot{c}_j\la v_{2k},J^{-1}\pd_c\pd_xu_{c_j}\ra\},
\\ &
F_{j,k}^2=\eps^{-1}\{\la v_{2k},(H''(U_k)-H''(u_{c_j}))\pd_c{u}_{c_j}\ra
\\ & +\eps^{-1}\{(\dot{x}_j-c_j)\la v_{2k},J^{-1}\pd_c\pd_xu_{c_j}\ra
  -\dot{c}_j\la v_{2k},J^{-1}\pd_c^2u_{c_j}\ra\}.
\end{align*}
If $\alpha_{j,k}(t)$ and $\beta_{j,k}(t)$ are chosen to be a solution of
\begin{equation}
  \label{eq:orthv2k3}
\mathcal{A}_k\begin{pmatrix}  \eps^{-3}\alpha_{j,k}\\
\beta_{j,k}\end{pmatrix}_{N+1-k\le j\le N\downarrow}
=\begin{pmatrix} F_{j,k}\end{pmatrix}_{N+1-k\le j\le N\downarrow},
\end{equation}
then $v_{2k}$ $(1\le k\le N-1)$ satisfy secular term conditions.
\begin{lemma}
  \label{lem:alphabeta}
Suppose that $x_i(t)$ and $c_i(t)$ $(1\le i\le N)$ are of class $C^1$ on
$[0,T]$ and that $v_{2k}$ $(1\le k\le N-1)$ satisfy \eqref{eq:v2k} and
\eqref{eq:orthv2k3} for $1\le k\le N-1$ and $t\in[0,T]$. Then
\begin{equation}
\label{eq:orthv2k}
\la v_{2k},J^{-1}\pd_xu_{c_i}\ra=\la v_{2k},J^{-1}\pd_cu_{c_i}\ra=0
\end{equation}
for every $N+1-k\le i\le N$, $1\le k\le N-1$ and $t\in[0,T]$.
\end{lemma}
\begin{proof}
First, we recall that $H(u_c(\cdot-ct))$ does not depend on $t$ and
\begin{align}
  \label{eq:secular1}
& \la \pd_xu_c,J^{-1}\pd_xu_c\ra=-\frac1c\la \pd_xu_c,H'(u_c)\ra
=\frac{1}{c^2}\frac{d}{dt}H(u_c(\cdot-ct))=0,\\
& \label{eq:secular2}
\la \pd_xu_c,J^{-1}\pd_cu_c\ra=-\la \pd_cu_c,J^{-1}\pd_xu_c\ra
=\frac1c\frac{d}{dc}H(u_c)>0.
\end{align}
Differentiating \eqref{eq:boundst} with respect to $x$ and $c$,
we have
\begin{equation}
  \label{eq:secularmode}
c\pd_x^2u_{c}+JH''(u_{c})\pd_xu_{c}=0,\quad
c\pd_c\pd_xu_{c}+JH''(u_{c})\pd_cu_{c}=-\pd_xu_c.  
\end{equation}
Using \eqref{eq:v2k}, \eqref{eq:secularmode}, $J^*=-J$  and the fact that
$J^{-1}\pd_xu_{c_j}$ and $J^{-1}\pd_c{u}_{c_j}$ $(N+1-k\le j\le N)$ are
orthogonal to the range of the projection $Q_k(t)$,
we have for  $N+1-k\le j\le N$ and $1\le k\le N-1$
\begin{align*}
&  \frac{d}{dt}\la v_{2k},J^{-1}\pd_xu_{c_j}(\cdot-x_j(t))\ra
\\=& 
\la JH''(U_k)v_{2k}+l_k+Q_kJR_k,J^{-1}\pd_xu_{c_j}\ra
\\ & -\dot{x}_j\la v_{2k},J^{-1}\pd_x^2u_{c_j}\ra
+\dot{c}_j\la v_{2k},J^{-1}\pd_c\pd_xu_{c_j}\ra
\\=& \la l_k,J^{-1}\pd_xu_{c_j}\ra
+\la v_{2k},(H''(u_{c_j})-H''(U_k))\pd_xu_{c_j}\ra
\\ &
+\dot{c}_j\la v_{2k},J^{-1}\pd_c\pd_xu_{c_j}\ra
-(\dot{x}_j-c_j)\la v_{2k},J^{-1}\pd_x^2u_{c_j}\ra
\\=&
\sum_{i=N+1-k}^N
\left(\alpha_{i,k}\la \pd_cu_{c_i},J^{-1}\pd_xu_{c_j}\ra
+\beta_{i,k}\la \pd_xu_{c_i},J^{-1}\pd_xu_{c_j}\ra\right)
\\ & -\la v_{2k},(H''(U_k)-H''(u_{c_j}))\pd_xu_{c_j}\ra
\\ & -(\dot{x}_j-c_j)\la v_{2k},J^{-1}\pd_x^2u_{c_j}\ra
+\dot{c}_j\la v_{2k},J^{-1}\pd_c\pd_xu_{c_j}\ra,
\end{align*}
and 
  \begin{align*}
&  \frac{d}{dt}\la v_{2k},J^{-1}\pd_cu_{c_j}(\cdot-x_j(t))\ra
\\=&
\la JH''(U_k)v_{2k}+l_k+Q_kJR_k,J^{-1}\pd_cu_{c_j}\ra
-\dot{x}_j\la v_{2k},J^{-1}\pd_x\pd_c{u}_{c_j}\ra
+\dot{c}_j\la v_{2k},J^{-1}\pd_c^2u_{c_j}\ra
\\=&
\la l_k,J^{-1}\pd_cu_{c_j}\ra+
\la v_{2k},(H''(u_{c_j})-H''(U_k))\pd_cu_{c_j}\ra
+\la v_{2k},J^{-1}\pd_xu_{c_j}\ra
\\ & 
+\dot{c}_j\la v_{2k},J^{-1}\pd_c^2u_{c_j}\ra
-(\dot{x}_j-c_j)\la v_{2k},J^{-1}\pd_c\pd_xu_{c_j}\ra
\\=&
\sum_{i=N+1-k}^N
\left(\alpha_{i,k}\la \pd_cu_{c_i},J^{-1}\pd_cu_{c_j}\ra
+\beta_{i,k}\la \pd_xu_{c_i},J^{-1}\pd_cu_{c_j}\ra\right)
\\ & -\la v_{2k},(H''(U_k)-H''(u_{c_j}))\pd_c{u}_{c_j}\ra
\\ & -(\dot{x}_j-c_j)\la v_{2k},J^{-1}\pd_c\pd_xu_{c_j}\ra
+\dot{c}_j\la v_{2k},J^{-1}\pd_c^2u_{c_j}\ra+\la v_{2k},J^{-1}\pd_xu_{c_j}\ra. 
\end{align*}
In the course of calculations, we abbreviate
$u_{c_j(t)}(\cdot-x_j(t))$ as $u_{c_j}$.
Substituting \eqref{eq:orthv2k3} into the above, we have for
$N+1-k\le j\le N$,
\begin{equation*}
  \frac{d}{dt}\la v_{2k}(t),J^{-1}\pd_xu_{c_j}\ra=0,\quad
  \frac{d}{dt}\la v_{2k}(t),J^{-1}\pd_cu_{c_j}\ra=\la v_{2k},J^{-1}\pd_xu_{c_j}\ra.
\end{equation*}
Since $v_{2k}(0)=0$,  we have \eqref{eq:orthv2k} for every $1\le j\le N$,
$N+1-k\le k\le N-1$ and $t\in[0,T]$.
Thus we complete the proof.
\end{proof}

Next we will derive modulation equations of $x_i$ and $c_i$ so that
$v_{2N}$ satisfies \eqref{eq:orth1} and \eqref{eq:orth2}.
\begin{lemma}
  \label{lem:modulation}
Let $u(t)$ be a solution of \eqref{eq:FPU} and $v_1(t)$ be a solution of
\eqref{eq:v1}. There exist positive numbers $L$, $\varepsilon_0$ and $\delta$
satisfying the following: Suppose $\varepsilon\in(0,\varepsilon_0)$,
that $c_i(t)$ and $x_i(t)$ $(i=1,\cdots,N)$ are $C^1$-functions satisfying
\eqref{eq:orth1} and \eqref{eq:orth2} on $[0,T]$ and that
\begin{gather*}
\max_{1\le i\le N}\sup_{t\in[0,T]}
(|c_i(t)-c_{i,0}|+|\dot{x}_i(t)-c_i(t)|)\le \delta\eps^2,
\\
\min_{1\le i\le N-1}\inf_{t\in[0,T]}(x_{i+1}(t)-x_i(t))\ge \eps^{-1}L,
\\
\sup_{t\in[0,T]}(\|v_1(t)\|_{W(t)}+\sum_{1\le k\le N}
\|v_{2k}(t)\|_{X_k(t)\cap W(t)})\le \delta\eps^{\frac32}.
\end{gather*}
Let $\sigma=\frac12\eps^{-2}\min_{2\le i\le N}(c_{i,0}-c_{i-1,0})$.
Then for $t\in[0,T]$,
\begin{equation}
\label{eq:modeqc}
  \begin{split}
 & \frac{d}{dt}\left\{c_i(t)\left(1-\theta_1(c_i(t))^{-1}
\la v_1(t)+\sum_{k=1}^{N-i}v_{2k}(t),\rho_{c_i(t)}\ra\right)\right\}
\\ = &
O\left(\eps^2\left(\|v_1(t)\|_{W(t)}^2+\sum_{k=1}^N
\|v_{2k}(t)\|_{W(t)\cap X_k(t)}^2\right)
+\eps^5e^{-2k_1(\sigma\eps^3t+L)}\right),    
  \end{split}
\end{equation}
\begin{equation}
  \label{eq:modeqx}
  \begin{split}
& \dot{x}_i(t)-c_i(t)\\ = & O\left(\eps^{\frac12}
\left(\|v_1(t)\|_{W(t)}+\sum_{k=1}^N\|v_{2k}(t)\|_{W(t)}\right)
+\eps^2e^{-k_1(\sigma\eps^3t+L)}\right),      
  \end{split}
\end{equation}
where $\theta_1(c)=dH(u_c)/dc$,
$\rho_c=\pd_x(c\pd_x+J)^{-1}(H'(u_c)-u_c)$ and
$$\|u\|_{W(t)}=\sum_{1\le i\le N}\|e^{-k_i\eps|n-x_i(t)|/2}u\|_{l^2},
\quad \|u\|_{X_k(t)\cap W(t)}=\|u\|_{X_k(t)}+\|u\|_{W(t)}.$$
\end{lemma}
\begin{remark}
A solution of a system \eqref{eq:v1}, \eqref{eq:v2k}, \eqref{eq:v2N},
\eqref{eq:orthv2k3}, \eqref{eq:modeqc} and \eqref{eq:modeqx} 
(more precisely \eqref{eq:modeq1}) exists at least locally in time.
If it satisfies an initial condition
\begin{equation}
  \label{eq:IC}
v_1(0)=v_1,\quad v_{21}(0)=\cdots=v_{2N}(0)=0,\quad x_i(0)=x_{i,0}, \quad
c_{i}(0)=c_{i,0},  
\end{equation}
then $u(t)=\sum_{i=1}^Nu_{c_i(t)}(\cdot-x_i(t))+v_1(t)+\sum_{k=1}^Nv_{2k}(t)$
becomes a solution to \eqref{eq:FPU} and
$$u(0)=\sum_{i=1}^Nu_{c_{i,0}}(\cdot-x_{i,0})+v_0.$$
\end{remark}

To prove Lemma \ref{lem:modulation}, we need the following:
\begin{lemma}
  \label{lem:ANbd}
Suppose that $c_i(t)$ and $x_i(t)$ be as in Lemma \ref{lem:modulation}.
Then there exists a positive constant $C$ depending only on 
$k_1,\cdots,k_N$, $\eps_0$, $\delta$ and $L_0$ such that
$$\sup_{t\in[0,T]}\left(|\mathcal{A}_{i,j}|+|\mathcal{A}_k^{-1}|\right)\le C
\quad\text{for $1\le i,j,k\le N$.}$$
\end{lemma}
\begin{lemma}
  \label{lem:FPUprb}
Suppose that $c_i(t)$ and $x_i(t)$ be as in Lemma \ref{lem:modulation}.
Then there exists a positive constant $C$ depending only on 
$k_1,\cdots,k_N$, $\eps_0$, $\delta$ and $L_0$ such that
$$\sup_{t\ge0}(\|P_k(t)\|_{B(l^2_{k_1\eps})}
+\eps^{-1}\|P_k(t)J\|_{B(l^2)})\le C
\quad\text{for $1\le k\le N$.}$$
\end{lemma}

\begin{proof}[Proof of Lemma \ref{lem:ANbd}]
Let $\theta_2(c)=\la \pd_c p_{c},1\ra\la \pd_cr_{c},1\ra$,
\begin{align*}
& \theta_3(c_i,c_j)=\la \pd_cp_{c_i},1\ra \la \pd_c r_{c_j},1\ra
+\la \pd_cp_{c_j},1\ra\la \pd_cr_{c_i},1\ra,\\
& \sigma_3=\begin{pmatrix}1 & 0 \\ 0 & -1\end{pmatrix},
\quad B_1(c)=-(c\eps)^{-1}\theta_1(c)\sigma_3+\eps^2\theta_2(c)
\begin{pmatrix}0 & 0\\ 1 & 0\end{pmatrix},
\\ & B_2(c_i,c_j)=\eps^2\theta_3(c_i,c_j)
\begin{pmatrix}0 & 0\\ 1 & 0\end{pmatrix},\quad
 B_3(c_i,c_j)=-B_1(c_i)^{-1}B_2(c_i,c_j)B_1(c_j)^{-1}.
\end{align*}
By \eqref{eq:secular1} and \eqref{eq:secular2},
we have $\mathcal{A}_{ii}=B_1(c_i)$.
Since
\begin{align*}
  x_i(t)-x_j(t)\ge & x_i(0)-x_j(0)
+\int_0^t(\dot{x}_i(s)-\dot{x}_j(s))ds
\\ \ge & \eps^{-1}L+(c_{i,0}-c_{j,0}-2\delta\eps^2)t
\\ \ge & \sigma\eps^2t+\eps^{-1}L\quad\text{for $i>j$,}
\end{align*}
it follows from Claims \ref{cl:intsize} and \ref{cl:J-1u} that
$$
\mathcal{A}_{i,j}=\left\{
  \begin{aligned}
&   B_2(c_i,c_j)+O(e^{-k_i(\sigma\eps^3t+L)}) &\text{if $i<j$,}\\
&   O(e^{-k_j\eps|x_i-x_j|}) &\text{if $i>j$.}
  \end{aligned}\right.$$
By a simple computation,
\begin{align*}
\mathcal{A}_N^{-1}=&
  \begin{pmatrix}
    B_1(c_1)^{-1} & B_3(c_1,c_2)& \cdots     &        & B_3(c_1,c_k)\\
           & B_1(c_2)^{-1}    & B_3(c_2,c_3) &        & \vdots\\
           &           & \ddots     & \ddots &  \\
           &     &            & B_1(c_{k-1})^{-1} & B_3(c_{k-1},c_k)\\
         O & & & & B_1(c_k)^{-1}
  \end{pmatrix}
\\ &+O(e^{-k_1(\sigma\eps^3t+L)}).
\end{align*}
\par
Next we prove that $B_1(c_i)$, $B_1(c_i)^{-1}$ and $B_2(c_i,c_j)$ are uniformly
bounded in $\eps$ in the case where
$V(r)=e^{r}-1-r$ (the Toda lattice).
By \cite{To},
\begin{align*}
& q_c(x)=-\log\frac{\cosh\{\kappa(x-1)\}}{\cosh\kappa x},\\
& p_c(x)=-c\pd_xq_c(x),\quad r_c(x)=q_c(x+1)-q_c(x),\\
& H(u_c)=\sinh2\kappa-2\kappa.
\end{align*}
In view of the above, we have $\la r_c,1\ra=2\kappa$,
$\la p_c,1\ra=-2\kappa c$ and
\begin{equation}
  \label{eq:limtheta}
\lim_{\eps\downarrow0}(c_i\eps)^{-1}\theta_1(c_i)=12k_i,
\quad \lim_{\eps\downarrow0}\eps^2\theta_2(c_i)=-\frac{36}{k_i^2},
\quad \lim_{\eps\downarrow0}\eps^2\theta_3(c_i,c_j)=-\frac{72}{k_ik_j}.
\end{equation}
Since the Toda lattice equation satisfies (H1), its $1$-soliton
solution satisfies (P4) as well as solitary wave solutions of \eqref{eq:FPU}.
Thus we see that \eqref{eq:limtheta} holds for \eqref{eq:FPU}
with nonlinearity satisfying (H1) and that  $B_1(c_i)$, $B_1(c_i)^{-1}$ and
$B_2(c_i,c_j)$ are uniformly bounded in $\eps\in(0,\eps_0)$.
\end{proof}
\begin{proof}[Proof of Lemma \ref{lem:FPUprb}]
By the definition of $P_k(t)$ and Cramer's rule,
\begin{equation}
  \label{eq:projformula}
  \begin{split}
& P_k(t)f=(\eps^3\pd_cu_{c_j},\pd_xu_{c_j})_{j=N+1-k,\cdots,N\rightarrow}
\mathcal{A}_k^{-1}
\begin{pmatrix}  \eps^{-4}\la f,J^{-1}\pd_xu_{c_i}\ra
\\ \eps^{-1}\la f,J^{-1}\pd_cu_{c_i}\ra\end{pmatrix}_{i=N+1-k,\cdots,N\downarrow}
\\=& 
 \frac{1}{|\mathcal{A}_k|}\sum_{j=1}^N \left\{
    \begin{vmatrix}
\mathcal{A}_{11} & \ldots &\Delta_{1j}^1& \ldots& \mathcal{A}_{1k}\\
\vdots& & \vdots & & \vdots\\
\mathcal{A}_{k1} & \ldots &\Delta_{kj}^1& \ldots& \mathcal{A}_{kk}
    \end{vmatrix}
+    \begin{vmatrix}
\mathcal{A}_{11} & \ldots &\Delta_{1j}^2& \ldots& \mathcal{A}_{1k}\\
\vdots& & \vdots & & \vdots\\
\mathcal{A}_{k1} & \ldots &\Delta_{kj}^2& \ldots& \mathcal{A}_{kk}
    \end{vmatrix}\right\},.    
  \end{split}
\end{equation}
where
\begin{align*}
& \Delta_{ij}^1=
\begin{pmatrix} \eps^{-1} \la f,J^{-1}\pd_xu_{c_i}\ra \pd_cu_{c_j}
& \eps^{-4}\la \pd_xu_{c_j},J^{-1}\pd_xu_{c_i}\ra \\
\eps^2 \la f,J^{-1}\pd_cu_{c_i}\ra\pd_cu_{c_j} &
\eps^{-1}\la \pd_xu_{c_j},J^{-1}\pd_cu_{c_i}\ra  
\end{pmatrix},
\\ &
\Delta_{ij}^2=
\begin{pmatrix} \eps^{-1}\la \pd_cu_{c_j},J^{-1}\pd_xu_{c_i}\ra 
&  \eps^{-4}\la f,J^{-1}\pd_xu_{c_i}\ra \pd_xu_{c_j} \\ 
\eps^2\la \pd_cu_{c_j},J^{-1}\pd_cu_{c_i}\ra
&  \eps^{-1} \la f,J^{-1}\pd_cu_{c_i}\ra\pd_xu_{c_j}
\end{pmatrix}.
\end{align*}
We have
\begin{align*}
& \|\text{the first column of }\Delta_{ij}^1\|_{l^2_{k_1\eps}}
 +\|\text{the second column of }\Delta_{ij}^2\|_{l^2_{k_1\eps}}
\\ \lesssim & \eps^{-4}
(\|\pd_xu_{c_j}\|_{l^2_{k_1\eps}}+\eps^3\|\pd_cu_{c_j}\|_{l^2_{k_1\eps}})
(\|J^{-1}\pd_xu_{c_j}\|_{l^2_{k_1\eps}}+\eps^3\|J^{-1}\pd_cu_{c_j}\|_{l^2_{k_1\eps}})
\|f\|_{l^2_{k_1\eps}}
\\ \lesssim & e^{k_1\eps(x_j-x_i)}\|f\|_{l^2_{k_1\eps}}.
\end{align*}
On the other hand, for $m=2i-1$, $2i$,  and $n=2j-1$, $2j$,
the $(m,n)$ cofactor of $\mathcal{A}_N$ decays as $e^{-k_1\eps(x_j-x_i)}$
if $i\le j$. Indeed, since the components of $\mathcal{A}_{i',j'}$  decays
as $e^{-k_1\eps|x_{i'}-x_{j'}|}$ if $i'\ge j'$, the $(m,n)$ cofactor of
$\mathcal{A}_N$ decays as 
$$\max_{\tau\in \mathfrak{S}}
\prod_{[(\tau(k)+1)/2]>[(k+1)/2]}\exp\left(-
(k_1\eps(x_{[(\tau(k)+1)/2]}-x_{[(k+1)/2]})\right)\le e^{-k_1\eps(x_i-x_j)},
$$
where $\mathfrak{S}$ is a set of all permutations from
$\{1,\cdots,m-1,m+1,\cdots,2N\}$ to $\{1,\cdots,n-1,n+1,\cdots,2N\}$.
Thus we conclude that $P_k(t)$ is uniformly bounded in $l^2_{k_1\eps}$.
We see that $\|P_kJ\|_{B(l^2)}=O(\eps)$ follows immediately from
\eqref{eq:projformula} and Claim \ref{cl:ucsize}.

\end{proof}
To prove Lemma \ref{lem:modulation}, we start with the following:
\begin{lemma}
  \label{lem:modulationpre}
Let $u(t)$, $v_1(t)$, $c_i(t)$ and $x_i(t)$ $(i=1,\cdots,N)$ be as in
Lemma \ref{lem:modulation}.
Then for $t\in[0,T]$,
\begin{equation}
\label{eq:modeq2}
\begin{split}
& \sum_{i=1}^N(\eps^{-3}|\dot{c}_i|+|\dot{x}_i-c_i|)
\\ \lesssim & \eps^{\frac12}\left(\|v_1\|_{W(t)}
+\sum_{k=1}^N\|v_{2k}\|_{W(t)}\right)
+\eps^2e^{-k_1(\sigma\eps^3t+L)},
\end{split}
\end{equation}
\begin{equation}
\label{eq:abbound2}
  \begin{split}
& \|l_k\|_{l^2}+\|l_k\|_{X_k(t)}
\\ \lesssim & \|v_{2k}(t)\|_{X_k(t)}\left\{\eps^{\frac32}
\left(\|v_1(t)\|_{W(t)}+\sum_{1\le k\le N}\|v_{2k}(t)\|_{W(t)}\right)
+\eps^{3}e^{-k_1(\sigma\eps^3t+L)}\right\}.
  \end{split}
\end{equation}
\end{lemma}
\begin{proof}
Differentiating \eqref{eq:orthv2k} for $k=N$ with respect to $t$ and
 substituting
\eqref{eq:v2N} and \eqref{eq:secularmode} into the resulting equation, we have
\begin{equation}
  \label{eq:modeqa}
  \begin{split}
&  \frac{d}{dt}\la v_{2N},J^{-1}\pd_xu_{c_j}(\cdot-x_j(t))\ra
\\=& \la \pd_tv_{2N}, J^{-1}\pd_xu_{c_j}\ra
-\dot{x}_j\la v_{2N},J^{-1}\pd_x^2u_{c_j}\ra
+\dot{c}_j\la v_{2N},J^{-1}\pd_c\pd_xu_{c_j}\ra
\\=&
\la l-\sum_{1\le k\le N-1}l_k, J^{-1}\pd_xu_{c_j}\ra
-\la v_{2N},(H''(U_N)-H''(u_{c_j}))\pd_xu_{c_j}\ra \\ &
-(\dot{x}_j-c_j)\la v_{2N},J^{-1}\pd_x^2u_{c_j}\ra
+\dot{c}_j\la v_{2N},J^{-1}\pd_c\pd_xu_{c_j}\ra
+\sum_{k~1}^N\la P_kJR_k,\pd_xu_{c_j}\ra
\\=&0,    
  \end{split}
\end{equation}
and
\begin{equation}
  \label{eq:modeqb}
  \begin{split}
&  \frac{d}{dt}\la v_{2N},J^{-1}\pd_cu_{c_j}(\cdot-x_j(t))\ra
\\=& \la \pd_tv_{2N}, J^{-1}\pd_cu_{c_j}\ra
 -\dot{x}_j\la v_{2N},J^{-1}\pd_c\pd_xu_{c_j}\ra
+\dot{c}_j\la v_{2N},J^{-1}\pd_c^2u_{c_j}\ra
\\=&
\la l-\sum_{1\le k\le N-1}l_k, J^{-1}\pd_cu_{c_j}\ra
-\la v_{2N},(H''(U_N)-H''(u_{c_j}))\pd_cu_{c_j}\ra\\ &
-(\dot{x}_j-c_j)\la v_{2N},J^{-1}\pd_c\pd_xu_{c_j}\ra
+\dot{c}_j\la v_{2N},J^{-1}\pd_c^2u_{c_j}\ra
+\sum_{k~1}^N\la P_kJR_k,\pd_cu_{c_j}\ra\\=&0.
  \end{split}
\end{equation}
By \eqref{eq:modeqa}, \eqref{eq:modeqb} and \eqref{eq:projformula},
\begin{equation}
  \label{eq:modeq1}
\begin{split}
& (\mathcal{A}_N-\delta\mathcal{A})
\begin{pmatrix}\eps^{-3}\dot{c}_i\\ c_i-\dot{x}_i
\end{pmatrix}_{i=1,\cdots,N\downarrow}
+  \sum_{1\le k\le N-1}
\widetilde{\mathcal{A}}_k\begin{pmatrix}\eps^{-3}\alpha_{j,k}\\ \beta_{j,k}
\end{pmatrix}_{j=N+1-k,\cdots,N\downarrow}
\\ & +\wR_1+\wR_2=0,
\end{split}  
\end{equation}
where
\begin{align*}
& \widetilde{\mathcal{A}}_k=
(\mathcal{A}_{i,j})_{\substack{1\le i \le N\downarrow\\
 N+1-k\le j\le N\rightarrow}},\quad 
\delta\mathcal{A}=\diag(\delta\mathcal{A}_i)_{1\le i\le N},
\\ &
\delta\mathcal{A}_i=
\begin{pmatrix} \eps^{-1}\la v_{2N} ,J^{-1}\pd_c\pd_xu_{c_i}\ra
& \eps^{-4}\la v_{2N}, J^{-1}\pd_x^2u_{c_i}\ra \\ 
\eps^2\la v_{2N}, J^{-1}\pd_c^2u_{c_i}\ra
& \eps^{-1}\la v_{2N}, J^{-1}\pd_c\pd_xu_{c_i}\ra \end{pmatrix},
\\ &
\wR_1=\sum_{k=1}^N\widetilde{\mathcal{A}}_k\mathcal{A}_k^{-1}
\begin{pmatrix}
\eps^{-4}\la R_k, \pd_xu_{c_i}\ra
 \\
\eps^{-1}\la R_k, \pd_cu_{c_i}\ra
\end{pmatrix}_{i=N+1-k,\cdots,N\downarrow},
\\ & \wR_2= 
\begin{pmatrix} \eps^{-4}\la v_{2N},(H''(U_N)-H''(u_{c_i}))\pd_xu_{c_i}\ra  \\
\eps^{-1}\la v_{2N},(H''(U_N)-H''(u_{c_i}))\pd_cu_{c_i}\ra
\end{pmatrix}_{1\le i\le N\downarrow}.
\end{align*}
Since $\|J^{-1}\|_{l^2_{-k_1\eps}}=O(\eps^{-1})$ and
$x_i(t)\ge x_1(t)$ for any $i\ge1$,
\begin{align*}
& |\delta\mathcal{A}_i| \\ \lesssim & \|v_{2N}(t)\|_{X_N(t)}
e^{k_1\eps x_1(t)}(\eps^{-4}\|\pd_x^2u_{c_i}\|_{l^2_{-k_1\eps}}
+\eps^{-1}\|\pd_x\pd_cu_{c_i}\|_{l^2_{-k_1\eps}}
+\eps^2\|\pd_c^2u_{c_i}\|_{l^2_{-k_1\eps}})
\\ \lesssim & \eps^{-\frac32}\|v_{2N}(t)\|_{X_N(t)}
\end{align*}
follows from Claim \ref{cl:ucsize}. 

Let $R_k=R_{k1}+R_{k2}+R_{k3}$ and
\begin{align*}
& R_{k1}=H'(U_k+w_k)-H'(U_k+w_{k-1})-H''(U_k)v_{2k},\\
& R_{k2}=H'(U_k)-H'(U_{k-1})-H'(u_{c_{N+1-k}}),\\
& R_{k3}=H'(U_k+w_{k-1})-H'(U_{k-1}+w_{k-1})-H'(U_k)+H'(U_{k-1}).
\end{align*}
Then by the mean value theorem,
\begin{equation}
  \label{eq:Rk}
  \begin{split}
& |R_{k1}|\lesssim (|w_{k-1}|+|v_{2k}|)|v_{2k}|,\quad
|R_{k2}|\lesssim |u_{c_{N+1-k}}||U_{k-1}|,\\
& |R_{k3}|\lesssim  |u_{c_{N+1-k}}||w_{k-1}|.
  \end{split}
\end{equation}
It follows from Claim \ref{cl:ucsize} that
 \begin{equation}
   \label{eq:nonest1}
\begin{split}
 |\la R_{k1}, \pd_xu_{c_i}\ra|  \lesssim & 
\eps^3\left(\|v_1\|_{W(t)}+\sum_{i=1}^k\|v_{2i}\|_{W(t)}\right)\|v_{2k}\|_{W(t)},
\\ 
|\la R_{k1}, \pd_cu_{c_i}\ra| \lesssim &
\left(\|v_1\|_{W(t)}+\sum_{i=1}^k\|v_{2i}\|_{W(t)}\right)\|v_{2k}\|_{W(t)}.
\end{split}   
 \end{equation}
By Claims \ref{cl:ucsize} and \ref{cl:intsize},
\begin{align}
  \label{eq:nonest2}
& |\la R_{k2},  \pd_xu_{c_i}\ra|+\eps^3|\la R_{k2},  \pd_cu_{c_i}\ra|
\lesssim \eps^6e^{-k_1(\sigma\eps^3t+L)},\\
  \label{eq:nonest3}
& |\la R_{k3}, \pd_xu_{c_i}\ra|+\eps^3|\la R_{k3}, \pd_cu_{c_i}\ra|
\lesssim \eps^{\frac92}\|w_{k-1}(t)\|_{W(t)}.
\end{align}
Thus we have
\begin{equation}
  \label{eq:nonest4}
|\wR_1|\lesssim \eps^{\frac12}\|v_1(t)\|_{W(t)}+\sum_{k=1}^N\|v_{2k}\|_{W(t)}
+\eps^2e^{-k_1(\sigma\eps^3t+L)}.  
\end{equation}
By Claims \ref{cl:ucsize}, \ref{cl:intsize} and \ref{cl:4}, 
\begin{equation}
  \label{eq:nonest5}
\wR_2= O(\eps^{\frac12}\|v_{2N}\|_{W(t)}e^{-k_1(\sigma\eps^3t+L)}).  
\end{equation}
In view of the definition of $F_{j,k}$,
\begin{equation}
  \label{eq:abbound3}
|F_{j,k}|\lesssim \eps^{-\frac32}e^{k_1\eps(x_{N+1-k}-x_j)}\|v_{2k}\|_{X_k(t)}
(\eps^2e^{-k_{N+k-1}(\sigma\eps^3t+L)}+\eps^{-3}|\dot{c}_j|+|\dot{x}_j-c_j|),
\end{equation}
and it follows from \eqref{eq:orthv2k3}, \eqref{eq:abbound3},
Lemma \ref{lem:ANbd} and its proof that
 \begin{equation}
   \label{eq:abbound}
   \begin{split}
& \sum_{j=N+1-k}^N e^{k_1\eps (x_j-x_{N+1-k})}(\eps^{-3}|\alpha_{j,k}|+|\beta_{j,k}|)
\\ \lesssim & \eps^{-\frac32}\|v_{2k}\|_{X_k(t)}
\{\eps^2e^{-k_{N+k-1}(\sigma\eps^3t+L)}+\sum_{j=N+1-k}^N
(\eps^{-3}|\dot{c}_j|+|\dot{x}_j-c_j|)\}.
   \end{split}
 \end{equation}
Combining \eqref{eq:modeq1}, \eqref{eq:nonest4}, \eqref{eq:nonest5} and
\eqref{eq:abbound}, we obtain \eqref{eq:modeq2}.
Moreover, since
$$\widetilde{\mathcal{A}}_k\mathcal{A}_k^{-1}=E_k+O(e^{-k_1(\sigma\eps^3t+L)}),
\quad E_k=(\delta_{i+k-N,j})_{\substack{i=1,\cdots,N\downarrow\\ j=1,\cdots,k\rightarrow}},$$
we have
\begin{equation}
  \label{eq:modeq1b}
  \begin{split}
& \left\{\mathcal{A}_N +O\left(\sum_{1\le k\le N}
\eps^{-\frac32}\|v_{2k}\|_{X_k(t)}\right)\right\}
\begin{pmatrix}\eps^{-3}\dot{c}_i\\ c_i-\dot{x}_i
\end{pmatrix}_{1\le i\le N\downarrow}
\\=& 
\sum_{k=1}^N E_k
\begin{pmatrix}
\eps^{-4}\la R_{k3}, \pd_xu_{c_i}\ra\\\eps^{-1}\la R_{k3}, \pd_cu_{c_i}\ra
\end{pmatrix}_{i=N+1-k,\cdots,N\downarrow}
\\ &+O\left(\eps^{-1}\left(\|v_1\|_{W(t)}^2+\sum_{1\le k\le N}\|v_{2k}\|_{X_k(t)}^2
+\eps^{3}e^{-k_1(\sigma\eps^3t+L)}\right)\right).
  \end{split}
\end{equation}
Substituting \eqref{eq:modeq2} into \eqref{eq:abbound},
we have \eqref{eq:abbound2}. Thus we complete the proof.
\end{proof}

The right hand side of \eqref{eq:modeq1b}is not necessarily integrable in time.
We will use normal form method to retrieve bad parts from this term
to prove convergence of speed parameters $c_i(t)$ $(1\le i\le N)$ as
$t\to\infty$.
\begin{proof}[Proof of Lemma \ref{lem:modulation}]
By Claim \ref{cl:4},
\begin{equation}
  \label{eq:huk-I}
  \begin{split}
R_{k3}=& (H''(U_k)-H''(U_{k-1}))w_{k-1}+O(w_{k-1}^2)
\\=& (H''(u_{c_{N+1-k}})-I)w_{k-1}+\sum_{\substack{N+1-k\le i,j\le N\\ i\ne j}}
O(|w_{k-1}|(|u_{c_i}||u_{c_j}|+|w_{k-1}|)).
  \end{split}
\end{equation}
Thus we have
\begin{equation}
  \label{eq:badterm1}
  \begin{split}
\la R_{k3}, \pd_xu_{c_N+1-k}\ra=& \la w_{k-1},(H''(u_{c_{N+1-k}})-I)\pd_xu_{c_{N+1-k}}\ra
\\ &+O(\eps^3\|w_{k-1}\|_{W(t)}
(\|w_{k-1}\|_{W(t)}+\eps^{\frac{3}{2}}e^{-k_1(\sigma\eps^3t+L)})).
  \end{split}
\end{equation}
and for $i\ne N+1-k$,
\begin{equation}
\label{eq:badterm2}
\la R_{k3}, \pd_xu_{c_i}\ra=O(\eps^3\|w_{k-1}\|_{W(t)}
(\|w_{k-1}\|_{W(t)}+\eps^{\frac{3}{2}}e^{-k_1(\sigma\eps^3t+L)})).
\end{equation}
\par

By \eqref{eq:v1},
\begin{equation}
\label{eq:comp1}
  \begin{split}
&  \frac{d}{dt}\la v_1, \rho_{c_i(t)}(\cdot-x_i(t))\ra
\\=& \la JH'(v_1),\rho_{c_i}\ra-\dot{x}_i\la v_1,\pd_x\rho_{c_i}\ra
+\dot{c}_i\la v_1,\pd_c\rho_{c_i}\ra
\\=&
-\la v_1,(c_i\pd_x+J)\rho_{c_i}\ra+\mathcal{R}_4,
  \end{split}
\end{equation}
where
$$\mathcal{R}_4=\la J(H'(v_1)-v_1),\rho_{c_i}\ra +
\dot{c}_i\la v_1,\pd_c\rho_{c_i}\ra
-(\dot{x}_i-c_i)\la v_1,\pd_x\rho_{c_i}\ra.$$
For $i\le N-k$, it follows from \eqref{eq:v2k} that
\begin{equation}
  \label{eq:comp2}
    \begin{split}
&  \frac{d}{dt}\la v_{2k}, \rho_{c_i(t)}(\cdot-x_i(t))\ra
\\=& \la JH''(U_k)v_{2k}+l_k+Q_kJR_k,\rho_{c_i}\ra
-\dot{x}_i\la v_{2k},\pd_x\rho_{c_i}\ra
+\dot{c}_i\la v_{2k},\pd_c\rho_{c_i}\ra
\\=&  -\la v_{2k},(c_i\pd_x+J)\rho_{c_i}\ra+\mathcal{R}_5,
  \end{split}
\end{equation}
where 
\begin{align*}
\mathcal{R}_5=& \la l_k,\rho_{c_i}\ra+\dot{c}_i\la v_{2k},\pd_c\rho_{c_i}\ra
-(\dot{x}_i-c_i)\la v_{2k},\pd_x\rho_{c_i}\ra
\\ & -\la v_{2k},(H''(U_k)-I)J\rho_{c_i}\ra+\la Q_kJR_k,\rho_{c_i}\ra.  
\end{align*}
By Claim \ref{cl:3}, we have  $\rho_{c_i}\in l^2_a\cap l^2_{-a}$ 
for any $a\in(0,2k_1\eps)$ and
\begin{equation}
  \label{eq:comp3'}
  \begin{split}
|\mathcal{R}_4|\lesssim & \eps^{\frac52}
(|\dot{x}_i-c_i|+\eps^{-3}|\dot{c}_i|)\|v_1\|_{W(t)}+O(\eps^{3}\|v_1\|_{W(t)}^2)
\\ \lesssim & (\|v_1\|_{W(t)}+\sum_{1\le i\le k-1}\|v_{2i}\|_{W(t)})^2
+\eps^6e^{-k_1(\sigma\eps^3t+L)}.    
  \end{split} 
\end{equation}
Let $\|u\|_{W(t)^*}=\min_{1\le i\le N}\|e^{k_1\eps|\cdot-x_i(t)|}u\|_{l^2}$.
By Claims \ref{cl:ucsize} and \ref{cl:intsize},
\begin{align*}
|\la v_{2k},(H''(U_k)-I)J\rho_{c_i}\ra| \le &
 \|v_{2k}\|_{W(t)}\|(H''(U_k)-I)J\rho_{c_i}(\cdot-x_i(t))\|_{W(t)^*}
\\ \le & \eps^{\frac92}e^{-k_1\eps(x_{N+1-k}-x_i)}\|v_{2k}\|_{W(t)}
\\ \le & \eps^{\frac92}e^{-k_1\eps(\sigma^3t+L)}\|v_{2k}\|_{W(t)}.
\end{align*}
By Claim \ref{cl:3} and \eqref{eq:abbound},
\begin{align*}
  |\la l_k,\rho_{c_i}\ra| \le & \sum_{N+1-k\le j\le N}
\left|\alpha_{j,k}\la \pd_cu_{c_j},\rho_{c_i}\ra+
\beta_{j,k}\la \pd_xu_{c_j},\rho_{c_i}\ra\right|
\\ \lesssim & \sum_{N+1-k\le j\le N}(\eps|\alpha_{j,k}|+\eps^4|\beta_{j,k}|)
e^{-k_1(\sigma\eps^3t+L)}
\\ \lesssim & \sum_{N+1-k\le j\le N}
\eps^{\frac52}e^{-k_1(\sigma\eps^3t+L)}\|v_{2k}\|_{X_k(t)}
\left(\eps^2+|\dot{x}_j-c_j|+\eps^{-3}|\dot{c}_j|\right).
\end{align*}
By \eqref{eq:Rk} and Claim \ref{cl:3},
\begin{align*}
|\la Q_kJR_k,\rho_{c_i}\ra|\lesssim & 
\eps^3\|v_{2k}\|_{W(t)}(\|v_1\|_{W(t)}
+\sum_{1\le i\le k}\|v_{2i}\|_{W(t)})+ \eps^6e^{-k_1(\sigma\eps^3t+L)}
\\ & +\eps^{\frac92}e^{-k_1(\sigma\eps^3t+L)}
(\|v_1\|_{W(t)}+\sum_{1\le i\le k-1}\|v_{2i}\|_{W(t)}).    
\end{align*}
Combining the above with Lemma \ref{lem:modulationpre} and Claim \ref{cl:3},
we have
\begin{equation}
  \label{eq:comp3}
|\mathcal{R}_5| \lesssim 
(\|v_1\|_{W(t)}+\sum_{1\le i\le k-1}\|v_{2i}\|_{W(t)})^2+\eps^6e^{-k_1(\sigma\eps^3t+L)}.
\end{equation}

In view of Lemma \ref{lem:modulationpre} and
\eqref{eq:badterm1}--\eqref{eq:comp3},
\begin{equation}
  \label{eq:modeq3}
  \begin{split}
& \left| \la w_{k-1},(H''(u_{c_{N+1-k}})-I)\pd_xu_{c_{N+1-k}}\ra
+\frac{d}{dt}\la w_{k-1},\rho_{c_{N+1-k}}(\cdot-x_{N+1-k})\ra\right|
\\ \lesssim &
\eps^3
\left(\|v_1\|_{W(t)}+\sum_{k=1}^N\|v_{2k}\|_{X_k(t)\cap W(t)}\right)^2
+\eps^6e^{-k_1(\sigma\eps^3t+L)}.    
  \end{split}
\end{equation}
Since $B_1(c_i)$ and $B_2(c_i,c_j)$ $(1\le i,j\le N)$ are lower
triangular matrices, it follows from Lemma \ref{lem:modulationpre},
\eqref{eq:modeq2} and \eqref{eq:modeq3} that
\begin{equation}
\label{eq:normalformofc}
\mathcal{B}\frac{d\mathbf{c}}{dt}+\frac{d}{dt}\mathcal{R}_6
=\mathcal{R}_7,
\end{equation}
where $\mathbf{c}(t)={}^t(c_1(t),\cdots,c_N(t))$,
\begin{align*}
& \mathcal{B}(t)=
\diag\left(-\frac{\theta_1(c_i(t))}{c_i(t)}\right)_{1\le i\le N},\quad  
\mathcal{R}_6=\begin{pmatrix}\la w_{N-i},\rho_{c_i}\ra
\end{pmatrix}_{i=1,\cdots,N\downarrow},
\\ & \mathcal{R}_7=
O\left(\eps^3\left(\|v_1\|_{W(t)}+\sum_{k=1}^N\|v_{2k}\|_{X_k(t)\cap W(t)}\right)^2
+\eps^6e^{-k_1(\sigma\eps^3t+L)}\right).
\end{align*}
Thus we have
\begin{equation}
  \label{eq:modeq4}
\frac{d}{dt}\left(\mathbf{c}+\mathcal{B}^{-1}\mathcal{R}_6\right)
=\mathcal{B}^{-1}\mathcal{R}_7
+\left(\tfrac{d}{dt}(\mathcal{B})^{-1}\right)\mathcal{R}_6.
\end{equation}

By \eqref{eq:limtheta}, \eqref{eq:modeq2} and the definition of
$\mathcal{B}$, we have
$|\mathcal{B}^{-1}|+|\pd_{c_i}\mathcal{B}|=O(\eps^{-1})$ and
\begin{align*}
|\dot{\mathcal B}|\le &
\sum_{1\le i\le N}|\pd_{c_i}\mathcal{B}||\dot{c}_i|
\\ \lesssim & 
\eps^{\frac52}\left(\|v_1\|_{W(t)}+\sum_{1\le k\le N}\|v_{2k}\|_{W(t)}\right)
+\eps^4e^{-k_1(\sigma\eps^3t+L)}.
\end{align*}
Since $|\mathcal{R}_6|\lesssim  \eps^{\frac32}
(\|v_1\|_{W(t)}+\sum_{1\le k\le N}\|v_{2k}\|_{W(t)})$ by Claim \ref{cl:3},
\begin{align*}
& \left(\frac{d}{dt}(\mathcal{B})^{-1}\right)\mathcal{R}_6
\lesssim  \eps^{-2}|\dot{\mathcal{B}}||\mathcal{R}_6|
\\ \lesssim & 
\eps^2\left(\|v_1\|_{W(t)}
+\sum_{1\le k\le N}\|v_{2k}\|_{W(t)}\right)^2
+\eps^5e^{-2k_1(\sigma\eps^3t+L)}.
\end{align*}
Combining the above with \eqref{eq:modeq4}, we obtain \eqref{eq:modeqc}.
Thus we complete the proof.
\end{proof}
\bigskip

\section{Energy identities and virial identities}
\label{sec:energy-virial}
First, we will estimate energy norm of $v(t)$ and $v_{2k}(t)$ by adopting
an argument of \cite{FP2} that uses the convexity of Hamiltonian
and the orthogonality condition \eqref{eq:orth1}.
\begin{lemma}
  \label{lem:speed-Hamiltonian}
Let $u(t)$ be a solution to \eqref{eq:FPU} satisfying
$u(0)=\sum_{1\le i\le N}u_{c_{i,0}}(\cdot-x_{0,i})+v_0$ and let
$c_{i,0}$ and $x_{i,0}$ be as in Theorem \ref{thm:1}.
Then there exist positive numbers $\eps_0$, $\delta$, $L_0$ and
$C$ satisfying the following: Suppose $\eps\in(0,\eps_0)$,  that $v_{2k}$
$(1\le k\le N)$ satisfy \eqref{eq:orthv2k} for $N+1-k\le i\le N$
and $t\in[0,T]$, and that
\begin{gather*}
 \sup_{t\in[0,T]}\left\{\eps^{-2}|c_i(t)-c_{i,0}|
+\sum_{k=1}^N\eps^{-\frac32}\|v_{2k}(t)\|_{l^2}\right\}\le \delta,
\\ L=\inf_{t\in[0,T]}\min_{1\le i\le N-1}\eps(x_{i+1}(t)-x_i(t))\ge L_0.
 \end{gather*}
Then for $t\in[0,T]$,
\begin{equation}
  \label{eq:v1norm}
\|v_1(t)\|_{l^2}\le C\|v_0\|_{l^2},  
\end{equation}
\begin{equation}
  \label{eq:enorm}
\|v(t)\|_{l^2}^2\le C\left(\eps\sum_{i=1}^N|c_i(t)-c_0|
+\eps^{\frac32}(\|v_0\|_{l^2}+\sum_{k=1}^{N-1}\|v_{2k}\|_{W(t)})
+\|v_0\|_{l^2}^2+\eps^3e^{-k_1L}\right),
\end{equation}
\begin{equation}
  \label{eq:locenorm}
  \begin{split}
& \|v_{2k}\|_{l^2}^2 \le  C\left(\eps\sum_{i=N+1-k}^N|c_i(t)-c_0|
+\eps^{\frac32}\|v_0\|_{l^2}+\|v_0\|_{l^2}^2+\eps^3e^{-k_1L}\right)
\\ & +C\left\{\eps^{\frac32}\sum_{i=1}^{k-1}\|v_{2i}\|_{W(t)}+
\eps^3\left(\|v_1\|_{L^2(0,T;W(t))}
+\sum_{i=1}^N\|v_{2i}\|_{L^2(0,T;W(t)\cap X_k(t))}\right)^2\right\}.
  \end{split}
\end{equation}
\end{lemma}
\begin{proof}
Since $H(v_1(t))=H(v_0)$ for $t\in\R$, there exists a nondecreasing function
$C(r)$ such that $\|v_1(t)\|_{l^2}\le C(\|v_0\|_{l^2})\|v_0\|_{l^2}a$.
Thus we have \eqref{eq:v1norm}.
\par
By (P2), there exists a positive constant $C'$
independent of $\eps$ such that
  \begin{align*}
\delta H:=& H(u(t))-\sum_{1\le i\le N}H(u_{c_{i,0}})
\\ = & H(U_N(t)+v(t))-\sum_{1\le i\le N}H(u_{c_{i,0}})
\\ = & I_1+I_2+\frac12\la H''(U_N)v,v\ra+O(\|v\|_{l^2}^3)
\\ \ge & C'\|v(t)\|_{l^2}^2+I_1+I_2,
  \end{align*}
where
$I_1=\la H'(U_N),v\ra$ and $I_2=H(U_N(t))-\sum_{i=1}^NH(u_{c_{i,0}}).$
By \eqref{eq:orth1} and \eqref{eq:orthv2k},
\begin{align*}
\la H'(u_{c_i(t)}(\cdot-x_i(t))),v(t)\ra =& 
-c_i\la v(t),J^{-1}\pd_xu_{c_i(t)}(\cdot-x_i(t))\ra
\\=& -c_i\la w_{N-i}(t),J^{-1}\pd_xu_{c_i(t)}(\cdot-x_i(t))\ra.
\end{align*}
Hence it follows from Claims \ref{cl:intsize} and \ref{cl:4} that
\begin{align*}
|I_1| \le & \left|\left\la H'(U_N(t))-
\sum_{1\le i\le N}H'(u_{c_i(t)}(\cdot-x_i(t))),v\right\ra\right|
\\ & +\sum_{1\le i\le N}|c_i(t)|
\left|\la w_{N-i}(t),J^{-1}\pd_xu_{c_i(t)}(\cdot-x_i(t))\ra\right|
\\ \lesssim & \|v(t)\|_{l^2}
\left\|H'(U_N(t))-\sum_{1\le i\le N}H'(u_{c_i(t)}(\cdot-x_i(t)))
\right\|_{l^2}
\\ & +\eps^{\frac32}\left(
\|v_1(t)\|_{W(t)}+\sum_{i=1}^{N-1}\|v_{2k}(t)\|_{W(t)}\right)
\\ \lesssim & 
\eps^{\frac72}e^{-k_1(\sigma\eps^3t+L)}\|v(t)\|_{l^2}
+\eps^{\frac32}\left(\|v_1(t)\|_{W(t)}+\sum_{i=1}^{N-1}\|v_{2k}(t)\|_{W(t)}\right),
\end{align*}
\begin{align*}
  |I_2|\le &
\sum_{1\le i \le N}|H(u_{c_i(t)})-H(u_{c_{i,0}})|
+\left|H(U_N(t))-\sum_{1\le i\le N}H(u_{c_i(t)})\right|
\\ \lesssim &
\sum_{1\le i \le N}\theta_1(c_{i,0})|c_i(t)-c_{i,0}|
+\sum_{j\ne i}
\left\|u_{c_i(t)}(\cdot-x_i(t))u_{c_j(t)}(\cdot-x_j(t))\right\|_{l^1}
\\ \lesssim & 
\eps\sum_{1\le i \le N}|c_i(t)-c_{i,0}|+
\eps^3e^{-k_1(\sigma\eps^3t+L)}.
\end{align*}

Since $H(u(t))$ does not depend on $t$, we have
$|\delta H|\le \left|I_3\right|+\left|I_4\right|,$ where
\begin{align*}
I_3=& H(U_N(0)+v_0)-H(U_N(0)),
\\ I_4=& H(U_N(0))-\sum_{1\le i\le N}H(u_{c_{i,0}}(\cdot-x_{i,0})).
\end{align*}
By the assumption and Claims \ref{cl:ucsize} and \ref{cl:intsize},
\begin{align*}
  \left|I_3\right| \le & |\la H'(U_N(0)),v_0\ra|+O(\|v_0\|_{l^2}^2)
\lesssim \eps^{\frac32}\|v_0\|_{l^2}+\|v_0\|_{l^2}^2,\\
\left|I_4\right| \lesssim & \eps^3e^{-k_1L}.
\end{align*}
Combining the above, we conclude \eqref{eq:enorm}.
\par
Finally we will prove \eqref{eq:locenorm}.
By \eqref{eq:v1}, \eqref{eq:v2k} and the definitions of $U_k$,
\begin{equation}
  \label{eq:ukwk}
  \pd_t(U_k+w_k)=JH'(U_k+w_k)+\tilde{l}_k+ \sum_{i=1}^k(l_i-P_iJR_i),
\end{equation}
where $\tilde{l}_k=\sum_{i=N+1-k}^N
(\dot{c}_i\pd_cu_{c_i}-(\dot{x}_i-c_i)\pd_xu_{c_i})$.
Since  $J$ is skew-adjoint, it follows from \eqref{eq:ukwk} that
\begin{equation}
  \label{eq:ham21}
\begin{split}
\frac{d}{dt}H(U_k+w_k)= & \left\la H'(U_k+w_k),
\tilde{l}_k+\sum_{i=1}^k(l_i-P_iJR_i)\right\ra
= \sum_{i=1}^6II_i,
\end{split}
\end{equation}
where  $U_{k,int}=H'(U_k)-\sum_{i=N+1-k}^N H'(u_{c_i})$ and
\begin{align*}
II_1=& \sum_{i=1}^k\la H'(U_k+w_k),l_i\ra,\quad
II_2= -\sum_{i=1}^k\sum_{j=N+1-k}^N\la H'(u_{c_j}),P_iJR_i\ra,
\\ II_3=& -\sum_{i=1}^k\la U_{k,int}, P_iJR_i\ra,\quad
II_4= -\sum_{i=1}^k\la H'(U_k+w_k)-H'(U_k),P_iJR_i\ra,
\\ II_5=& \sum_{j=N+1-k}^N \la H'(u_{c_j}), \tilde{l}_k\ra,\quad
II_6=\la H(U_k+w_k)-\sum_{j=N+1-k}^N H(u_{c_j}),\tilde{l}_k\ra.
\end{align*}
By \eqref{eq:abbound2} and the fact that $\|H'(U_k+w_k)\|_{l^2}=O(\eps^{\frac32})$,
\begin{equation}
\label{eq:ham22}
  \begin{split}
|II_1|\lesssim & \eps^{\frac32}\sum_{i=1}^k\|l_i\|_{l^2}
\\ \lesssim &
\eps^3\left(\|v_1\|_{W(t)}+\sum_{k=1}^{N}\|v_{2k}\|_{X_k(t)\cap W(t)}\right)^2
+\eps^6e^{-k_1(\sigma\eps^3t+L)}.
  \end{split}
\end{equation}
Next, we will estimate $II_2$.
Using \eqref{eq:Rk} and the fact that $(P_iJ)^*H'(u_{c_j})=c_j\pd_xu_{c_j}$
for $j\ge N+1-i$ and $(P_iJ)^*H'(u_{c_j})=O(\eps^3e^{-k_1\eps|x_{N+1-i}-x_j|})$
for $j\le N-i$, we have
\begin{equation}
  \label{eq:ham23}
  \begin{split}
II_2=& -\sum_{i=1}^k\sum_{j=N+1-k}^N c_j\la R_{i3},\pd_xu_{c_j}\ra
\\ &+O\left(\eps^3(\|v_1\|_{W(t)}+\|w_k\|_{W(t)})^2
+\eps^6e^{-k_1(\sigma\eps^3t+L)}\right)
\\=& -\sum_{i=1}^{k-1}\sum_{j=N+1-i}^{N}c_j \la w_{i-1},
(H''(u_{c_{N+1-i}})-I)\pd_xu_{c_j}\ra
\\ & +O\left(\eps^3(\|v_1\|_{W(t)}^2+\sum_{1\le i\le k}\|v_{2i}\|_{W(t)}^2)
+\eps^6e^{k_1(\sigma\eps^3t+L)})\right)
\\=& -\sum_{j=N+1-k}^N c_j\left\la w_{N-j},(c_j\pd_x+J)\rho_{c_j}\right\ra
\\ & +O\left(\eps^3(\|v_1\|_{W(t)}^2+\sum_{1\le i\le k}\|v_{2i}\|_{W(t)}^2)
+\eps^6e^{k_1(\sigma\eps^3t+L)})\right).    
  \end{split}
\end{equation}
\par
Secondly, we will estimate $II_3$ and $II_4$.
In view of \eqref{eq:projformula}, Claim \ref{cl:ucsize} and the proof
of Lemma \ref{lem:FPUprb}, we have $\|PJ\|_{B(W(t),W(t)^*)}=O(\eps)$,
$\|PJu^2\|_{W(t)^*}\lesssim \eps^{\frac32}\|u\|_{W(t)}^2$.
Hence it follows from \eqref{eq:Rk} that
\begin{align}
  \label{eq:ham24}
 \begin{split}
|II_3| \le & \|U_{k,int}\|_{l^2}\sum_{i=1}^k\|P_iJR_i\|_{l^2}
\\ \lesssim & \eps^4e^{-k_1(\sigma\eps^3t+L)}
\left(\|v_1\|_{W(t)}+\sum_{i=1}^k\|v_{2i}\|_{W(t)}+\eps^2\right)^2,
  \end{split}
\\ \label{eq:ham25}  
\begin{split}
|II_4|\le & \|w_k(t)\|_{W(t)}\sum_{i=1}^k\|P_iJR_i\|_{W(t)^*}
\\ \lesssim & \eps^{\frac32}\|w_k\|_{W(t)}(\|v_1\|_{W(t)}+\|w_k\|_{W(t)})^2
\\ & +\|w_k\|_{W(t)}
\left\{\eps^3(\|v_1\|_{W(t)}+\sum_{i=1}^k\|v_{2i}\|_{W(t)})
+\eps^{\frac92}e^{-k_1(\sigma\eps^3t+L)}\right\}.
\end{split}
\end{align}
By \eqref{eq:secular1}, \eqref{eq:secular2} and Claim \ref{cl:intsize},
\begin{equation}
  \label{eq:ham26}
  \begin{split}
II_5=& \sum_{1\le i\le k}\left\{\theta_1(c_i)\dot{c}_i
+O(e^{-k_1(\sigma\eps^3t+L)}(\eps|\dot{c}_i|+\eps^4|\dot{x}_i-c_i|))\right\}
\\ =& \sum_{i=N+1-k}^N\theta_1(c_i)\dot{c}_i
+O(\eps^6e^{-k_1(\sigma\eps^3t+L)})
\\ & +O\left(\eps^3\left(\|v_1\|_{W(t)}^2+\sum_{i=1}^N
\|v_{2i}\|_{X_i(t)\cap W(t)}^2\right)\right).
  \end{split}
\end{equation}
By \eqref{eq:modeq2},
\begin{equation}
  \label{eq:ham27}
  \begin{split}
|II_6|\lesssim & (\|U_{k,int}\|_{l^2}+\|w_k\|_{W(t)})\|\tilde{l}_k\|_{W(t)^*}
\\ \lesssim & 
(\eps^{\frac72}e^{-k_1(\sigma\eps^3t+L)}+\|w_k\|_{W(t)})
\sum_{i=N+1-k}^N (\eps^{-\frac12}|\dot{c}_i|+\eps^{\frac52}|\dot{x}_i-c_i|)
\\ \lesssim & \eps^3(\|v_1\|_{W(t)}^2+\sum_{k=1}^N
\|v_{2i}\|_{W(t)\cap X_i(t)}^2)+\eps^6e^{-k_1(\sigma\eps^3t+L)}.    
  \end{split}
\end{equation}
Using \eqref{eq:comp1} and \eqref{eq:comp2} and following the proof of
Lemma \ref{lem:modulation}, we have
\begin{equation}
  \label{eq:ham28}
II_2+II_5= O\left(\eps^3(\|v_1\|_{W(t)}+\sum_{1\le i\le N}
\|v_{2i}\|_{W(t)\cap X_i(t)})^2+\eps^6e^{-k_1(\sigma\eps^3t+L)}\right).
\end{equation}
By \eqref{eq:ham21}, \eqref{eq:ham24}, \eqref{eq:ham25},
\eqref{eq:ham27}  and \eqref{eq:ham28},
\begin{equation}
  \label{eq:v2kenloss}
  \begin{split}
&  \left|\frac{d}{dt}H(U_k+w_k)\right|
\\ \lesssim & \eps^3\left(
\|v_1\|_{W(t)}^2+\sum_{k=1}^N(\|v_{2k}\|_{W(t)}^2+\|v_{2k}\|_{X_k(t)}^2)
\right)+\eps^6e^{-k_1(\sigma\eps^3t+L)}.
  \end{split}
\end{equation}
Integrating \eqref{eq:v2kenloss} over $[0,t]$, we obtain
\begin{equation}
  \label{eq:v2khamest}
  \begin{split}
& H(U_k(t)+w_k(t))-H(U_{k,0}+v_0)
 \\=& 
O\left(\eps^3\left(\|v_1\|_{L^2(0,T;W(t))}^2
+\sum_{i=1}^N\|v_{2i}\|_{L^2(0,T;X_i(t)\cap W(t))}^2
+e^{-k_1L}\right)\right).
  \end{split}
\end{equation}
Using the convexity of the Hamiltonian, we conclude
\begin{equation}
  \label{eq:w2knorm}
  \begin{split}
& \|w_k(t)\|_{l^2}^2 \lesssim  \eps\sum_{N+1-k\le i\le N}|c_i(t)-c_0|+
\eps^{\frac32}\|v_0\|_{l^2}+\|v_0\|_{l^2}^2+\eps^3e^{-k_1L}
\\ & +\eps^{\frac32}\sum_{i=1}^{k-1}\|v_{2i}\|_{W(t)}
 +\eps^3\left(\|v_1\|_{L^2(0,T;W(t))}
+\sum_{i=1}^N\|v_{2i}\|_{L^2(0,T;W(t)\cap X_i(t))}\right)^2
  \end{split}
\end{equation}
from \eqref{eq:v2khamest} in exactly the same way as the proof of 
\eqref{eq:enorm}.
Combining \eqref{eq:w2knorm} with \eqref{eq:v1norm} and \eqref{eq:enorm},
we obtain \eqref{eq:locenorm}.
\end{proof}

Since $v_1(t)$ is small, it moves slowly and will be decoupled from the
$N$-soliton part of the solution. The following is an analog of
\textit{virial lemma} for small solutions in Martel and Merle \cite{MM1}
and was used in \cite{Mi1} to prove orbital stability of 1-solitons of the
FPU lattice equations. Here we confirm how coefficients of the virial
identity depend on $\eps$.
\begin{lemma}
  \label{lem:monotonicity}
Let $v_1(t)$ be a solution to \eqref{eq:v1}. Let $a>0$, $\tilde{x}(t)$ be a
$C^1$-function and $\psi_a(t,x)=1+\tanh a(x-\tilde{x}(t))$.
There exist positive numbers $\eps_0$, $\delta$ and $C$ such that if
$\inf_{t\ge0}\tilde{x}_t\ge 1+k_1^2\eps^2/24$ and
$a\eps+\|v_0\|_{l^2}\le \delta\eps^2$ for an $\eps\in(0,\eps_0)$, then
$$\|\psi_a(t)^{\frac12}v_1(t)\|_{l^2}^2
+Ca\eps^2\int_0^t\|\sech a(\cdot-\tilde{x}(s)) v_1(s)\|_{l^2}^2ds \le
\|\psi_a(0)^{\frac12}v_0\|_{l^2}^2.$$
\end{lemma}
\begin{proof}
Let $v_1(t,n)={}^t\!(r_1(t,n),p_1(t,n))$,
$h_1(t,n)=\frac12p_1(t,n)^2+V(r_1(t,n))$ and
$\tpa(t,x)=a^{\frac12}\sech a(x-\tilde{x}(t))$.
By \eqref{eq:v1norm},
\begin{align*}
&\left|V(r_1(t,n))-\frac12 V'(r_1(t,n))^2\right|
\lesssim  \|v_0\|_{l^2}|r_1(t,n)|^2,
\\ & \left|V'(r_1(t,n))-r_1(t,n)\right| \lesssim \|v_0\|_{l^2}|r_1(t,n)|.
\end{align*}
Using \eqref{eq:FPU} and the above, we have
\begin{equation}
  \label{eq:locdecayv1}
  \begin{split}
&  \frac{d}{dt}\sum_{n\in\Z}\psi_a(t,n)h_1(t,n)
\\ =&
\sum_{n\in\Z}p_1(t,n)V'(r_1(t,n-1))\left(\psi_a(t,n-1)-\psi_a(t,n)\right)\
 +\sum_{n\in\Z}\pd_t\psi_a(t,n)h_1(t,n)
\\ \le &
-\frac{\tilde{x}_t(t)}2\sum_{n\in\Z}\tpa(t,n)^2p_1(t,n)^2
\\ & +(1+C'\|v_0\|_{l^2})
\sum_{n\in\Z}\left(\psi_a(t,n)-\psi_a(t,n-1)\right)|p_1(t,n)r_1(t,n-1)|
\\ &-\frac{\tilde{x}_t(t)}2(1-C'\|v_0\|_{l^2})
\sum_{n\in\Z}\tpa(t,n-1)^2r_1(t,n-1)^2,    
  \end{split}
\end{equation}
where $C'$ is a positive constant.
\par
Substituting
\begin{equation}
  \label{eq:dif-differ}
  \begin{split}
\psi_a(t,n)-\psi_a(t,n-1)=&
\sinh a\sech a(n-\tilde{x}(t))\sech a(n-\tilde{x}(t)-1)
\\=& \tpa(t,n)\tpa(t,n-1)(1+O(a^2))    
  \end{split}
\end{equation}
into \eqref{eq:locdecayv1} and using H\"older inequality, we obtain
$$
 \frac{d}{dt}\sum_{n\in\Z}\psi_a(t,n)h_1(t,n)
 \le  -\frac{\tilde{x}_t}{2}(1-C''(\|v_0\|_{l^2}+a^2))
\sum_{n\in\Z}\tpa(t,n)^2(p_1(t,n)^2+r_1(t,n)^2)$$
for a $C''>0$. Thus we have
\begin{equation*}
   \frac{d}{dt}\sum_{n\in\Z}\psi_a(t,n)h_1(t,n)
 \le  -C\eps^2
\sum_{n\in\Z}\tpa(t,n)^2(p_1(t,n)^2+r_1(t,n)^2)
\end{equation*}
for a $C>0$ if $\delta>0$ is sufficiently small.
We have thus proved Lemma \ref{lem:monotonicity}.
\end{proof}

Finally, we will prove propagation estimates on $v_{2k}$.
\begin{lemma}
  \label{lem:virialv2k}
Let $u(t)$ be as in Theorem \ref{thm:1} and let
$\pai(t,x)=1+\tanh a(x-x_i(t))$. Then there exist positive numbers $\eps_0$,
$\delta$, $L_0$ and $C$ satisfying the following:
Suppose that
\begin{gather}
\label{ass:vk}
a\eps+\sup_{t\in[0,T]}\left\{\|v_1(t)\|_{l^2}+\sum_{k=1}^N\|v_{2k}(t)\|_{l^2}
\right\}\le \delta\eps^2,
\\ \inf_{t\in[0,T]}\min_{1\le i\le N-1}\eps(x_{i+1}(t)-x_i(t))\ge L,
\notag
\\
\min_{1\le i\le N}\inf_{t\in[0,T]}\dot{x}_i(t)\ge 1+\frac{k_1^2\eps^2}{24}
\notag
 \end{gather}
for $\eps\in(0,\eps_0)$, $L\ge L_0$ and $T\ge0$.
Then for $t\in[0,T]$ and $1\le k\le N$,
\begin{align*}
& 
\|\psi_{a,1}(t)^{\frac12}v_{2k}(t)\|_{l^2}+\eps^{\frac32}\|v_{2k}(t)\|_{L^2(0,T;W(t))}
\\ \le & C\left(\|v_0\|_{l^2}
+\eps^{\frac32}\sum_{i=1}^k\|v_{2i}(t)\|_{L^2(0,T;X_k(t))}+\eps^{\frac32}e^{-k_1L}\right).
\end{align*}
\end{lemma}
\begin{proof}
In order to prove the lemma, it suffices to show that
\begin{equation}
  \label{eq:locdwform}
  \begin{split} &
\|\psi_{a,1}^{\frac12}w_k\|_{l^2}+ \eps^{\frac32}\|w_k\|_{L^2(0,T;W(t))}
\\ \lesssim & \|v_0\|_{l^2}+\eps^{\frac32}\left(\|v_{2k}\|_{L^2(0,T;X_k(t))}
+\sum_{i=1}^{k-1}\|v_{2i}\|_{L^2(0,T;W(t))}+e^{-k_1L}\right)
  \end{split}
\end{equation}
for $1\le k\le N$. Indeed, it follows from \eqref{eq:locdwform}
\begin{align*}
& \|\psi_{a,1}v_{2k}\|_{l^2}+\eps^{\frac32}\|v_{2k}\|_{L^2(0,T;W(t))}
\\ \le &
\|\psi_{a,1}w_k\|_{l^2}+\|\psi_{a,1}w_{k-1}\|_{l^2}
+\eps^{\frac32}(\|w_k\|_{L^2(0,T;W(t))}+\|w_{k-1}\|_{L^2(0,T;W(t))})
\\ \lesssim & \|v_0\|_{l^2}
+\eps^{\frac32}\left(\|v_{2k}\|_{L^2(0,T;X_k(t))}
+\|v_{2\,k-1}\|_{L^2(0,T;X_{k-1}(t))}\right)
\\ & +\eps^{\frac32}\sum_{i=1}^{k-1}\|v_{2i}\|_{L^2(0,T;W(t))}
+\eps^{\frac32}e^{-k_1L}
\\ \lesssim &  \|v_0\|_{l^2}
+\eps^{\frac32}\left(\sum_{i=1}^k\|v_{2i}\|_{L^2(0,T;X_i(t))}+e^{-k_1L}\right).
\end{align*}
\par
Let $u={}^t(r,p)$, $h(u)=\frac12p^2+V(r)$ and $h'(u)={}^t(V'(r),p)$,
$$
H_{k,i}=\la h(U_k+w_k)-h(U_k)-h'(U_k)\cdot w_k,
\pai\ra_{l^2},$$
where $\cdot$ denotes the inner product in $\R^2$. Then
\begin{align*}
  \frac{dH_{k,i}}{dt}=& 
-\dot{x}_i\la h(U_k+w_k)-h(U_k)-h'(U_k)\cdot w_k,\pai'\ra_{l^2}
\\ & +\la H'(U_k+w_k)-H'(U_k),\pai\pd_t(U_k+w_k)\ra
-\la H''(U_k)\pd_tU_k,\pai w_k\ra
\\ =:I+II.
\end{align*}
By the mean value theorem, there exists a $\theta=\theta(t,n)\in(0,1)$
such that 
$$I= -\frac{\dot{x}_i}{2}\la H''(U_k+\theta w_k)w_k,\pai'w_k\ra.$$ Since
$\|U_kw_k^2\|_{l^1}\lesssim \eps^2(\|v_{2k}\|_{X_k(t)}+\|w_{k-1}\|_{W(t)})^2$,
we have
\begin{align*}
I =-\frac{\dot{x}_i}{2}(1+O(\|w_k\|_{l^\infty}))\|\tpai w_k\|_{l^2}^2
+O(\eps^3(\|v_{2k}\|_{X_k(t)}+\|w_{k-1}\|_{W(t)})^2),
\end{align*}
where $\tpai=a^{\frac12}\sech a(x-x_i(t))$.
By \eqref{eq:ukwk} and the definition of $U_k(t)$,
we have
\begin{align*}
II=& \left\la H'(U_k+w_k)-H'(U_k), \pai JH'(U_k+w_k)
+\sum_{i=1}^k \pai(l_i-P_iJR_i)\right\ra
\\ & +\la \wR_3, \pai \tilde{l}_k\ra
 -\sum_{i=N+1-k}^N \la \pai H''(U_k)w_k,JH'(u_{c_i})\ra
\\ & = \sum_{i=1}^6II_i,
\end{align*}
where $\wR_3=H'(U_k+w_k)-H'(U_k)-H''(U_k)w_k$ and
\begin{align*}
& II_1=\la H'(U_k+w_k)-H'(U_k),\pai J(H'(U_k+w_k)-H'(U_k))\ra,\\
& II_2=\la \wR_3,\pai JH'(U_k)\ra,\quad
 II_3=\la \wR_3,\pai\tilde{l}_k \ra,\\
& II_4= \sum_{i=1}^k\la H'(U_k+w_k)-H'(U_k),\pai l_i\ra,\\
& II_5=-\sum_{i=1}^k \la H'(U_k+w_k)-H'(U_k),\pai P_iJR_i\ra,\\
& II_6=\left\la H''(U_k)w_k, \pai JU_{k,int}\right\ra.
\end{align*}
Using the Schwarz inequality and \eqref{eq:dif-differ},
we have
$$
|II_1|\le \frac{1}{2}\|\tpai(H'(U_k+w_k)-H'(U_k))\|_{l^2}^2
(1+O(a^2))$$
as in the proof of Lemma \ref{lem:monotonicity}.
Since 
\begin{align*}
& \|\tpai(H'(U_k+w_k)-H'(U_k))\|_{l^2}
\\ \le & \|\tpai w_k\|_{l^2}(1+O(\|w_k\|_{l^\infty})
+O(\|\tpai\|_{l^\infty}\|U_kw_k\|_{l^2})
\\ \le & \|\tpai w_k\|_{l^2}(1+O(\|w_k\|_{l^\infty}))
+O(\eps^3(\|v_{2k}\|_{X_k(t)}+\|w_{k-1}\|_{W(t)})),
\end{align*}
there exists a $\delta'>0$ such that
\begin{align*}
  I+II_1\le &-\frac{\dot{x}_i-1+O(\delta\eps^2)}{2}\|\tpai w_k\|_{l^2}^2
+O(\eps^3(\|v_{2k}\|_{X_k(t)}+\|w_{k-1}\|_{W(t)})^2)
\\ \le & -\delta'\eps^2\|\tpai w_k\|_{l^2}^2
+O(\eps^3(\|v_{2k}\|_{X_k(t)}+\|w_{k-1}\|_{W(t)})^2).
\end{align*}
Let
\begin{align*}
& \|u\|_{W_k(t)}=\sum_{i=N+1-k}^N\|e^{-k_1\eps|\cdot-x_i(t)|}u\|_{l^2},
\quad \|u\|_{W_k(t)^*}=\min_{i=N+1-k}^N \|e^{k_1\eps|\cdot-x_i(t)|}u\|_{l^2},
\\ & \|u\|_{\widetilde{W}_k(t)}=\sum_{i=N+1-k}^N
\|e^{-k_1\eps|\cdot-x_i(t)|}u\|_{l^1},
\quad \|u\|_{\widetilde{W}_k(t)^*}
=\min_{i=N+1-k}^N \|e^{k_1\eps|\cdot-x_i(t)|}u\|_{l^\infty}.
\end{align*}
By Claim \ref{cl:ucsize},
\begin{align*}
|II_2|\lesssim  \|w_k^2\|_{\widetilde{W}_k(t)}\sum_{i=N+1-k}^N
\|Ju_{c_i}\|_{\widetilde{W}_k(t)^*}
 \lesssim  \eps^3(\|v_{2k}\|_{X_k(t)}^2+\|w_{k-1}\|_{W(t)}^2).
\end{align*}
By \eqref{eq:modeq2}, \eqref{ass:vk} and Claim \ref{cl:ucsize},
\begin{align*}
  |II_3| \lesssim & \|w_k^2\|_{\widetilde{W}_k(t)}
\|\tilde{l}_k\|_{\widetilde{W}_k(t)^*}
\\ \lesssim & (\|v_{2k}\|_{X_k}^2+\|w_{k-1}\|_{W(t)}^2)
 \sum_{i=N+1-k}^N(|\dot{c}_i|+|\dot{x}_i-c_i|\eps^3)
\\ \lesssim & \eps^5(\|v_{2k}\|_{X_k}^2+\|w_{k-1}\|_{W(t)}^2).
\end{align*}
By \eqref{eq:abbound2},
\begin{align*}
  |II_4|\lesssim & \|w_k\|_{W_k(t)}\sum_{i=1}^k\|l_i\|_{W_k(t)^*}
\\ \lesssim & 
\eps^3(\|v_{2k}\|_{X_k(t)}+\|w_{k-1}\|_{W(t)})
\\ & \times \|v_{2k}\|_{X_k(t)}
\{e^{-k_1L}+\eps^{-\frac32}(\|v_1\|_{W(t)}+\sum_{k=1}^N\|v_{2k}\|_{W(t)})\}
\\ \lesssim & \eps^3\|v_{2k}\|_{X_k(t)}
(\|v_{2k}\|_{X_k(t)}+\|w_{k-1}\|_{W(t)}).
\end{align*}
In view of \eqref{eq:projformula} and Claim \ref{cl:ucsize}, we have
$\|PJ\|_{B(\widetilde{W}_k(t),W_k(t)^*)}=O(\eps^{\frac32})$ for $i\le k$.
Thus by \eqref{eq:Rk},
\begin{align*}
|II_5|\le & \|H'(U_k+w_k)-H'(U_k)\|_{W_k(t)}\sum_{i=1}^k \|P_iJR_i\|_{W_k(t)^*}
\\ \lesssim &
\eps^{\frac32} \|w_k(t)\|_{W_k(t)}\sum_{i=1}^k
(\|R_{i1}\|_{\widetilde{W}_k(t)}+\|R_{i2}\|_{\widetilde{W}_k(t)}+\|R_{i3}\|_{\widetilde{W}_k(t)})
\\ \lesssim & 
\eps^{\frac32}(\|v_{2k}\|_{X_k(t)}+\|w_{k-1}\|_{W(t)})^3
\\  & + \eps^{\frac92}(\|v_{2k}\|_{X_k(t)}+\|w_{k-1}\|_{W(t)})
e^{-k_1(\sigma\eps^3t+L)}
\\ & +\eps^3(\|v_{2k}\|_{X_k(t)}+\|w_{k-1}\|_{W(t)})^2
\\ \lesssim & \eps^3\left(\|v_{2k}\|_{X_k(t)}+\|w_{k-1}\|_{W(t)}\right)^2
+\eps^6e^{-2k_1(\sigma\eps^3t+L)},
\end{align*}
and
\begin{align*}
  |II_6|\lesssim & \|w_k\|_{W_k(t)}\|JU_{k,int}\|_{W_k(t)^*}
\\ \lesssim &\eps^{\frac92}e^{-k_1(\sigma\eps^3t+L)}
(\|v_{2k}\|_{X_k(t)}+\|w_{k-1}\|_{W(t)})
\\ \lesssim & \eps^3(\|v_{2k}\|_{X_k(t)}+\|w_{k-1}\|_{W(t)})^2
+\eps^6e^{-2k_1(\sigma\eps^3t+L)}
\end{align*}
as in the proof of Lemma \ref{lem:speed-Hamiltonian}.
Combining the above, we obtain
\begin{equation}
  \label{eq:v2kvirial}
  \begin{split}
& \frac{dH_{k,i}}{dt}+\delta'\eps^2\|\tpai w_k\|_{l^2}^2
\\ \lesssim & \eps^3\left(\|v_1\|_{W(t)}+\|v_{2k}\|_{X_k(t)}+\sum_{i=1}^{k-1}
\|v_{2i}\|_{W(t)}\right)^2+\eps^6e^{-2k_1(\sigma\eps^3t+L)}.    
  \end{split}
\end{equation}
Integrating \eqref{eq:v2kvirial} over $[0,T]$ and summing up for
$1\le i\le k$, we have
\begin{align*}
& \sum_{i=1}^N\left\{ H_{k,i}(t)- H_{k,i}(0)
+\eps^2\int_0^T\|\tpai(t)w_k(t)\|_{l^2}^2dt\right\}
\\ \lesssim   & 
\int_0^T\left\{\eps^3\left(\|v_{2k}\|_{X_k(t)}^2+\sum_{i=1}^{k-1}
\|v_{2i}\|_{W(t)}^2\right)+\eps^6e^{-2k_1(\sigma\eps^3t+L)}\right\}dt.
\end{align*}
Since
$H_{k,i}=\|\pai^{\frac12}w_k\|_{l^2}^2
(1+O(\|U_k\|_{l^\infty}+\|w_k\|_{l^\infty}))$, we have \eqref{eq:locdwform}.
Thus we prove Lemma \ref{lem:virialv2k}.
\end{proof}
\bigskip

\section{Proof of Theorem \ref{thm:1}}
\label{sec:prthm1}
In this section, we will show {\it \`a priori}
estimates on $v_1$, $v_{2k}$, $x_i$ and $c_i$ to prove
stability of $N$-soliton solutions.
Let
\begin{align*}
& \bM_1(T)=\eps^{-2}\sup_{t\in[0,T]}\sum_{1\le i\le N}
\left(|c_i(t)-c_{i,0}|+|\dot{x}_i(t)-c_i(t)|\right),
\\ &
\bM_2(T)=\eps^{-3}\sum_{k=1}^N\sup_{0\le t\le T}\|v_{2k}(t)\|_{l^2}^2,
\\ &
\bM_3(T)=\eps^{-\frac32}\sup_{0\le t\le T}\|v_1(t)\|_{l^2}
+\|v_1\|_{L^2(0,T;W(t))},
\\ &
\bM_4(T)=\sum_{1\le k\le N}\left(\eps^{-\frac32}\sup_{0\le t\le T}
\|\psi_{k_1\eps,1}v_{2k}(t)\|_{l^2}+\|v_{2k}\|_{L^2(0,T;W(t))}\right),
\\ &
\bM_5(T)=\sum_{1\le k\le N}\left(\eps^{-\frac32}\|v_{2k}\|_{L^\infty(0,T;X_k(t))}
+\|v_{2k}\|_{L^2(0,T;X_k(t))}\right).
\end{align*}
Lemmas \ref{lem:modulation}, \ref{lem:speed-Hamiltonian},
\ref{lem:monotonicity} and \ref{lem:virialv2k} imply a priori bound on
$\bM_i$ $(1\le i\le 4)$ by
$\|v_0\|_{H^1}$ and $\bM_5$.
\begin{lemma}
\label{lem:apb1}
There exists a positive constant $\delta$ such that if
$$\|v_0\|_{l^2}+\eps^{\frac32}\sum_{i=1}^5\bM_i(T)\le \delta\eps^{\frac52},$$
\begin{align}
\label{eq:apb11}
& \bM_1(T) \lesssim  \eps^{-\frac32}\|v_0\|_{l^2}+\bM_5(T)+e^{-k_1L}, \\
\label{eq:apb13}
& \bM_2(T) \lesssim \eps^{-\frac32}\|v_0\|_{l^2}+\bM_5(T)+e^{-k_1L},\\
\label{eq:apb14}
& \bM_3(T)\lesssim \eps^{-\frac32}\|v_0\|_{l^2},\\
\label{eq:apb15}
& \bM_4(T)\lesssim \eps^{-\frac32}\|v_0\|_{l^2}+\bM_5(T)+e^{-k_1L}.
  \end{align}
\end{lemma}
\begin{proof}
It follows from Lemma \ref{lem:modulation} that for $t\in[0,T]$,
\begin{equation}
  \label{eq:apbest11}
  \begin{split}
&  \sum_{i=1}^N|c_i(t)-c_{i,0}| \\ \le &
\sum_{i=1}^N\sum_{k=1}^{N-i}\theta_1(c_i(t))^{-1}|\la w_k(t),\rho_{c_i(t)}\ra|
\\ & + C\eps^2\int_0^t\left\{
\|v_1\|_{W(s)}^2+\sum_{k=1}^N\|v_{2k}\|_{X_k(s)\cap W(s)}^2
+\eps^3e^{-k_1(\sigma\eps^3s+L)}\right\}ds
\\ \lesssim &
\eps^{\frac12}\sum_{i=1}^N\|\psi_{k_1\eps,1}(t)^{\frac12}w_i(t)\|_{l^2}
\\ & +\eps^2\left(\|v_1\|_{L^2(0,T;W(t))}^2+\sum_{k=1}^N
\|v_{2k}\|_{L^2(0,T;X_k(t)\cap W(t))}^2+e^{-k_1L}\right)
\\ \lesssim &
\eps^2\left\{\bM_4(T)+(\bM_3(T)+\bM_4(T)+\bM_5(T))^2+e^{-k_1L}\right\},    
  \end{split}
\end{equation}
and
\begin{equation}
  \label{eq:apbest12}
  \begin{split}
 |\dot{x}_i(t)-c_i(t)|\lesssim &  \eps^{\frac12}\left(\|v_1\|_{W(t)}
+\sum_{i=1}^N\|\psi_{k_1\eps,1}(t)^{\frac12}v_{2i}(t)\|_{l^2}\right)
+\eps^2e^{-k_1L}
\\ \lesssim & 
\eps^2(\bM_3(T)+\bM_4(T)+e^{-k_1L}).    
  \end{split}
\end{equation}
Lemmas \ref{lem:speed-Hamiltonian}, \ref{lem:monotonicity}
and \ref{lem:virialv2k} imply \eqref{eq:apb14}, \eqref{eq:apb15} and
\begin{align}
  \label{eq:apbest13}
\bM_2(T)^{\frac12}\lesssim \bM_1(t)^{\frac12}+\eps^{-\frac34}\|v_0\|_{l^2}^{\frac12}
+e^{-k_1L}+\bM_3(T)+\bM_4(T)+\bM_5(T).
\end{align}
Substituting \eqref{eq:apb14} and \eqref{eq:apb15} into
\eqref{eq:apbest11}-\eqref{eq:apbest13}, we obtain
\eqref{eq:apb11} and \eqref{eq:apb13}.
Thus we prove Lemma \ref{lem:apb1}.
\end{proof}
Now we will estimate $\bM_5(T)$.
\begin{lemma}
  \label{lem:apb2}
There exists a positive constant $\delta$ such that if
$$\|v_0\|_{l^2}+\eps^{\frac32}\sum_{i=1}^5\bM_i(T)\le \delta\eps^{\frac52},$$
then $\bM_5(T)\lesssim \eps^{-\frac32}\|v_0\|_{l^2}+e^{-k_1L}.$
\end{lemma}
To prove Lemma \ref{lem:apb2}, we need the following exponential stability
result of $k$-soliton solutions $(1\le k\le N)$.
\begin{lemma}
  \label{lem:linearstability}
Let $x_{i,0}(t)=c_{i,0}t+x_{i,0}$ and 
$\widetilde{U}_k(t)=\sum_{i=N+1-k}^Nu_{c_{i,0}}(\cdot-x_{i,0}(t))$.
Let $\zeta={}^t(\zeta_1,\zeta_2)\in C^1(\R^2)$,
 $\mathcal{F}_n\zeta\in L^1(\T)$,
$F_1$, $F_2\in C([0,\infty);l^2_{k_1\eps})$ and let
$w(t)\in C^1(\R;l^2_{k_1\eps})$ be a solution of 
\begin{equation}
  \label{eq:LFPU2}
\pd_tw(t)=JH''(\widetilde{U}_k(t)+\zeta(t))w(t)+F_1(t)+JF_2(t).
\end{equation}
There exist positive numbers $\eps_0$, $L_0$, $\delta_1$, $\delta_2$, $M$ and
$b$ satisfying the following: Suppose $\eps\in(0,\eps_0)$,
$0\le T_1\le T_2\le \infty$ and that
\begin{gather*}
\inf_{t\in[T_1,T_2]}\min_{2\le j\le N}\eps(x_{j,0}-x_{j-1,0})\ge L_0,
\\ \sup_{t\in[T_1,T_2]}\sup_{x\in\R}
(|\zeta_1(t,x)|+\eps^{-1}|\pd_x\zeta_1(t,x)|)\le \delta_1\eps^2,
\end{gather*}
and
\begin{equation}
  \label{eq:orth3'}
  \begin{split}
 & \eps^{-\frac32}|\la w(t), J^{-1}\pd_xu_{c_{i,0}}(\cdot-x_{i,0}(t))\ra|
+\eps^{\frac32}|\la w(t),J^{-1}\pd_cu_{c_i}(\cdot-x_{i,0}(t))\ra|
\\ \le &  \delta_2\|e^{\eps k_1(\cdot-x_{N+1-k,0}(t))}w(t)\|_{l^2}
\end{split}
\end{equation}
for $N+1-k\le i\le N$ and  $t\in[T_1,T_2]$.
Then for every $t, t_1\in[T_1,T_2]$ satisfying $t\ge t_1$,
\begin{align*}
& \|e^{\eps k_1(\cdot-x_{N+1-k,0}(t))}w(t)\|_{l^2} \\ \le &
Me^{-b\eps^3(t-t_1)}\|e^{\eps k_1(\cdot-x_{N+1-k,0}(t_1))}w(t_1)\|_{l^2}
\\ &+M\int_{t_1}^t e^{-b\eps^3(t-s)}\|e^{\eps k_1(\cdot-x_{N+1-k,0}(s))}F_1(s)\|_{l^2}ds
\\ & + M\eps^{-\frac12} \int_{t_1}^t(t-s)^{-\frac12}
\|e^{\eps k_1(\cdot-x_{N+1-k,0}(s))}F_2(s)\|_{l^2}ds.
\end{align*}
\end{lemma}
Lemma \ref{lem:linearstability} follows immediately from
Lemma \ref{lem:LFPU-1}. See Appendix \ref{sec:lemlinearstability}.

\begin{proof}[Proof of Lemma \ref{lem:apb2}]
Let $\{t_j\}_{j\ge 0}$ be a monotone increasing sequence such that
$t_0=0$ and $\sup_{j\ge 0}[t_j,t_{j+1}]=[0,T]$ that satisfies
\eqref{eq:tj1} and \eqref{eq:tj2} below. We remark that $t_{j+1}-t_j\sim \eps^{-3}$.
\par
To begin with, we will show that Lemma \ref{lem:linearstability} is applicable
provided $\delta$ is small.
Let $x_{ij}(t):=x_i(t_j)+c_{i,0}(t-t_j)$, $h_{ij}(t)=x_i(t)-x_{ij}(t)$
and $U_{kj}(t)=\sum_{i=N+1-k}^N u_{c_{i,0}}(\cdot-x_{ij}(t))$.
Lemma \ref{lem:apb1} implies that for $t\in[t_j,t_{j+1}]$,
\begin{align*}
  |h_{ij}(t)|\le & \int_{t_j}^t\left(|\dot{x}_i(s)-c_i(s)|+|c_i(s)-c_{i,0}|
\right)ds
\\ \lesssim & \eps^2\bM_1(T)(t_{j+1}-t_j).
\end{align*}
Thus there exists an $A_2>0$ such that  for $t\in[t_j,t_{j+1}]$, 
\begin{align*}
&  \sup_x\left|U_k(t)-U_{kj}(t)\right| \\ \le & 
\sum_{i=N+1-k}^N \left(
\|\pd_xu_{c_{i,0}}\|_{L^\infty}|x_i(t)-x_{ij}(t)|
+\sup_{|c-c_{i,0}|\le \delta\eps^2}\|\pd_cu\|_{L^\infty}|c_i(t)-c_{i,0}|
\right)
\\ \le  & A_2\eps^2\bM_1(T)\{\eps^3(t_{j+1}-t_j)+1\},
\end{align*}
and
\begin{align*}
&  \sup_x\left|\pd_xU_k(t)-\pd_xU_{kj}(t)\right| \\ \le & 
\sum_{i=N+1-k}^N \left(
\|\pd_x^2u_{c_{i,0}}\|_{L^\infty}|x_i(t)-x_{ij}(t)|
+\sup_{|c-c_{i,0}|\le \delta\eps^2}\|\pd_x\pd_cu\|_{L^\infty}|c_i(t)-c_{i,0}|
\right)
\\ \le & A_2\eps^3\bM_1(T)\{\eps^3(t_{j+1}-t_j)+1\}.
\end{align*}
Suppose 
\begin{equation}
  \label{eq:tj1}
  A_2\delta\{1+\eps^3\sup_{j\ge 0}(t_{j+1}-t_j)\}<\delta_1.
\end{equation}
Since $\sup_{t\in[t_j,t_{j+1}]}\eps|x_i(t)-x_{ij}(t)|=O(\delta)$,
there exist positive constants $c_1$ and $c_2$ such that
$$ c_1\|e^{k_1\eps(\cdot-x_{k,j}(t))}u\|_{l^2}
\le \|u\|_{X_k(t)}\le c_2\|e^{k_1\eps(\cdot-x_{k,j}(t))}u\|_{l^2}$$
for every $t\in [t_j,t_{j+1}]$, $j\ge 0$ and $u\in l^2_{k_1\eps}$.
Hence it follows from Lemma \ref{lem:linearstability} that for
$t\in[t_j,t_{j+1}]$, $j\ge0$ and $1\le k\le N-1$,
\begin{equation}
  \label{eq:v2kint1}
  \begin{split}
& \|v_{2k}(t)\|_{X_k(t)}\lesssim e^{-b\eps^3(t-t_j)}\|v_{2k}(t_j)\|_{X_k(t_j)}
\\ & +\int_{t_j}^te^{-b\eps^3(t-s)}
\left(\|l_k(s)\|_{X_k(s)}+\|[Q_k(s),J]R_k\|_{X_k(s)}\right)ds
\\ & +\eps^{-\frac12}\int_{t_j}^te^{-b\eps^3(t-s)}(t-s)^{-\frac12}
\|Q_k(s)R_k\|_{X_k(s)}ds.  
\end{split}
\end{equation}
By Lemma \ref{lem:modulationpre},
  \begin{align*}
\|l_k\|_{X_k(t)}\lesssim & 
\eps^{\frac32}\|v_{2k}\|_{X_k(t)}(\|v_1\|_{l^2}+\sum_{i=1}^{N}\|v_{2i}\|_{l^2}
+\eps^{\frac32}e^{-k_1(\sigma\eps^3t+L)})
\\ \lesssim & \delta\eps^3\|v_{2k}\|_{X_k(t)}.  
\end{align*}
 By \eqref{eq:Rk},
 \begin{align*}
\|R_k\|_{X_k(t)}\lesssim & 
\|R_{k1}\|_{X_k(t)}+\|R_{k2}\|_{X_k(t)}+\|R_{k3}\|_{X_k(t)}
\\ \lesssim & 
\|v_{2k}\|_{X_k(t)}(\|v_{2k}\|_{l^2}+ \|w_{k-1}\|_{l^2})
+\eps^{\frac72}e^{-k_{N+1-k}(\sigma\eps^3t+L)}+\eps^2\|w_{k-1}\|_{W(t)}
\\ \lesssim & \delta\eps^2\|v_{2k}\|_{X_k(t)}
+\eps^{\frac72}e^{-k_{N+1-k}(\sigma\eps^3t+L)}+\eps^2\|w_{k-1}\|_{W(t)}.   
 \end{align*}
Substituting the above inequalities and $\|[Q_k(s),J]\|_{B(X_k(s))}=O(\eps)$
into \eqref{eq:v2kint1}, we have
\begin{align*}
&  \|v_{2k}(t)\|_{X_k(t)} \\ \lesssim &
e^{-b\eps^3(t-t_j)}\|v_{2k}(t_j)\|_{X_k(t_j)}
+\eps^{\frac92}\int_{t_j}^t e^{-b\eps^3(t-s)}
(1+\eps^{-\frac32}(t-s)^{-\frac12})
e^{-k_1(\sigma\eps^3s+L)}ds
\\ & + \delta\eps^3\int_{t_j}^te^{-b\eps^3(t-s)}
(1+\eps^{-\frac32}(t-s)^{-\frac12})\|v_{2k}(s)\|_{X_k(s)}ds
\\ & +\eps^3\int_{t_j}^t
(1+\eps^{-\frac32}(t-s)^{-\frac12})e^{-b\eps^3(t-s)}
\left(\|v_1(s)\|_{W(s)}+\sum_{i=1}^{k-1}\|v_{2i}(s)\|_{W(s)}\right)ds
\\ \lesssim &
e^{-2b_1\eps^3(t-t_j)}\|v_{2k}(t_j)\|_{X_k(t_j)}
+\eps^{\frac32}e^{-(k_1L+2b_1\eps^3t)}
\\ & + \eps^{\frac32}\delta\int_{t_j}^te^{-2b_1\eps^3(t-s)}
(t-s)^{-\frac12}\|v_{2k}(s)\|_{X_k(s)}ds
\\ & +\eps^{\frac32}\int_{t_j}^te^{-2b_1\eps^3(t-s)}(t-s)^{-\frac12}
\left(\|v_1(s)\|_{W(s)}+\sum_{i=1}^{k-1}\|v_{2i}(s)\|_{W(s)}\right)ds,
\end{align*}
where $b_1=\min\{\frac{b}{4},\frac{k_1\sigma}{4}\}$.
Applying Gronwall's inequality (\cite[Lemma 7.1.1]{He}) to the above,
we see that for small $\delta$, there exist positive constants
$C_1$ and  $C_2$ such that
\begin{equation}
  \label{eq:v2kint2}
  \begin{split}
& \|v_{2k}(t)\|_{X_k(t)}\le 
 C_1\{e^{-b_1\eps^3(t-t_j)}\|v_{2k}(t_j)\|_{X_k(t_j)}
+\eps^{\frac32}e^{-(b_1\eps^3t+k_1L)}\}
\\ & \quad +C_2\eps^{\frac32}\int_{t_j}^t e^{-b_1\eps^3(t-s)}(t-s)^{-\frac12}
\left(\|v_1(s)\|_{W(s)}+\sum_{i=1}^{k-1}\|v_{2i}(s)\|_{W(s)}\right)ds
  \end{split}
\end{equation}
for every $t\in [t_j,t_{j+1}]$, $j\ge0$ and $1\le k\le N-1$.
Suppose that $\{t_j\}_{j\ge0}$ satisfies
\begin{equation}
  \label{eq:tj2}
  C_1\sup_{j\ge0}e^{-b_1\eps^3(t_{j+1}-t_j)}\le \frac12.
\end{equation}
Lemma \ref{lem:virialv2k} implies 
\begin{equation}
  \label{eq:aux1}
\sup_{t\in[0,T]}\|v_{2i}\|_{W(t)}\lesssim \|v_0\|_{l^2}+\eps^{\frac32}e^{-k_1L}
+\eps^{\frac32}\sum_{j=1}^i\|v_{2j}\|_{L^2(0,T;X_j(t))}.  
\end{equation}
By  \eqref{eq:v2kint2}, \eqref{eq:aux1} and Lemma \ref{lem:monotonicity},
there exists a positive constant $C_3$ such that
\begin{align*}
& \|v_{2k}(t_{j+1})\|_{X_k(t_{j+1})} \\ \le &
\frac12(\|v_{2k}(t_j)\|_{X_k(t_j)}+\eps^{\frac32}e^{-k_1L})
\\ & +C_2\eps^3\|e^{-b_1\eps^3t}t^{-\frac12}\|_{L^1(0,T)}
\sup_{t\in[0,T]}\left(\|v_1(t)\|_{W(t)}+\sum_{i=1}^{k-1}\|v_{2i}(t)\|_{W(t)}\right)
\\ \le & \frac12\|v_{2k}(t_j)\|_{X_k(t_j)} 
+C_3\left\{\|v_0\|_{l^2}+\eps^{\frac32}
\left(e^{-k_1L}+\sum_{i=1}^{k-1}\|v_{2i}(t)\|_{L^2(0,T;X_i(t))}\right)\right\}
\end{align*}
for any $j\ge0$. Thus we have 
$$\sup_{j\ge0}\|v_{2k}(t_j)\|_{X_k(t_j)}
\lesssim 
\|v_0\|_{l^2}+\eps^{\frac32}
\left(e^{-k_1L}+\sum_{i=1}^{k-1}\|v_{2i}(t)\|_{L^2(0,T;X_i(t))}\right).$$
Substituting the above into \eqref{eq:v2kint2} and applying  Young's
inequality to the resulting equation and using Lemmas \ref{lem:monotonicity}
and \ref{lem:virialv2k} again, we have for $1\le k\le N-1$,
\begin{equation*}
  \begin{split}
& \|v_{2k}\|_{L^2(0,T;X_k(t))} \\ \lesssim &
\eps^{-\frac32}\|v_0\|_{l^2}+e^{-k_1L}+\sum_{i=1}^{k-1}\|v_{2i}\|_{L^2(0,T;X_i(t))}
\\ & +\eps^{\frac32}\|e^{-b_1\eps^3t} t^{-\frac12}\|_{L^1(0,T)}
\left(\|v_1\|_{L^2(0,T;W(t))}+\sum_{i=1}^{k-1}\|v_{2i}\|_{L^2(0,T;X_i(t))}\right)
\\ \lesssim & \eps^{-\frac32}\|v_0\|_{l^2}+e^{-k_1L}
+\sum_{i=1}^{k-1}\|v_{2i}\|_{L^2(0,T;X_k(t))}.
  \end{split}
\end{equation*}
Similarly, we have
\begin{equation*}
\sup_{t\in[0,T]} \|v_{2k}(t)\|_{X_k(t)} \lesssim \|v_0\|_{l^2}+\eps^{\frac32}
\left(e^{-k_1L}+\sum_{i=1}^{k-1}\|v_{2i}(t)\|_{L^2(0,T;X_i(t))}\right)
\end{equation*}
by using \eqref{eq:v2kint2} and \eqref{eq:aux1}.
Thus we conclude that for $1\le k \le N-1$,
\begin{equation}
    \label{eq:<n}
\sup_{t\in[0,T]} \|v_{2k}(t)\|_{X_k(t)}+\eps^{\frac32}\|v_{2k}\|_{L^2(0,T;X_k(t))}
\lesssim \|v_0\|_{l^2}+\eps^{\frac32}e^{-k_1L}.
\end{equation}
\par
Finally, we will estimate $\|v_{2N}\|_{X_N(t)}$.
Eq. \eqref{eq:v2N} is transformed into 
\begin{equation}
  \label{eq:v2N2}
\left\{
  \begin{aligned}
    &   \pd_tv_{2N}=JH''(U_N)v_{2N}+l_N+Q_NJR_N,\\ & v_{2N}(0)=0,
  \end{aligned}\right.
\end{equation}
where $l_N=P_N(t)(\pd_t-JH''(U_N(t)))v_{2N}
=-\dot{P}_N(t)v_{2N}-P_N(t)JH''(U_N(t))v_{2N}.$
Let
\begin{align*}
f_N=\begin{pmatrix}  f_{Ni}^1\\ f_{Ni}^2\end{pmatrix}_{i=1,\cdots,N\downarrow}
=& \begin{pmatrix}
\eps^{-4}\la v_{2N},(H''(U_N)-H''(u_{c_i}))\pd_xu_{c_i}\ra\\
\eps^{-1}\la v_{2N},(H''(U_N)-H''(u_{c_i}))\pd_c{u}_{c_i}\ra
\end{pmatrix}_{i=1,\cdots,N\downarrow}\\
+&
 \begin{pmatrix}
\eps^{-4}\la v_{2N},J^{-1}\{(\dot{x}_i-c_i)\pd_x^2u_{c_i}
-\dot{c}_i\pd_x\pd_cu_{c_i}\}\ra\\
\eps^{-1}\la v_{2N},J^{-1}\{(\dot{x}_i-c_i)\pd_c\pd_xu_{c_i}
-\dot{c}_i\pd_c^2u_{c_i}\}\ra
\end{pmatrix}_{i=1,\cdots,N\downarrow}.
\end{align*}
By \eqref{eq:projformula} and \eqref{eq:secularmode}, we have
\begin{align*}
  l_N=&
(\eps^3\pd_cu_{c_j},\pd_xu_{c_j})_{j=1,\cdots,N\rightarrow}\mathcal{A}_N^{-1}f_N
\\=& \frac{1}{|\mathcal{A}_N|}\sum_{j=1}^N \left\{
    \begin{vmatrix}
\mathcal{A}_{11} & \ldots &\widetilde{\Delta}_{1j}^1& \ldots& \mathcal{A}_{1N}\\
\vdots& & \vdots & & \vdots\\
\mathcal{A}_{N1} & \ldots &\widetilde{\Delta}_{Nj}^1& \ldots& \mathcal{A}_{NN}
    \end{vmatrix}
+    \begin{vmatrix}
\mathcal{A}_{11} & \ldots &\widetilde{\Delta}_{1j}^2& \ldots& \mathcal{A}_{1N}\\
\vdots& & \vdots & & \vdots\\
\mathcal{A}_{N1} & \ldots &\widetilde{\Delta}_{Nj}^2& \ldots& \mathcal{A}_{NN}
    \end{vmatrix}\right\},
\end{align*}
where
\begin{align*}
& \widetilde{\Delta}_{ij}^1=
\begin{pmatrix}
\eps^3\pd_cu_{c_j}f_{Ni}^1 &
\eps^{-4}\la \pd_xu_{c_j},J^{-1}\pd_xu_{c_i}\ra \\
\eps^3\pd_cu_{c_j}f_{Ni}^2 &
\eps^{-1}\la \pd_xu_{c_j},J^{-1}\pd_cu_{c_i}\ra  
\end{pmatrix},
\\ &
\widetilde{\Delta}_{ij}^2=
\begin{pmatrix} \eps^{-1}\la \pd_cu_{c_j},J^{-1}\pd_xu_{c_i}\ra 
& \pd_xu_{c_j}f_{Ni}^1 \\ 
\eps^2\la \pd_cu_{c_j},J^{-1}\pd_cu_{c_i}\ra
& \pd_xu_{c_j}f_{Ni}^2
\end{pmatrix}.
\end{align*}
Noting that
\begin{align*}
& \|\text{the first column of }\widetilde{\Delta}_{ij}^1\|_{X_N(t)}
+\|\text{the second column of }\widetilde{\Delta}_{ij}^2\|_{X_N(t)}
\\ \lesssim &
 \eps^3\{\eps^{-\frac32}(\|v_1\|_{W(t)}+\sum_{k=1}^N\|v_{2k}\|_{W(t)})
+e^{-k_1(\sigma\eps^3t+L)}\}e^{k_1\eps(x_j-x_i)}\|v_{2N}\|_{X_N(t)},
\end{align*}
and following the argument of the proof of Lemma \ref{lem:FPUprb},
we have
$$\|l_N(t)\|_{X_N(t)}\lesssim \eps^3(\delta+e^{-k_1L})
\|v_{2N}(t)\|_{X_N(t)}.$$
Thus we have
\begin{align*}
& \sup_{t\in[0,T]}\|v_{2N}\|_{X_N(t)}+\eps^{\frac32}\|v_{2N}\|_{L^2(0,T;X_N(t))}
\\ \lesssim & \|v_0\|_{l^2}+\eps^{\frac32}e^{-k_1L}+
\sum_{1\le k\le N-1}\|v_{2i}\|_{L^2(0,T;X_k(t))}  
\end{align*}
exactly in the same way as \eqref{eq:<n}.
\end{proof}

Now we are in position to prove Theorem \ref{thm:1}.
\begin{proof}[Proof of Theorem \ref{thm:1}]
Let $(v_1,v_{21},\cdots,v_{2N},x_1,c_1,\cdots,x_N,c_N)$ be a solution
to the system \eqref{eq:v1}, \eqref{eq:v2k}, \eqref{eq:v2N},
\eqref{eq:orthv2k3}, \eqref{eq:modeq1} satisfying the initial condition
\eqref{eq:IC}. It exists as long as $v_{2k}$ $(1\le k\le N)$ and $c_i$
remain bounded. 
Let $\delta$ be a positive number  given in Lemmas \ref{lem:apb1} and
\ref{lem:apb2}.  By \eqref{eq:modeqx} and \eqref{eq:IC},
\begin{align*}
  \|v_0\|_{l^2}+\eps^{\frac32}\sum_{i=1}^5\bM_i(0)=&
2\|v_0\|_{l^2}+\eps^{-\frac12}\sum_{i=1}|\dot{x}_i(0)-c_i(0)|
\\ \lesssim & \delta_0\eps^2+\eps^{\frac32}e^{-k_1L}.
\end{align*}
If $\delta_0$ is sufficiently small and $L$ is sufficiently large,
$$\|v_0\|_{l^2}+\eps^{\frac32}\sum_{i=1}^5\bM_i(0)
\le \frac\delta2\eps^{\frac52}.$$
Let $T_*=\sup\{T_1\ge0:\|v_0\|_{l^2}+\eps^{\frac32}\sum_{i=1}^5\bM_i(T)
\le \delta\eps^{\frac52} \text{ for $0\le T\le T_1$}\}$. 
Lemmas \ref{lem:apb1} and \ref{lem:apb2} imply that there exists a $C>0$
such that
\begin{align*}
\|v_0\|_{l^2}+\eps^{\frac32}\sum_{i=1}^5\bM_i(T)
\le & C(\|v_0\|_{l^2}+\eps^2e^{-k_1L})
\\ < & \delta\eps^{\frac52} \quad\text{for $0\le T\le T_*$}
\end{align*}
provided $\eps_0$, $\delta_0$ are sufficiently small and $L$ is sufficiently
large. Thus we have $T_*=\infty$ and \eqref{eq:orbital-stability}.
We can prove \eqref{eq:asympstability} and \eqref{eq:spphconv}
in exactly the same way as \cite[pp.140-143]{Mi1}.
Thus we complete the proof of Theorem \ref{thm:1}.
\end{proof}

\bigskip

\section{Linear estimate}
\label{sec:linear}

In this section, we prove exponential linear stability of small
$N$-soliton solutions of \eqref{eq:FPU}. Let $T=t/24$, $X=x-t$ and 
\begin{gather*}
r_{N,\eps}(t,x;\mk,\mgamma)
=\varphi_N\left(T,X;\eps\mk,\eps^{-1}\mgamma\right)
=\eps^2\varphi_N\left(\eps^3T,\eps X;\mk,\mgamma\right),
\\ \quad
u_{N,\eps}(t,n;\mk,\mgamma)=
{}^t(r_{N,\eps}(t,n;\mk,\mgamma), -r_{N,\eps}(t,n;\mk,\mgamma)).
\end{gather*}
Gardner {\it et al.}  \cite{GGKM} tells us that an $N$-soliton $u_{N,\eps}$
uniformly converges to a train of solitary waves
$u_{c_{i,\eps}}(n-c_{i,\eps}t-\eps^{-1}\tilde{\gamma}_i)$
($1\le i\le N$) as $t\to\infty$ (see also \cite{HS}).
Since solitary waves of \eqref{eq:FPU} are approximated by KdV 1-solitons
in the continuous limit (\cite{FP1}), $u_{N,\eps}$ is an approximate solution
of \eqref{eq:FPU}.
\par

The linearized equation of \eqref{eq:FPU} around $u_{N,\eps}$ has a similar exponential stability property
as the linearized KdV equation \eqref{eq:LKdV} if $\eps$ is close to $0$.
\begin{lemma}
  \label{lem:LFPU-1}
Let $0<k_1<\cdots<k_N$, $\zeta=(\zeta_1,\zeta_2)
\in C^1(\R)$, $\mathcal{F}_n\zeta\in L^1(\T)$ and
$F_1$, $F_2\in C([0,\infty);l^2_{k_1\eps})$.
Let $w(t)\in C^1(\R;l^2_{k_1\eps})$ be a solution of 
\begin{equation}
  \label{eq:LFPU}
\pd_tw(t)=JH''(u_{N,\eps}(t,\cdot;\mk,\mgamma)+\zeta(t,\cdot))w(t)+F_1(t)
+JF_2(t).
\end{equation}
There exist positive numbers $\eps_0$, $\delta_1$, $\delta_2$, $M$ and $b$
satisfying the following: If $\eps\in(0,\eps_0)$, 
$\sup_{t,x}(|\zeta_1(t,x)|+\eps^{-1}|\pd_x\zeta_1(t,x)|)
\le \delta_1\eps^2$ and
\begin{equation}
  \label{eq:orth3}
\begin{split}
& \sum_{1\le i\le N}
\left(|\la w(t), J^{-1}\pd_{\gamma_i}u_{N,\eps}(t)\ra|
+|\la w(t),J^{-1}\pd_{k_i}u_{N,\eps}(t)\ra|\right)
\\ \le & \delta_2\eps^{\frac12}\|e^{k_1\eps(\cdot-c_{1,\eps}t-\eps^{-1}\gamma_1)}w(t)\|_{l^2}
\end{split}
\end{equation}
for $1\le i\le N$ and $t\ge t_1$, then for every $t\ge t_1\ge0$,
\begin{equation*}
  \begin{split}
& \|e^{\eps k_1(\cdot-c_{1,\eps}t)}w(t)\|_{l^2} \\ \le &
Me^{-b\eps^3(t-s)}\|e^{\eps k_1(\cdot-c_{1,\eps}t_1)}w(t_1)\|_{l^2}
+M\int_{t_1}^t e^{-b\eps^3(t-s)}\|e^{\eps k_1(\cdot-c_{1,\eps}s)}F_1(s)\|_{l^2}ds
\\ & + M\eps^{-\frac12}\int_{t_1}^t e^{-b\eps^3(t-s)} (t-s)^{-\frac12}
\|e^{\eps k_1(\cdot-c_{1,\eps}s)}F_2(s)\|_{l^2})ds.
  \end{split}
\end{equation*}
\end{lemma}
\par
Let
\begin{align*}
& \widehat{J}=\begin{pmatrix} 0 & e^{i\xi}-1 \\ 1-e^{-i\xi} & 0\end{pmatrix},
\quad P(\xi)=\frac{1}{\sqrt{2}}
\begin{pmatrix}1 & e^{i\xi/2} \\ -e^{-i\xi/2} & 1 \end{pmatrix},
\\ &
f(t,\xi)=
\begin{pmatrix}  f_+(t,\xi)\\ f_-(t,\xi)\end{pmatrix}
=e^{ic_{1,\eps}t\xi}P(\xi)^{*}\mathcal{F}_nw(t,\xi),
\\ & f_{\#}(t,\xi)=e^{-ic_{1,\eps}t\xi}
(f_{+}(t,\xi)+e^{\frac{i\xi}2}f_-(t,\xi)),
\\ &
G_1(t,\xi)=\frac{e^{ic_{1,\eps}t\xi}}{\sqrt{2\pi}}
\left((\widetilde{r_{N,\eps}}(t,\xi;\mk,\mgamma)+\widetilde{\zeta_1}(t,\xi))
*_\T f_{\#}(t,\xi)\right),
\\ &
G_2(t,\xi)=\begin{pmatrix}G_{2,+}(t,\xi)\\ G_{2,-}(t,\xi)\end{pmatrix}
=ie^{ic_{1,\eps}t\xi}P(\xi)^{*}\widetilde{F}_1(t,\xi),\\
& G_3(t,\xi)=\begin{pmatrix}G_{3,+}(t,\xi)\\ G_{3,-}(t,\xi)\end{pmatrix}
=-2e^{ic_{1,\eps}t\xi}\sigma_3P(\xi)^{*}\widetilde{F}_2(t,\xi).
\end{align*}
By the definition, $f_{\#}$ is $2\pi$-periodic in $\xi$.
Using $P(\xi)^{*}\hat{J}P(\xi)=-2i\sin\frac\xi2\sigma_3$, we see that
\eqref{eq:LFPU} translates into
\begin{equation}
  \label{eq:f}
  \begin{split}
\pd_tf
=& ic_{1,\eps}\xi f +e^{ic_{1,\eps}t\xi}P(\xi)^{*}
\mathcal{F}_n(JH''(u_{N,\eps}+\zeta)w)-i\left(G_2+\sin\frac{\xi}2G_3\right)
\\=&
\Lambda_\eps f+\frac{e^{ic_{1,\eps}t\xi}}{\sqrt{2\pi}}
P(\xi)^{*}\widehat{J}\left\{
\begin{pmatrix}\widetilde{r_{N,\eps}} +\widetilde{\zeta_1}&
 0\\ 0 & 0\end{pmatrix}*_\T(e^{-ic_{1,\eps}t\xi}
P(\xi)f)\right\}
\\ & -i\left(G_2+\sin\frac{\xi}2G_3\right)
\\=& \Lambda_\eps f-i
\begin{pmatrix}(G_1(t,\xi)+G_{3,+}(t,\xi))\sin\frac\xi2+G_{2,+}(t,\xi) \\ 
-(G_1(t,\xi)e^{-i\xi/2}-G_{3,-}(t,\xi))\sin\frac\xi2+G_{2,-}(t,\xi)
\end{pmatrix},
\end{split}
\end{equation}
where $\Lambda_\eps=\diag(i\lambda_{+,\eps},i\lambda_{-,\eps})$ and
$\lambda_{\pm,\eps}(\xi)=c_{1,\eps}\xi\mp2\sin(\frac{\xi}2)$
for $\xi\in[-\pi,\pi]$.
By Parseval's equality, we have
\begin{align*}
\|e^{\eps k_1(\cdot-c_{1,\eps}t)}w(t)\|_{l^2}
=& e^{-\eps k_1c_{1,\eps}t}\|\tau_{ik_1\eps}\mathcal{F}_nw(t)\|_{L^2(\T)}
\\=& \|e^{-ic_{1,\eps}t\xi}P(\cdot+i\eps k_1)f(t,\cdot+i\eps k_1)\|_{L^2(-\pi,\pi)}
\\ \lesssim & \|\tau_{ik_1\eps}f(t)\|_{L^2(-\pi,\pi)}.
\end{align*}
Thus to prove Lemma \ref{lem:LFPU-1}, it suffices to estimate
$\|\tau_{ik_1\eps}f(t)\|_{L^2(\T)}$.
\par
To begin with, we will show the lower bound of $\Im\lambda_{\pm}$.
\begin{lemma}
\label{lem:lbd}
Let $a\in(0,2k_1)$ and $\delta\in(0,\pi)$. Then there exist positive numbers
$K$ and $\eps_0$ such that for $\eps\in(0,\eps_0)$,
\begin{gather*}
 \lambda_{+,\eps}(\eps(\eta+ia))
=\frac{\eps^3}{24}\{(\eta+ia)^3+4k_1^2(\eta+ia)\}
+O(\eps^5\la \eta\ra^5) \quad\text{for $\eta\in [-2K,2K]$,}
\\
\Im\lambda_{+,\eps}(\eps(\eta+ia))\ge \frac{\eps^3a}{16}\eta^2
\quad\text{for $\eta\in[-2\delta\eps^{-1},-K]\cup[K,2\delta\eps^{-1}]$,}
\\
\Im\lambda_{+,\eps}(\eps(\eta+ia))
\ge \eps a(1-\cos\delta)
\quad\text{for $\eta\in [-\pi\eps^{-1},-\delta\eps^{-1}]
\cup [\delta\eps^{-1},\pi\eps^{-1}]$,}
\\
\Im\lambda_{-,\eps}(\eps(\eta+ia))\ge \eps a
\quad\text{for $\eta\in [-\pi\eps^{-1},\pi\eps^{-1}]$.}
\end{gather*}
\end{lemma}
\begin{proof}
Let $\xi=\eps(\eta+ia)$.
For $\eta\in[-2K,2K]$, we have
\begin{align*}
  \lambda_{+,\eps}(\xi)=& \eps c_{1,\eps}(\eta+ia)
-2\sin\frac{\eps(\eta+ia)}2
\\= &
\frac{\eps^3k_1^2}{6}(\eta+ia)+\frac{\eps^3}{24}(\eta+ia)^3
+O(\eps^5(\eta+ia)^5)
\\ =&
\frac{\eps^3}{24}\{(\eta+ia)^3+4k_1^2(\eta+ia)\}
+O(\eps^5\la \eta\ra^5).
\end{align*}
Since
$$\lambda_{\pm,\xi}(\xi)=
\eps c_{1,\eps}(\eta+ia)\mp2\left(\sin\frac{\eps\eta}{2}
\cosh\frac{\eps a}2+i\cos\frac{\eps\eta}{2}\sinh\frac{\eps a}{2}\right),
$$
we have 
$\Im \lambda_{-,\eps}(\xi)\ge \eps c_{1,\eps}a \ge \eps a$
for $\eta\in[-\pi\eps^{-1},\pi\eps^{-1}]$, and
\begin{align*}
 \Im \lambda_{+,\eps}(\xi) =&
\eps c_{1,\eps}a-2\sinh\frac{\eps a}{2}\cos\frac{\eps\eta}{2}
\\=& 2\sinh\frac{\eps a}{2}\left(1-\cos\frac{\eps\eta}{2}\right)
+\eps c_{1,\eps}a-2\sinh\frac{\eps a}{2}
\\ \ge& \frac{\eps^3a}8(1+O(\delta^2))\eta^2+O(\eps^3)
\quad\text{for $\eta\in[K,\delta\eps^{-1}]\cup
[-\delta\eps^{-1},-K]$,}
\end{align*}
\begin{align*}
\Im \lambda_{+,\eps}(\xi) \ge & 2\sinh\frac{\eps a}{2}(1-\cos\delta)
+\eps c_{1,\eps}a-2\sinh\frac{\eps a}{2}
\\ \ge & \eps a(1-\cos\delta)+O(\eps^3)
\quad\text{for $\eta\in[-\pi\eps^{-1},\delta\eps^{-1}]\cup
[\delta\eps^{-1},\pi\eps^{-1}]$.}
\end{align*}
\end{proof}
We need the following lemma to estimate the potential term of \eqref{eq:f}.
\begin{lemma}
  \label{lem:FT-FS}
  \begin{enumerate}
  \item
Suppose $f\in L^\infty(\R)$, $\mathcal{F}_nf\in L^1(\T)$
and  $g\in L^2(\T)$. Then
$$
\left\|\int_\T \widetilde{f}(\xi_1)g(\xi-\xi_1)d\xi_1\right\|_{L^2(\T)}
\le \|f\|_{L^\infty(\R)}\|g\|_{L^2(\T)}.$$
\item
Let $0<\delta<\pi(4\sum_{n=1}^Nk_i)^{-1}$. Then as $\eps\to0$,
\begin{align*}
\sup_{t\ge0,\,\xi\in[-\pi,\pi],\,\mgamma\in\R^N}
\left|\widetilde{r}_{N,\eps}(t,\xi_1,\mgamma)
-\widehat{r}_{N,\eps}(t,\xi_1,\mgamma))\right|=O(e^{-\pi\delta/\eps}).
\end{align*}
  \end{enumerate}
\end{lemma}
See Appendix \ref{sec:apa} for the proof.
Now we start to prove Lemma \ref{lem:LFPU-1}.
\begin{proof}[Proof of Lemma \ref{lem:LFPU-1} (the former part)]
Since $\inf_{\xi\in\R}t^{-1}\log |e^{t \Lambda_\eps(\xi+ik_1\eps)}|
\lesssim -\eps^3$ and
is of the same order as the size of the potential term of \eqref{eq:f},
Lemma \ref{lem:lbd} is not sufficient to prove exponential linear stability.
We will decompose solutions into a high frequency part,
a middle frequency part and a low frequency part.
\par
 Let $\chi$ and $\tilde{\chi}$ be nonnegative smooth functions
such that $\chi+\tilde{\chi}=1$ and $\chi(\xi)=1$ if $\xi\in[-1,1]$ and
$\chi(\xi)=0$ if $|\xi|\ge 2$. Let $\chi_b(\xi)=\chi(\xi/b)$ and
$\tilde{\chi}_b(\xi)=\tilde{\chi}(\xi/b)$.
Let $K$ be a large number satisfying $K\eps_0^{\frac12}\le 1$,
$\xi_\eps=\xi+ik_1\eps$ and
\begin{align*}
& f_{1,+}(t,\xi)=\chi_{K\eps}(\xi)f_+(t,\xi_\eps),\quad
f_{2,+}(t,\xi)=(\chi_{\delta}(\xi)-\chi_{K\eps}(\xi)) f_+(t,\xi_\eps),
\\ & f_{3,+}(t,\xi)=\tilde{\chi}_\delta f_+(t,\xi_\eps), \quad
f_3(t,\xi)=(f_{3,+}(t,\xi),\,f_-(t,\xi_\eps)).
\end{align*}
Then by \eqref{eq:f},
\begin{align*}
\pd_tf_{1,+}(t,\xi)= & i\lambda_{+,\eps}(\xi_\eps)f_{1,+}(t,\xi)
\\ & -i\chi_{K\eps}(\xi)
\left(G_{2,+}(t,\xi_\eps)+(G_1(t,\xi_\eps)+G_{3,+}(t,\xi_\eps))\sin\frac{\xi_\eps}2\right),
\\
\pd_tf_{2,+}(t,\xi)= & i\lambda_{+,\eps}(\xi_\eps)f_{2,+}(t,\xi)
\\ & -i(\chi_{\delta}(\xi)-\chi_{K\eps}(\xi))
\left(G_{2,+}(t,\xi_\eps)+(G_1(t,\xi_\eps)+G_{3,+}(t,\xi_\eps))\sin\frac{\xi_\eps}2\right),
\\
\pd_tf_{3,+}(t,\xi)= & i\lambda_{+,\eps}(\xi_\eps)f_{3,+}(t,\xi)
\\ & -i\tilde{\chi}_{\delta}(\xi)
\left(G_{2,+}(t,\xi_\eps)+(G_1(t,\xi_\eps)+G_{3,+}(t,\xi_\eps))\sin\frac{\xi_\eps}2\right),
\\
\pd_tf_{-}(t,\xi)= & i\lambda_{-,\eps}(\xi_\eps)f_{-}(t,\xi)
\\ & +i\left((G_1(t,\xi_\eps)e^{\frac{it\xi_\eps}2}-G_{3,-}(t,\xi_\eps))
\sin\frac{\xi_\eps}{2}-G_{2,-}(t,\xi_\eps)\right).
\end{align*}
Except for the low frequency part $f_{1,+}$, potential terms of the above
equations are negligible. In the former part of the proof, we will estimate
$\|f_{2,+}\|_{L^2}$ and $\|f_3\|_{L^2}$.
\par
Lemma \ref{lem:lbd} implies that $\Im\lambda_{-,\eps}(\xi_\eps)\ge k_1\eps$ for $\xi\in[-\pi,\pi]$
and that there exists $\alpha\in(0,k_1)$ such that
$\Im\lambda_{+,\eps}(\xi_\eps)\ge \alpha\eps$ for $\xi\in\supp\tilde{\chi}_\delta$.
Using the variation of constants formula and Minkowski's inequality, we have
\begin{align*}
\|f_{3,+}(t)\|_{L^2} \lesssim & e^{-\alpha\eps t}\|f_{3,+}(0)\|_{L^2}
\\ & + \int_0^t e^{-\alpha\eps(t-s)}(\|G_1(s,\xi_\eps)\|_{L^2}
+\|G_{2}(s,\xi_\eps)\|_{L^2}+\|G_3(s,\xi_\eps)\|_{L^2})ds.
\end{align*}
Using Parseval's identity, we have
$$\|G_2(s,\xi_\eps)\|_{L^2}\lesssim e^{-k_1\eps c_{1,\eps}s} \|F_1(s)\|_{l^2_{k_1\eps}},
\quad \|G_3(s,\xi_\eps)\|_{L^2}\lesssim e^{-k_1\eps c_{1,\eps}s}\|F_2(s)\|_{l^2_{k_1\eps}}.$$
Since $\|r_{N,\eps}\|_{L^\infty}=O(\eps^2)$, it follows from Lemma \ref{lem:FT-FS} that
\begin{align*}
\|G_1(s,\xi_\eps)\|_{L^2}\lesssim & \|r_N+\zeta_1\|_{L^\infty}
(\|f_+(s,\xi_\eps)\|_{L^2}+\|f_-(s,\xi_\eps)\|_{L^2})
\\ \lesssim & \eps^2(\|f_{1,+}(s)\|_{L^2}+\|f_{2,+}(s)\|_{L^2}+\|f_3(s)\|_{L^2}).
\end{align*}
Combining the above, we obtain
\begin{equation}
    \label{eq:f3-est}
\begin{split}
& \|f_{3,+}(t)\|_{L^2} \\ \lesssim & 
e^{-\alpha\eps t}\|f_{3,+}(0)\|_{L^2}
+\int_0^t e^{-\alpha\eps(t-s)}e^{-k_1\eps c_{1,\eps}s}
(\|F_1(s)\|_{l^2_{k_1\eps}}+\|F_2(s)\|_{l^2_{k_1\eps}})ds
\\ & +\eps^2\int_0^t e^{-\alpha\eps(t-s)}
(\|f_{1,+}(s)\|_{L^2}+\|f_{2,+}(s)\|_{L^2}+\|f_3(s)\|_{L^2})ds.
\end{split}
\end{equation}
\par
Next we will estimate $\|f_{2,+}(t)\|_{L^2}$.
Noting that  $\Im\lambda_{+,\eps}\ge k_1\eps\xi^2/16$ and
$\left|\sin\frac{\xi_\eps}{2}\right|\lesssim |\xi|$ on $\supp(\chi_\delta-\chi_{K\eps})$
 and using the variation of constants formula, we have 
\begin{align*}
\|f_{2,+}(t)\|_{L^2}\lesssim & \|e^{-\frac{k_1\eps t\xi^2}{16}}f_{2,+}(0)\|_{L^2} +
\int_0^t \|\xi e^{-\frac{k_1\eps\xi^2(t-s)}{16}}G_1(s,\xi_\eps)\|_{L^2}ds
\\ &  +\int_0^t \left(\|e^{-\frac{k_1\eps\xi^2(t-s)}{16}}G_{2}(s,\xi_\eps)\|_{L^2}
+\|\xi e^{-\frac{k_1\eps\xi^2(t-s)}{16}}G_3(s,\xi_\eps)\|_{L^2}\right)ds.  
\end{align*}
Since $|\xi|e^{-\frac{k_1\eps\xi^2(t-s)}{16}}\lesssim 
(\eps(t-s))^{-\frac12} e^{-\frac{k_1K^2\eps^3(t-s)}{32}}$
for $\xi\in\supp(\chi_\delta-\chi_{K\eps})$,
\begin{equation}
  \label{eq:f2-est}
  \begin{split}
& \|f_{2,+}(t)\|_{L^2} \\ \lesssim & e^{-\frac{k_1K^2\eps^3t}{16}}
\|f_{2,+}(0)\|_{L^2}
+\int_0^t e^{-\frac{k_1K^2\eps^3(t-s)}{16}}e^{-k_1\eps c_{1,\eps}s}\|F_1(s)\|_{l^2_{k_1\eps}}ds
\\ & +
\eps^{-\frac12} \int_0^t (t-s)^{-\frac12} e^{-\frac{k_1K^2\eps^3(t-s)}{32}}
e^{-k_1\eps c_{1,\eps}s}\|F_2(s)\|_{l^2_{k_1\eps}}ds
\\ & +\eps^{\frac32} \int_0^t (t-s)^{-\frac12} e^{-\frac{k_1K^2\eps^3(t-s)}{32}}
(\|f_{1,+}(s)\|_{L^2}+\|f_{2,+}(s)\|_{L^2}+\|f_3(s)\|_{L^2})ds.
  \end{split}
\end{equation}
\end{proof}
For the low frequency part, both the dispersion induced by
discreteness of spatial variable and the potential produced by an
$N$-soliton $r_{N,\eps}$ are the same order.
We will show that the balance between the dispersion and the potential is
described by the linearized KdV equation around an $N$-soliton solution as
was observed by \cite{FP4} for a $1$-soliton solution. 

We need that $\mathcal{P}(\tau)$ is uniformly bounded for $\tau\ge0$.
\begin{lemma}
 \label{lem:KdVPb}
 Let $0<k_1<\cdots<k_N$, $a\in(0,2k_1)$ and $\tau_0\in\R$.
  There exists a positive constant $C$ depending only on $k_1,\cdots, k_N$ and $a$
 such that if $4k_1^2\tau_0+\gamma_1\le\cdots\le 4k_N^2+\gamma_N$,
  $$\sup_{\tau\ge \tau_0}\|\mathcal{P}(\tau)\|_{B(L^2_a)}\le C.$$
\end{lemma}
To estimate $\|f_{1,+}\|_{L^2}$, we need to show that the low
frequency part $f_{1,+}$ approximately satisfies the secular term condition
for a linearized KdV equation \eqref{eq:secKdV1} and \eqref{eq:secKdV2}.
Let $\mathcal{P}_1(\tau)=e^{k_y}\tau_{4k_1^2\tau}\mathcal{P}(\tau)
\tau_{-4k_1^2\tau}e^{-k_1y}$ and let $h_i(\tau)$ be an $L^2(\R)$-function such that
$$h_i(\tau,y)=\frac{1}{\sqrt{2\pi}}\int_{-\pi\eps^{-1}}^{\pi\eps^{-1}}
f_{i,+}(t,\eps\eta)e^{iy\eta}dy \quad\text{for $i=1$, $2$.}$$
\begin{lemma}
  \label{lem:projerr}
If $w(t)$ satisfies \eqref{eq:orth3}, then
$$\eps^{\frac12}\|\mathcal{P}_1(\tau)h_1(\tau)\|_{L^2}\lesssim
(\eps^2+\delta_2+K^{-2})\|\tau_{ik_1\eps} f(t)\|_{L^2(\T)}.$$
\end{lemma}

The proof of Lemmas \ref{lem:KdVPb} and \ref{lem:projerr} 
will be given in Appendix \ref{sec:apb}.

\begin{proof}[Proof of Lemma \ref{lem:LFPU-1} (continued)]
Finally, we will estimate $f_{1,+}$.
Let $\tau=\eps^3 t/24$, $\xi=\eps\eta$ and
\begin{gather*}
G_4(t,\xi)=\frac{1}{\sqrt{2\pi}}\int_{-\pi}^\pi
e^{ic_{1,\eps}t\xi}\widehat{r_{N,\eps}}(t,\xi_1;\mk,\mgamma)f_{\#}(t,\xi-\xi_1)
d\xi_1,\\ G_5(t,\xi)=\frac{e^{ic_{1,\eps}t\xi}\sin\tfrac{\xi}{2}}{\sqrt{2\pi}\xi}
(\widetilde{\zeta}*_\T f_{\#})(t,\xi_\eps).
\end{gather*}
Lemma \ref{lem:FT-FS} implies that for any $N>0$,
\begin{align*}
& \left\|\chi_K(\eta)\left((G_1(t,\xi_\eps)-G_4(t,\xi_\eps))
\sin\tfrac{\xi_\eps}{2}-\xi_\eps G_5(t,\xi_\eps)
\right)\right\|_{L^2_\eta(\R)}
\\ \lesssim & \eps^{-\frac12}e^{-k_1\eps c_{1,\eps}t} \left\|
\int_{-\pi}^\pi (\widetilde{r}_{N,\eps}(t,\xi_1;\mk,\mgamma)
-\widehat{r}_{N,\eps}(t,\xi_1;\mk,\mgamma)) f_\#(t,\xi_\eps-\xi_1)d\xi_1
\right\|_{L^2(-\pi,\pi)}
\\ \lesssim & \eps^{N-\frac12}e^{-k_1\eps c_{1,\eps}t}
\|f_{\#}(t,\xi_\eps)\|_{L^2(\T)}
\\ \lesssim & \eps^N(\|h_1\|_{L^2(\R)}+\|h_2\|_{L^2(\R)}).
\end{align*}
Using Parseval's identity and the fact that
$\sup_{t,n}|\zeta_1|=O(\delta_1\eps^2)$, we have
\begin{align*}
  \|\chi_KG_5\|_{L^2_\eta(\R)}\lesssim & \eps^{-\frac12}
\|\widetilde{\zeta_1}*_\T f_{\#}\|_{L^2(\T)}
\\ \lesssim & \delta_1\eps^2(\|h_1\|_{L^2(\R)}+\|h_2\|_{L^2(\R)}).
\end{align*}
Since $\sin(\eps(\eta+ia))=\frac{\eps}{2}(\eta+ia)
+O(\eps^3\la \eta\ra^3)$ for $\eta\in [-K,K]$ and
\begin{equation}
  \label{eq:rphi}
e^{ic_{1,\eps}t\xi}\widehat{r_{N,\eps}}(t,\xi;\mk,\mgamma)=
\eps e^{4ik_1^2\tau\eta}\widehat{\varphi_N}(\tau,\eta;\mk,\eps\mgamma)  
\end{equation}
by the definition of $r_{N,\eps}$, it follows that
\begin{align*}
&  \left\|\chi_{K}(\eta)
\left(\sin\frac{\eps(\eta+ik_1)}{2}-\frac{\eps(\eta+ik_1)}{2}\right)G_4(t,\xi_\eps)\right\|_{L^2_\eta(\R)}
\\ \lesssim & \eps^3\|G_4(t,\xi_\eps)\|_{L^2_\eta(-2K,2K)}
\\ \lesssim & \eps^5 \left\|\int_{-\pi\eps^{-1}}^{\pi\eps^{-1}} |\widehat{\varphi}_{N}(\tau,\eta_1)|
|e^{ic_{1,\eps}\eps t(\eta-\eta_1+ik_1)} f_{\#}(t,\eps(\eta-\eta_1+ik_1))|d\eta_1\right\|_{L^2(-2K,2K)}
\\ \lesssim & \eps^5(\|h_1(\tau)\|_{L^2(\R)}+\|h_2(\tau)\|_{L^2(\R)}).
\end{align*}
 Let
\begin{align*}
G_6(\tau,y)=& -\frac{\eps^3}{2}(\pd_y-k_1)\left\{
\varphi_N(\tau,y+4k_1^2\tau;\mk,\eps\mgamma)
(h_1(\tau,y)+h_2(\tau,y)\right\}
\\ =& : G_{6,1}+G_{6,2}.
\end{align*}
By \eqref{eq:rphi},
\begin{align*}
& \widehat{G_6}(\tau,\eta)
\\ =& -\frac{i\eps^3(\eta+ik_1)}{2\sqrt{2\pi}}
\int_{\R}e^{i4k_1^2\tau(\eta-\eta_1)}
\widehat{\varphi_N}(\tau,\eta-\eta_1;\mk,\eps\mgamma)
(\hat{h}_1(\tau,\eta_1)+\hat{h}_2(\tau,\eta_1))d\eta_1.
\end{align*}
Since $\supp h_{i}(\tau,\cdot)\subset [-\pi/\eps,\pi/\eps]$ for $i=1,2$,
$$
\widehat{G_6}(\tau,\eta)+\frac{i\xi_\eps}{2}G_4(t,\xi_\eps)
=\frac{i\xi_\eps}{2\sqrt{2\pi}}
\left(\int_{-\pi}^{\pi}-\int_{-\pi+\xi}^{\pi+\xi}\right)
e^{ic_{1,\eps}t\xi}\widehat{r_{N,\eps}}(t,\xi_\eps-\xi_1)
f_{\#}(t,\xi_1)d\xi_1.$$
If $\xi\in[-2K\eps,2K\eps]$ and $|\xi_1\pm\pi|\le|\xi|$,
we have $\widehat{r_{N,\eps}}(t,\xi-\xi_1)
=O(e^{-\pi^2/(8\sum_{i=1}^Nk_i\eps)})$ and
\begin{align*}
& \left\|\chi_K(\eta)\left(\widehat{G_6}
+\frac{i\eps(\eta+ik_1)}{2}G_4\right)\right\|_{L^2_\eta(\R)}
\\ \lesssim &
K^{\frac12}e^{-\pi^2/(8\sum_{i=1}^Nk_i\eps)}e^{-k_1\eps c_{1,\eps}t}
\|f_{\#}\|_{L^2(-\pi,\pi)}
\\ \lesssim & \eps^N(\|h_1\|_{L^2(\R)}+\|h_2\|_{L^2(\R)})
\quad\text{for any $N\ge1$.}
\end{align*}
Since
\begin{equation*}
\lambda_{+,\eps}(\xi_\eps)=\frac{\eps^3}{24}(\eta+ik_1)\{(\eta+ik_1)^2+4k_1^2+O(\eps^2\la\eta\ra^4)\}
\end{equation*}
for $\eta\in [-2K,2K]$,
\begin{align*}
& \pd_tf_{1,+}(t,\xi)-i\lambda_{+,\eps}(\xi_\eps)f_{1,+}(t,\xi)
\\=& \frac{\eps^3}{24}\mathcal{F}_y
\left\{\pd_\tau h_1-4k_1^2(\pd_y-k_1)h_1+(\pd_y-k_1)^3h_1+O(\eps^2h_1)\right\}.
\end{align*}
Combining the above, we obtain
 \begin{equation}
   \label{eq:h1}
   \begin{split}
& \pd_\tau h_1+\{(\pd_y-k_1)^3-4k_1^2(\pd_y-k_1)\}h_1
+12(\pd_y-k_1)\{\varphi_N(\tau,y+4k_1^2\tau)h_1\}
\\=&  24\eps^{-3}\mathcal{F}^{-1}_\eta\{\chi_K\widehat{G_{6,2}}-\tilde{\chi}_K\widehat{G_{6,1}}
-i\chi_K(\xi_\eps G_5+G_2'+G_3'\sin\frac{\xi_\eps}{2})\}
\\ & +O(\eps^2(h_1+h_2)),
   \end{split}
 \end{equation}
 where $G_2'(\tau,\eta):=G_{2,+}(t,\xi_\eps)$ and $G_3'(\tau,\eta):=
G_{3,+}(t,\xi_\eps)$. 
\par

By \eqref{eq:h1} and Theorem \ref{thm:linearizedKdV},
\begin{equation}
  \label{eq:h1-est1}
  \begin{split}
& \|\mathcal{Q}(\tau)h_1(\tau)\|_{L^2}  \lesssim  e^{-3k_1^3\tau}\|\mathcal{Q}(0)h_1(0)\|_{L^2}
+\eps^2\int_0^\tau e^{-3k_1^3(\tau-\tau_1)}(\|h_1\|_{L^2}+\|h_2\|_{L^2})d\tau_1
\\ & +\eps^{-3} \int_0^\tau e^{-3k_1^3(\tau-\tau_1)}(\tau-\tau_1)^{-\frac34}
\|\la\eta\ra^{-\frac32}\tilde{\chi}_K\widehat{G_{6,1}}\|_{L^2}d\tau_1
\\ & +\eps^{-3}\int_0^\tau  e^{-3k_1^3(\tau-\tau_1)}(\tau-\tau_1)^{-\frac12}
 \{\eps(\|G_3'\|_{L^2}+\|G_5\|_{L^2})
+\|\la\eta\ra^{-\frac12}\widehat{G_{6,2}}\|_{L^2}\}d\tau_1
\\ & +\eps^{-3}\int_0^\tau e^{-3k_1^3(\tau-\tau_1)} \|G_2'\|_{L^2}d\tau_1
\\ \lesssim & a(\tau)+\int_0^\tau e^{-3k_1^3(\tau-\tau_1)}
\{\eps^2+\delta_1(\tau-\tau_1)^{-\frac12}+K^{-\frac12}(\tau-\tau_1)^{-\frac34}\}
\|\mathcal{Q}(\tau_1)h_1(\tau_1)\|_{L^2}d\tau_1,
\end{split}
\end{equation}
where
\begin{align*}
 & a(\tau) = e^{-3k_1^3\tau}\|\mathcal{Q}(0)h_1(0)\|_{L^2}
\\ & +\int_0^\tau e^{-3k_1^3(\tau-\tau_1)}
(\tau-\tau_1)^{-\frac12}(\|h_2(\tau_1)\|_{L^2}+\eps^{-2}\|G_3'\|_{L^2})d\tau_1
\\ &  +\int_0^\tau e^{-3k_1^3(\tau-\tau_1)}
(\eps^2\|h_2(\tau_1)\|_{L^2}+\eps^{-3}\|G_2'\|_{L^2})d\tau_1
\\ & +
\int_0^\tau e^{-3k_1^3(\tau-\tau_1)}
\{\eps^2+\delta_1(\tau-\tau_1)^{-\frac12}+K^{-\frac12}(\tau-\tau_1)^{-\frac34}\}
\|\mathcal{P}(\tau_1)h_1(\tau_1)\|_{L^2}d\tau_1.
\end{align*}
Applying Gronwall's inequality to \eqref{eq:h1-est1}, we have
$$ \|\mathcal{Q}(\tau)h_1(\tau)\|_{L^2} \lesssim 
a(\tau)+\int_0^\tau e^{-2k_1^3(\tau-\tau_1)}(\tau-\tau_1)^{-\frac34}
a(\tau_1)d\tau_1$$
if $\eps$, $\delta_1$ and $K^{-\frac12}$ are sufficiently small.
Now we use the following computation result.
\begin{claim}
  \label{cl:grhelp}
Let $b>a>0$, $0<\alpha$, $\beta<1$, $t\ge0$ and $g(t)$ be
a nonnegative measurable function. Then
\begin{align*}
& \int_0^t e^{-b(t-s)}(t-s)^{-\beta}
\left(\int_0^s e^{-a(s-\tau)}(s-\tau)^{-\alpha}g(\tau)d\tau\right)ds
\\ \lesssim & \int_0^t e^{-a(t-s)}(t-s)^{1-(\alpha+\beta)}g(s)ds.  
\end{align*}
\end{claim}
By Lemma \ref{lem:projerr}, the definition of $h_2$ and Claim \ref{cl:grhelp},
we have
\begin{align*}
& \|\mathcal{Q}(\tau)h_1(\tau)\|_{L^2} \lesssim 
e^{-2k_1^3\tau}\|\mathcal{Q}(0)h_1(0)\|_{L^2}
\\\ & +\int_0^\tau e^{-2k_1^3(\tau-\tau_1)}
\{\eps^{-3}\|G_2'(\tau_1)\|_{L^2}
+\eps^{-2}(\tau-\tau_1)^{-\frac12}\|G_3'(\tau_1)\|_{L^2}\}d\tau_1
\\ & +\eps^{-\frac12}\int_0^\tau e^{-2k_1^3(\tau-\tau_1)}
(\tau-\tau_1)^{-\frac34}
(\delta_3\|f_{1,+}\|_{L^2}+\|f_{2,+}\|_{L^2}+\|f_3\|_{L^2})d\tau_1,
  \end{align*}
where $\delta_3=(\eps^2+\delta_2+K^{-2})(\eps^2+\delta_1+K^{-\frac12})$.
Combining the above and  Lemma \ref{lem:projerr}, we obtain
\begin{equation}
  \label{eq:f1-est}
  \begin{split}
& \|f_{1,+}(t)\|_{L^2} \lesssim  e^{-\frac{k_1^3\eps^3t}{12}}\|f_{1,+}(0)\|_{L^2}
+\delta_4(\|f_{2,+}(t)\|_{L^2}+\|f_3(t)\|_{L^2})
\\ & + \int_0^t e^{-\frac{k_1^3\eps^3(t-s)}{12}} e^{-k_1\eps c_{1,\eps}s}
\{\|F_1(s)\|_{l^2_{k_1\eps}}+\eps^{-\frac12}(t-s)^{-\frac12}\|F_2(s)\|_{l^2_{k_1\eps}}
\}ds
\\ & + \eps^{\frac34}\int_0^t e^{-\frac{k_1^3\eps^3(t-s)}{12}}(t-s)^{-\frac34}
(\delta_3\|f_{1,+}(s)\|_{L^2}+\|f_{2,+}(s)\|_{L^2}+\|f_3(s)\|_{L^2})ds,
  \end{split}
\end{equation}
where $\delta_4=\eps+\delta_2+K^{-2}$.
\par
By \eqref{eq:f2-est} and \eqref{eq:f3-est} and the fact that
\begin{gather*}
e^{-\frac{k_1K^2\eps^3(t-s)}{32}}\lesssim
 K^{-\frac12}\eps^{-\frac34}(t-s)^{-\frac14}
e^{-\frac{k_1^3\eps^3(t-s)}{12}},\\
e^{-\alpha\eps(t-s)}\lesssim  \max\{\eps^{-\frac34}(t-s)^{-\frac34},\eps^{-\frac12}(t-s)^{-\frac12}\}
e^{-\frac{k_1^3\eps^3(t-s)}{12}},
\end{gather*}
we have 
\begin{equation}
  \label{eq:f23-est}
  \begin{split}
& \|f_{2,+}(t)\|_{L^2}+\|f_3(t)\|_{L^2}  \lesssim 
 e^{-\frac{k_1^3\eps^3t}{12}}
(\|f_{1,+}(0)\|_{L^2}+\|f_{2,+}(0)\|_{L^2}+\|f_{3}(0)\|_{L^2})
\\ & +
\int_0^t e^{-\frac{k_1^3\eps^3(t-s)}{12}} e^{-k_1\eps c_{1,\eps}s}
\{\|F_1(s)\|_{l^2_{k_1\eps}}+\eps^{-\frac12}(t-s)^{-\frac12}\|F_2(s)\|_{l^2_{k_1\eps}}\}ds
\\ & + \eps^{\frac34}(K^{-\frac12}+\eps^{\frac12})
\int_0^t e^{-\frac{k_1^3\eps^3(t-s)}{12}}(t-s)^{-\frac34}
\\ & \phantom{ +\int_0^t e^{-k_1^3\eps^3(t-s)/12}\quad}
\times (\|f_{1,+}(s)\|_{L^2}+\|f_{2,+}(s)\|_{L^2}+\|f_3(s)\|_{L^2})ds.
  \end{split}
\end{equation}
Let $X(t)=\|f_{1,+}(t)\|_{L^2}
+\delta_4^{-\frac18}(\|f_{2,+}(t)\|_{L^2}+\|f_3(t)\|_{L^2})$.
By \eqref{eq:f1-est} and \eqref{eq:f23-est},
$$
X(t)\lesssim a_1(t) +\delta_4^{\frac18}\eps^{\frac34}
\int_0^t e^{-\frac{k_1^3\eps^3(t-s)}{12}}(t-s)^{-\frac34}X(s)ds,$$
where
\begin{align*}
& a_1(t)= e^{-\frac{k_1^3\eps^3t}{12}}X(0)
\\ & +\int_0^t e^{-\frac{k_1^3\eps^3(t-s)}{12}} e^{-k_1\eps c_{1,\eps}s}
\{\|F_1(s)\|_{l^2_{k_1\eps}}
+\eps^{-\frac12}(t-s)^{-\frac12}\|F_2(s)\|_{l^2_{k_1\eps}}\}ds.
\end{align*}
Applying \cite[Lemma 7.1.1]{He} to the above and using Claim \ref{cl:grhelp},
we obtain 
\begin{equation}
  \label{eq:X-est1}
  \begin{split}
X(t)\lesssim & a_1(t)+\eps^{\frac34}\int_0^t e^{-(\frac{k_1^3}{12}+O(\delta_4^{\frac12}))\eps^3(t-s)}
(t-s)^{-\frac34}a_1(s)ds
\\ \lesssim &
e^{-\frac{k_1^3\eps^3t}{24}}X(0)+\int_0^t e^{-\frac{k_1^3\eps^3(t-s)}{24}}
\{\|F_1(s)\|_{l^2_{k_1\eps}}+\eps^{-\frac12}(t-s)^{-\frac12}\|F_2(s)\|_{l^2_{k_1\eps}}\}ds.
\end{split}
\end{equation}
Thus we prove Lemma \ref{lem:LFPU-1}.
\end{proof}

\bigskip

\section{Exponential stability property of KdV N-solitons}
\label{sec:backlund}
In this section, we will prove linear stability  property of an $N$-soliton
solution of KdV equation \eqref{eq:KdV}.
We find that linear stability of an $N$-soliton in $L^2_a(\R)$ is
equivalent to that of an $(N-1)$-soliton connected by the
B\"acklund transformation \eqref{eq:BT} and it turns out that
exponential stability property of $N$-solitons in $L^2_a(\R)$ follows from
that of the null solution.
\par

First, we recall the B\"acklund transformation of KdV.
If $u$ is a solution of \eqref{eq:KdV} and $v(t,x)=-\int_x^\infty u(t,y)dy$,
\begin{equation}
  \label{eq:KdV2}
\pd_tv+\pd_x^3v+6(v_x)^2=0 \quad\text{for $x\in\R$ and $t>0$.}
\end{equation}
Eq. \eqref{eq:KdV2} admits a B\"acklund transformation
determined by the equations
\begin{equation}
  \label{eq:BT}
\left\{  \begin{aligned}
& \pd_x(v'+v)=k^2-(v'-v)^2\\
& \pd_t(v'+v)=2(v'-v)\pd_x^2(v'-v)
-4\{(\pd_xv')^2+(\pd_xv')(\pd_xv)+(\pd_xv)^2\}.
  \end{aligned}\right.
\end{equation}
If $v$ and $v'$ satisfy \eqref{eq:BT} and $v$ is a solution of \eqref{eq:KdV2},
then $v'$ is necessarily a solution of \eqref{eq:KdV2}
\par
To begin with, we recall that the B\"acklund transformation \eqref{eq:BT}
creates a 1-soliton solution from the null solution and an $N$-soliton
solution from an $(N-1)$-soliton solution (see \cite{WE}).
Let $0<k_1<\cdots<k_N$, $\mk^m=(k_1,\cdots,k_m)$,
$\mgamma^m=(\gamma^m_1,\cdots,\gamma^m_m)$,
$\theta^m_i=k_i(x-4k_i^2t-\gamma^m_i)$, 
\begin{align*}
& C_m=\begin{pmatrix}
\frac{e^{-(\theta^m_i+\theta^m_j)}}{k_i+k_j}\end{pmatrix}_{m\times m},
\\ & \Delta_m=
\begin{cases}
\exp(-\sum_{i=1}^N\theta_i^N) & \text{if $m=0$,}
\\  \exp(-\sum_{i=m+1}^N\theta_i^N)\det(I+C_m) & \text{if $1\le m\le N-1$,}
\\  \det(I+C_N) &\text{if $m=N$.}\end{cases}  
\end{align*}
Then $v^m=\pd_x\log\Delta_m$ $(0\le m\le N)$ is a solution of \eqref{eq:KdV2}
and $\varphi_m(t,x;\mk^m,\mgamma^m):=\pd_x^2\log \Delta_m$
is an $m$-soliton solution of \eqref{eq:KdV} (see \cite{GGKM}).
\par
An $m$-soliton solution is connected
to an $(m-1)$-soliton solution by \eqref{eq:BT}.
\begin{lemma}
  \label{lem:vm,m-1}
Suppose $1\le m\le N$ and that
\begin{equation}
  \label{eq:btphase}
\gamma^{m-1}_i=\gamma^{m}_i+\frac{1}{2k_i}
 \log\left(\frac{k_m-k_i}{k_m+k_i}\right)\quad
\text{for $1\le i\le m-1$.}
\end{equation}
Then
\begin{equation}
  \label{eq:m,m-1}
  \pd_x(v^m+v^{m-1})=k_m^2-(v^m-v^{m-1})^2.
\end{equation}
\end{lemma}
\begin{proof}
By the definition,
\begin{equation}
  \label{eq:v0-v1}
  v^0=-\sum_{i=1}^Nk_i\quad\text{and}\quad v^1=-\sum_{i=2}^Nk_i
-\frac{2k_1}{1+e^{2\theta^1_1}}.
\end{equation}
and \eqref{eq:m,m-1} is true for $m=1$.
\par
Let $m\ge 2$ and let $Q^m_{i\,j}$ be the $(i,j)$ cofactor of $I+C_m$.
Following the argument of \cite[p.121]{GGKM}, we have
\begin{equation*}
\psi_m:= \frac{\sum_{l=1}^m e^{-\theta_l^m}Q^m_{lm}}{\det(I+C_m)}
 =\frac{e^{-\theta_m^m}\det(I+C_{m-1})}{\det(I+C_{m})}
= e^{k_m(\gamma_m^m-\gamma_m^N)}\frac{\Delta_{m-1}}{\Delta_m},  
\end{equation*}
whence 
\begin{equation}
  \label{eq:psi_m}
 v^{m-1}-v^m=\pd_x\log \psi_m.
\end{equation}
On the other hand, Theorem 3.2 in \cite{GGKM} implies that
$$\pd_x^2\psi_m=(k_m^2-2\pd_xv^m)\psi_m.$$
Thus we have
\begin{align*}
\pd_x(v^m+v^{m-1})=&  \pd_x^2\log\psi_m +2\pd_x^2\log\Delta_m
\\=& \frac{\pd_x^2\psi_m}{\psi_m}-\left(\frac{\pd_x\psi_m}{\psi_m}\right)^2
+2\pd_xv^m
\\=& k_m^2-(v^m-v^{m-1})^2.
\end{align*}
\end{proof}

Now we linearize the B\"acklund transformation \eqref{eq:BT} around 
$v=v^m$ and $v'=v^{m-1}$. Then we obtain a linearized B\"acklund transformation
\begin{equation}
  \label{eq:LBT}
  \pd_x(w^m+w^{m-1})=-2(v^m-v^{m-1})(w^m-w^{m-1}).
\end{equation}
The semiflows generated by
\begin{gather}
 \label{eq:w}
\pd_tw^m+\pd_x^3w^m+12(\pd_xv^m)(\pd_xw^m)=0 \quad\text{for $x\in\R$,}\\
  \label{eq:w'}
\pd_tw^{m-1}+\pd_x^3w^{m-1}+12(\pd_xv^{m-1})(\pd_xw^{m-1})=0
\quad\text{for $x\in\R$,}
\end{gather}
leave the linearized B\"acklund transformation \eqref{eq:LBT} invariant.
Note that \eqref{eq:w} is a linearized equation of \eqref{eq:KdV2} around $v^m$ and the adjoint equation of \eqref{eq:KdV} if
$m=N$ and $\pd_xv_m=\varphi_N$.
\begin{lemma}
  \label{lem:2}
Let $a>0$, $t_0\in\R$ and let $w^m$, $w^{m-1}\in
C((-\infty,t_0];L^2_{-a}(\R))$ be solutions of \eqref{eq:w} and
\eqref{eq:w'}, respectively.
If \eqref{eq:LBT} holds at $t=t_0$, it holds for every $t\le t_0$.
\end{lemma}
Before we start to prove Lemma \ref{lem:2},
we remark that linearized equation of \eqref{eq:KdV2}
around $v^m$ is well posed in $L^2_{-a}$ (see e.g. \cite{K}).
\begin{lemma}
Let $a>0$, $\varphi\in L^2_{-a}(\R)$ and $t_0$ be a real number.
There exists a unique solution of 
\begin{equation*}
\left\{\begin{aligned} &
\pd_tw+\pd_x^3w+(\pd_xv^m)\pd_xw=0 \quad \text{for $x\in\R$ and $t<t_0$,}
\\ & w(t_0)=\varphi,    
  \end{aligned}\right.
\end{equation*}
in the class $C((-\infty,t_0];L^2_a(\R))$.
\end{lemma}
\begin{proof}[Proof of Lemma \ref{lem:2}]
Let 
$$W=(w^m+w^{m-1})_x+2(v^m-v^{m-1})(w^m-w^{m-1}).$$
By \eqref{eq:w}, \eqref{eq:w'} and the fact that
$v^m$ and $v^{m-1}$ are solutions of \eqref{eq:KdV2}, we have
\begin{align*}
  W_t+W_{xxx}+& 6(v^m+v^{m-1})_xW_x =  -6\{(v^m+v^{m-1})_x(w^m+w^{m-1})_x\}_x
\\ & +24(v^{m-1}_x w^{m-1}_x-v^m_xw^m_x)
+12(w^m-w^{m-1})(v^{m-1}_x)^2-(v^m_x)^2).
\end{align*}
Using \eqref{eq:LBT} twice and \eqref{eq:m,m-1}, we find
$$
\left\{
  \begin{aligned}
& W_t+W_{xxx}+6(v^{m}+v^{m-1})_xW_x=0,
\\ & W(t_0)=0.
  \end{aligned}\right.
$$
Let $\widetilde{W}(t,x)=(\pd_x^{-1}W)(t,x)=\int_{-\infty}^x W(t,y)dy$ and
$b=6(v^{m}+v^{m-1})_{xx}.$ Then
\begin{equation*}
\left\{
  \begin{aligned}
& \widetilde{W}_t+\widetilde{W}_{xxx}=
b\widetilde{W}_x-b_x\widetilde{W}-\pd_x^{-1}(b_{xx}\widetilde{W}),
\\ & \widetilde{W}(t_0)=0.    
  \end{aligned}\right.
\end{equation*}
Since $\pd_x^{-1}$ is bounded on $L^2_{-a}(\R)$ $(a>0)$,
we have $\widetilde{W}\in C((-\infty,t_0];L^2_{-a})$ and 
\begin{equation}
  \label{eq:f2}
 \|\widetilde{W}(t)\|_{L^2_{-a}}\lesssim \int_t^{t_0} (1+(s-t)^{-\frac12})
e^{a^3(s-t)}\|\widetilde{W}(s)\|_{L^2_{-a}}ds\quad\text{for $t\le t_0$.}
\end{equation}
by using \cite[Lemma 9.1]{K}.
Applying Gronwall's inequality to \eqref{eq:f2}, we have $\widetilde{W}(t)=0$
and $W(t)=\pd_x\widetilde{W}(t)=0$ for every $t\ge0$.
\end{proof}
The linearized B\"acklund transformation \eqref{eq:LBT} defines an isomorphism
between $L^2_a$ and its subspace
$$X_m(t,\mgamma^m)=\left\{w\in L^2_a:
\int_\R w\pd_{\gamma_m^m}\pd_xv^mdx=\int_\R w\pd_{k_m}\pd_xv^mdx=0\right\}.$$
First, let us consider the case $m=1$.
\begin{lemma}
\label{prop:bij1}
Let $a\in(-2k_1,2k_1)$. Then for any $w^0\in L^2_a(\R)$, there exists a unique
$w^1\in X_1(t,\mgamma^1)$ satisfying \eqref{eq:LBT1}.
Furthermore the map 
$\Phi_1(t,\mgamma^1):L^2_a\to X_1(t,\mgamma^1)$ defined by \eqref{eq:LBT1} is
isomorphic and
$$\sup_{t,\mgamma^1}\left(\|\Phi_1(t,\mgamma^1)\|_{B(L^2_a;X_1(t,\mgamma^1))}
+\|\Phi_1(t,\mgamma^1)^{-1}\|_{B(X_1(t,\mgamma^1);L^2_a)}\right)<\infty.$$
\end{lemma}
\begin{proof}
Substituting \eqref{eq:v0-v1} into \eqref{eq:LBT} with $m=1$, we have
\begin{equation}
  \label{eq:LBT1}
 \pd_x(w^1+w^0)=-2(\pd_xv^1)(w^1-w^0).
\end{equation}
Since  $\|\Phi_1(t,\mgamma^1)\|_{B(L^2_a;X_1(t,\mgamma^1))}$ and 
$\|\Phi_1(t,\mgamma^1)^{-1}\|_{B(X_1(t,\mgamma^1);L^2_a)}$ do not depend on
$t$ and  $\mgamma^1$, we may assume $t=0$ and $\mgamma^1=(0)$.
\par
Let $c=4k_1^2$ and  $\phi_c(x)=k_1^2\sech^2k_1x$. Then
\eqref{eq:LBT1} and be rewritten as
\begin{equation}
  \label{eq:LBT1'}
 \pd_x(w^1+w^0)=\frac{\pd_x\phi_c}{\phi_c}(w^1-w^0).
\end{equation}
By \eqref{eq:LBT1'}, there exists 
a real constant $\alpha$ such that
\begin{equation}
  \label{eq:w1}
w^1(x)=-w^0(x)-(I_1w^0)(x)+\alpha\phi_c(x),
\end{equation}
where
$$(I_1w)(x)_0=2\phi_c(x)
\int_0^x\frac{\pd_x\phi_c(y)}{\phi_c(y)^2}w^0(y)dy.$$
The constant $\alpha$ is uniquely determined by the orthogonality conditions.
Hereafter, we  use the notation $(f,g):=\int_\R f(x)g(x)dx$ in this section.
Since $d\|\phi_c\|_{L^2(\R)}^2/dc\ne0$ and $\int_\R\pd_x\phi_cdx=0$,
there exists a unique $\alpha=\alpha(w^0)$ such that
\begin{equation}
  \label{eq:alphadef0}
\begin{split}
(w^1,\pd_c\phi_c) =& -(w^0+I_1w^0,\pd_c\phi_c)
+\alpha( \phi_c,\pd_c\phi_c)
\\=& 0,
\end{split}  
\end{equation}
and 
\begin{align*}
(w^1,\pd_x\phi_c)=& (- w^0+I_1w^0+\alpha\phi_c,\pd_x\phi_c)
\\ =& -(w^0,\pd_x\phi_c)+(w^0,\pd_x\phi_c)=0.
\end{align*}
\par
Next we prove that $\Phi_1:w^0\mapsto w^1$ is continuous linear operator
from $L^2_a$ to $X_1$.
Noting that
\begin{align*}
 \phi_c(x)|\pd_x\phi_c(y)|\phi_c(y)^{-2} \lesssim &  \cosh^2(k_1y)\sech^2(k_1x)
\\  \lesssim  & e^{-\sqrt{c}|x-y|}\quad\text{for any $y\in (-|x|,|x|)$,} 
\end{align*}
we see that $I_1$ is a bounded linear operator on $L^2_a$.
Eq. \eqref{eq:alphadef0} and the boundedness of $I_1$ imply that
$\alpha(w^0)$ is continuous linear functional on $L^2_a$.
Thus we prove that \eqref{eq:LBT1'} defines $\Phi\in B(L^2_a, X_1)$.
\par
Next, we will prove that $\Phi_1$ has a bounded inverse.
By \eqref{eq:LBT1'},
$$\pd_x\{\phi_c(w^1+w^0)\}=2w^1\pd_x\phi_c,$$
and 
$$
w^0(x)=-w^1(x)-(J_1w^1)(x),$$
where
$$(J_1f)(x)=2\phi_c(x)^{-1}\int_x^\infty \pd_x\phi_c(y)f(y)dy
=-2\phi_c(x)^{-1}\int_{-\infty}^x \pd_x\phi_c(y)f(y)dy$$
for any $f\in X_1$. Noting that 
$$
\phi_c(x)^{-1}|\pd_x\phi_c(y)| \lesssim e^{-\sqrt{c}|x-y|}
\quad\text{ for $0\le x\le y$ or $y\le x\le 0$,}$$
we have
$$
\|J_1f\|_{L^2_a}\lesssim \|e^{-(\sqrt{c}-|a|)|x|}\|_{L^1}\|f\|_{L^2_a}
\lesssim \|f\|_{L^2_a}.$$
Thus we see that \eqref{eq:LBT1'} defines a bounded linear operator
$$\Psi_1 w^1:=w^0=-w^1-2J_1w^1$$ from $X_1$ to $L^2_a$.
\par
Since $\Phi_1\in B(L^2_a,X_1)$, $\Psi_1\in B(X_1,L^2_a)$ and
 $\Psi_1\Phi_1=I$ on $C^1(\R)\cap L^2_a$ and $\Phi_1\Psi_1=I$ on
$C^1(\R)\cap X_1$ by the definitions of $\Phi_1$ and $\Psi_1$,
we conclude that $\Phi_1:L^2_a\to X_1$ is isomorphic.
Thus we complete the proof of Lemma \ref{prop:bij1}.
\end{proof}
Next we will consider the case where $2\le m\le N$.
\begin{lemma}
  \label{prop:bij2}
Suppose  $a\in(-2k_m,2k_m)$ and \eqref{eq:btphase}.
Then for any $w^{m-1}\in L^2_a(\R)$,
there exists a unique $w^m\in X_m$ satisfying \eqref{eq:LBT}.
Furthermore the map $\Phi_m(t,\mgamma^m):L^2_a\to X_m$ defined by \eqref{eq:LBT}
is isomorphic and
$$\sup_{t,\mgamma^m}\left(\|\Phi_m(t,\mgamma^m)\|_{B(L^2_a;X_m)}
+\|\Phi_m(t,\mgamma^m)^{-1}\|_{B(X_m;L^2_a)}\right)<\infty.$$
\end{lemma}
To prove Lemma \ref{prop:bij2}, we need the following:
\begin{lemma}
  \label{lem:psimbd}
Suppose \eqref{eq:btphase}. Then there exist positive constants $C_1$ and
$C_2$ depending only on $\mk^m$ $(1\le i\le N)$ such that
$$C_1\sech\theta_m^m\le \psi_m\le C_2\sech\theta_m^m.$$
\end{lemma}
\begin{proof}
Expanding $\det(I+C_m)$, we obtain the sum of all the principal minors of
$C_m$ of every order:
$$\det(I+C_m)=1+\sum_{l=1}^m\sum_{1_1\le\cdots i_l}C_{i_1,\cdots,i_l}
e^{-(\theta_{i_1}^m+\cdots+\theta_{i_l}^m)},$$
where $C_{i_1,\cdots,i_l}$ are positive constants depending only on
$k_1,\cdots,k_N$ (see \cite[p.110]{GGKM}).
By \eqref{eq:btphase} and the above, there exist positive constants $C_1$
and $C_2$ depending only of $k_1,\cdots,k_N$ such that
\begin{align*}
 \frac{2C_1e^{-\theta_m^m}}{1+e^{-2\theta_m^m}} \le
e^{\theta_m^m}\psi_m= \frac{\det(I+C_{m-1})}{\det(I+C_m)}
\le  \frac{2C_2e^{-\theta_m^m}}{1+e^{-2\theta_m^m}}.
\end{align*}
\end{proof}
Now we are in position to prove Lemma \ref{prop:bij2}.
\begin{proof}[Proof of Lemma \ref{prop:bij2}]
Without loss of generality, we may assume $t=0$.
Let $A=\pd_x+2(v^m-v^{m-1})$ and $B=-\pd_x+2(v^m-v^{m-1})$.
Differentiating \eqref{eq:m,m-1} with respect to $k_m$ and $\gamma_m^m$,
we have
\begin{equation}
  \label{eq:gkerA}
A\pd_{\gamma_m^m}v^m=B^*\pd_{\gamma_m^m}v^m=0, \quad A\pd_{k_m}v^m=
B^*\pd_{k_2}v^m=2k_m.  
\end{equation}
\par
First, we solve \eqref{eq:LBT} for $w^m$.
Eq. \eqref{eq:LBT} can be translated into
\begin{equation}
  \label{eq:LBT2'}
  A(w^m+w^{m-1})=4(v^m-v^{m-1})w^{m-1}.
\end{equation}
By \eqref{eq:psi_m}, \eqref{eq:gkerA} and \eqref{eq:LBT2'},
\begin{equation}
  \label{eq:w1tow2}
w^m=-w^{m-1}+I_m(w^{m-1})+\alpha \pd_{\gamma_m^m}v^m,
\end{equation}
where $\alpha$ is a real number and
$$
I_m(f):=4\int_{\gamma^m_m}^x\left(v^m(y)-v^{m-1}(y)\right)
\frac{\psi_m(t,x,\mk^m,\mgamma^m)^2}{\psi_m(t,y,\mk^m,\mgamma^m)^2}f(y)dy.
$$
Lemma \ref{lem:psimbd} implies that there exists a positive constant
$C_3$ depending only on $\mk^m$ such that 
for every $x\ge y\ge \gamma^m_m$ or $x\le y\le \gamma_m^m$,
\begin{align*}
\frac{\psi_m(t,x,\mk^m,\mgamma^m)^2}{\psi_m(t,y,\mk^m,\mgamma^m)^2}
\le & C_3\frac{\sech\theta_m(t,x)^2}{\sech\theta_m(t,y)^2}
\\ \le & 4C_3e^{-2k_m|x-y|}.
\end{align*}
Thus we have $I_m\in B(L^2_a)$ for $a\in(0,2k_m)$.
\par
Next, we will show that $w^m\in X_m(t,\mgamma^m)$.
By\eqref{eq:LBT} and the definitions of $A$ and $B$,
$$Aw^m=Bw^{m-1}\quad\text{and}\quad\pd_x=(B^*-A^*)/2.$$
Using \eqref{eq:gkerA} and the above, we have
\begin{align*}
  2(w^m,\pd_x\pd_{\gamma_m^m}v^m)
=& (w^m,(B^*-A^*)\pd_{\gamma_m^m}v^m)
\\=&  -(Aw^m,\pd_{\gamma_m^m}v^m)
\\=&  -(Bw^{m-1},\pd_{\gamma_m^m}v^m)
\\=& -(w^{m-1},B^*\pd_{\gamma_m^m}v^m)=0,
\end{align*}
and
\begin{align*}
2(\pd_{\gamma_m^m}v^m,\pd_x\pd_{k_m}v^m)=&
(\pd_{\gamma_m^m}v^m,(B^*-A^*)\pd_{k_m}v^m)
\\=&   (\pd_{\gamma_m^m}v^m,B^*\pd_{k_m}v^m)
\\=& 2k_m(\pd_{\gamma_m^m}v^m,1)
\\=& -2k_m\left[\pd_{\gamma_m^m} \log\Delta_m\right]_{x=-\infty}^{x=\infty}
\\= & 2k_m\left[\frac{\pd_{\gamma_m^m}\det(I+C_m)}{\det(I+C_m)}
\right]_{x=-\infty}^{x=\infty}
\\=& -2k_m \frac{\pd_{\gamma_m^m}\det C_m}{\det C_m}\biggl|_{x=-\infty}
=-4k_m^2\ne0.
\end{align*}
Hence there exists a unique $\alpha=\alpha(w^{m-1})$ such that
$(w^m,\pd_x\pd_{k_m}v^m)=0$. Moreover, $\alpha(w^{m-1})$ is a continuous
linear functional on $w^{m-1}\in L^2_a$.
Thus we prove $\Phi_m(t,\mgamma^m)=-I+4I_m+\alpha(\cdot)\pd_{\gamma_m^m}v^m$
satisfies
$\sup_{t,\mgamma^m}\|\Phi_m(t,\mgamma^m)\|_{B(L^2_a,X_m(t,\mgamma^m))}<\infty.$
\par
Finally, we will prove  $\sup_{t,\mgamma^m}\|\Phi_m(t,\mgamma^m)^{-1}
\|_{B(X_m(t,\mgamma^m),L^2_a)}<\infty.$
Let us solve \eqref{eq:LBT} for $w^{m-1}$.
Since $\ker(B)=\{0\}$ in $L^2_a$ and 
$$B(w^{m-1}+w^m)=-4(v^m-v^{m-1})w^m,$$
we have for any $w^m \in C_0^1(\R)\cap X_m(t,\mgamma^m)$,
\begin{equation}
  \label{eq:T2}
  \begin{split}
w^{m-1}(x)=& -w^m(x)+4\int_x^\infty
\frac{\psi_m(t,y,\mk^m,\mgamma^m)^2}{\psi_m(t,x,\mk^m,\mgamma^m)^2}w^m(y)dy
\\=&  -w^m(x)-4\int^x_{-\infty}
\frac{\psi_m(t,y,\mk^m,\mgamma^m)^2}{\psi_m(t,x,\mk^m,\mgamma^m)^2}w^m(y)dy
\\=:& -w^{m}(x)+J_m(w^m)(x).
  \end{split}
\end{equation}
Lemma \ref{lem:psimbd} implies that there exists a positive
constant $C$ depending only on $\mk^m$ such that
$$
\frac{\psi_m(t,y,\mk^m,\mgamma^m)^2}{\psi_m(t,x,\mk^m,\mgamma^m)^2}
\le Ce^{-2k_m|x-y|}$$
for $\gamma_m^m\le y\le x$ or $x \le y\le \gamma_m^m$.
Hence $J_m$ can be uniquely extended on $X_m(t,\mgamma^m)$ and
$\Psi_m:=-I+J_m\in B(X_m(t,\mgamma^m),L^2_a)$ satisfies
$\sup_{t,\mgamma^m}\|\Psi_m\|_{B(X_m(t,\mgamma^m),L^2_a)}<\infty.$
By the definitions of $\Phi_m$ and $\Psi_m$, it is clear that
$\Psi_m\Phi_m=I$ on $L^2_a$ and $\Phi_m\Psi_m=I$ on $X_m(t,\mgamma^m)$.
Thus we prove \eqref{eq:LBT} defines an isomorphism between
$X_m(t,\mgamma^m)$ and $L^2_a$ uniformly bounded with respect to $t$
and $\mgamma^m$.
\end{proof}
Let 
$$Y_m(t,\mgamma^m)=\left\{w\in L^2_a:\int_\R w\pd_x\pd_{\gamma_i}v^mdx=
\int_\R w\pd_x\pd_{k_i}v^mdx=0.\right\}$$
Note that $\pd_x\pd_{\gamma_i}v^m$ and $\pd_x\pd_{k_i}v^m$
$(1\le i\le m)$ are secular mode solutions of the adjoint equation of
\eqref{eq:w}. We will show that $w^{m-1}$ satisfies the symplectical
orthogonality condition for $v^{m-1}$ if and only if $w^m$ satisfy 
the symplectical orthogonality condition for $v^m$.
\begin{lemma}
  \label{lem:orth2}
Let $a\in(-2k_1,2k_1)$ and let $\Phi(t,\mgamma^m)$ be as in
Lemma \ref{prop:bij2}. Suppose $2\le m\le N$ and \eqref{eq:btphase}.
Then $\Phi_m(t,\mgamma^m)(Y_m(t,\mgamma^m))=Y_{m-1}(t,\mgamma^{m-1})$.
\end{lemma}
\begin{proof}
We abbreviate $\gamma^m_i$ as $\gamma_i$ $(1\le i\le m)$ if there is no
confusion.
Differentiating \eqref{eq:m,m-1} with respect to $\gamma_i$ and
$k_i$ $(1\le i\le m-1)$, we have
\begin{equation}
  \label{eq:list}
B^*\pd_{\gamma_i}v^m=A^*\pd_{\gamma_i}v^{m-1}, \quad
B^*\pd_{k_i}v^m=A^*\left(\pd_{k_i}v^{m-1}
+(\pd_{k_i}\gamma_i^{m-1})\pd_{\gamma_i}v^{m-1}\right).
\end{equation}
Using \eqref{eq:list} and the fact that 
$Aw^m=Bw^{m-1}$ and $2\pd_x=B^*-A^*$,
we compute
\begin{align*}
  2(w^m,\pd_x\pd_{\gamma_i}v^m)
=& (w^m,(B^*-A^*)\pd_{\gamma_i}v^m)
\\= & (w^m,A^*(\pd_{\gamma_i}v^{m-1}-\pd_{\gamma_i}v^m))
\\=&  (Bw^{m-1},\pd_{\gamma_i}v^{m-1}-\pd_{\gamma_i}v^m)
\\=& (w^{m-1},(B^*-A^*)\pd_{\gamma_i}v^{m-1})
\\=& 2(w^{m-1},\pd_x\pd_{\gamma_i}v^{m-1}),
\end{align*}
and
$$
(w^m,\pd_x\pd_{k_i}v^m)=
(w^{m-1},\pd_x\pd_{k_i}v^{m-1})
+(\pd_{k_i}\gamma_i^{m-1})(w^{m-1},\pd_x\pd_{\gamma_i}v^{m-1}).$$
Therefore $w^{m}\in Y_m(t,\mgamma^m)$ if and only if
$w^{m-1}\in Y_{m-1}(t,\mgamma^{m-1})$.
This completes the proof of Lemma \ref{lem:orth2}.
\end{proof}

Now we are in position to prove linear stability of $N$-soliton solutions.
We first establish a decay estimate for \eqref{eq:w}.
\begin{proposition}
\label{prop:wNdecay}
Let $0<k_1<\cdots<k_N$, $a\in(0,2k_1)$ and let $t_0$ be a real number.
Suppose that $w^N\in C((-\infty,t_0];L^2_{-a})$ is a solution of
\begin{equation}
  \label{eq:wN}
\left\{\begin{aligned}
& \pd_tw^N+\pd_x^3w^N+12(\pd_xv^N)(\pd_xw^N)=0
\quad\text{for $x\in\R$, $t<t_0$,}
\\ & w^N(t_0)\in Y_N(t_0,\mgamma^N).
    \end{aligned}\right.
\end{equation}
Then $w^N(t)\in Y_N(t,\mgamma^N)$  for $t\le t_0$ and
$$\|w^N(t)\|_{L^2_{-a}}\le Me^{-a^3(t-s)}\|w^N(s)\|_{L^2_{-a}}
\quad\text{for every $t\le s\le t_0$,}$$
where $M$ is a positive constant depending only on $k_1,\cdots,k_N$.
Furthermore, there exists a positive constant $M'=M'(\mk,l,b)$ for any $l\in\N$ and $b>a^3$
such that
$$\|e^{-ax}w^N(t)\|_{H^l}\le M'(t-s)^{-\frac12}e^{-b(t-s)}\|w^N(s)\|_{L^2_{-a}}
\quad\text{for every $t<s\le t_0$.}$$
\end{proposition}
\begin{proof}[Proof of Proposition \ref{prop:wNdecay}]
First, we will prove that $w^N\in Y_N(t,\mgamma^N)$ for every $t \le s$.
Since $v^N$ is a solution of \eqref{eq:KdV2} and $\pd_{\gamma_i}v^N$ and
$\pd_{k_i}v^N$ $(1\le i\le N)$ are solutions of \eqref{eq:LKdV} with
$\varphi_N=\pd_xv_N$, we have for $1\le i\le N$,
\begin{align*}
\frac{d}{dt}(w^N,\pd_{\gamma_i}v^N) =&
(\pd_tw^N,\pd_{\gamma_i}v^N)+(w^N,\pd_t\pd_{\gamma_i}v^N)=0,
\\
\frac{d}{dt}(w^N,\pd_{k_i}v^N) =&
(\pd_tw^N,\pd_{k_i}v^N)+(w^N,\pd_t\pd_{k_i}v^N)=0.
\end{align*}
Combining the above with $w^N(t_0)\in Y_N(t_0,\mgamma^N)$, 
we have $w^N(t)\in Y_m(t,\mgamma^m)$ for every $t\le t_0$.
\par
Let $w^0(t)=\Phi_1(t,\mgamma^1)^{-1}\cdots\Phi_N(t,\mgamma^N)^{-1}w^N(t).$
Lemmas \ref{lem:orth2}, \ref{prop:bij1} and \ref{prop:bij2} imply that
a map
$\Phi_1(t,\mgamma^1)^{-1}\cdots\Phi_N(t,\mgamma^N)^{-1}$ is well defined
on $Y_N(t,\mgamma^N)$ and
we have $w^0(t)\in C([0,\infty);L^2_a(\R))$ and 
\begin{equation}
  \label{eq:wNw0equiv}
C^{-1}\|w^0(t)\|_{L^2_{-a}}\le \|w^N(t)\|_{L^2_{-a}}\le
C\|w^0(t)\|_{L^2_{-a}},  
\end{equation}
where $C$ is positive constant depending only on $\mk^N$ and $a\in(0,2k_1)$.
Combining \eqref{eq:wNw0equiv} with \eqref{eq:LBT} for $m=1,\cdots,N$, we see that
there exists a $C_l>0$ depending only on $\mk$ and $l\in\N$ such that
\begin{equation}
  \label{eq:wNw0equiv'}
C_l^{-1}\|e^{-ax} w^0(t)\|_{H^l}\le \|e^{-ax}w^N(t)\|_{H^l}\le
C_l\|e^{-ax}w^0(t)\|_{H^l}.  
\end{equation}
\par
Lemma \ref{lem:2} implies that
\begin{equation}
\label{eq:null}
 \pd_tw^0+\pd_x^3w^0=0\quad\text{for $t>s$ and $x\in\R$.}
\end{equation}
It follows from  \cite[Lemma 9.1]{K} that  for any $a>0$ and $t\le s$,
\begin{gather}
  \label{eq:w0decay}
\|w^0(t)\|_{L^2_{-a}(\R)}\le  e^{-a^3(t-s)}\|w^0(s)\|_{L^2_{-a}(\R)},
\\ \label{eq:w0decay'}
 \|e^{-ax}w^0(t)\|_{H^l(\R)}\le  \{1+(3a(t-s))^{-\frac{l}{2}}\} e^{-a^3(t-s)}
\|w^0(s)\|_{L^2_{-a}(\R)}.
\end{gather}
Proposition \ref{prop:wNdecay} follows immediately from \eqref{eq:wNw0equiv},
\eqref{eq:wNw0equiv'}, \eqref{eq:w0decay} and \eqref{eq:w0decay'}.
Thus we complete the proof.
\end{proof}
\begin{proof}[Proof of Theorem \ref{thm:linearizedKdV}]
Let $U(t,s)$ denotes the evolution operator associated with
\begin{equation}
  \label{eq:LKdV'}
\left\{
  \begin{aligned}
& \pd_tw+\pd_x^3w^N+12\pd_x((\pd_xv^N(t))w)=0\quad\text{for $x\in\R$, $t>s$,}
\\ & w(s)\in L^2_a.    
  \end{aligned}\right.
\end{equation}
Since \eqref{eq:LKdV'} is the adjoint equation of \eqref{eq:wN}, it follows 
from Proposition \ref{prop:wNdecay} that for every $t\ge s$ and $f\in L^2_{-a}$,
\begin{align*}
 & \|\mathcal{Q}(s)^*U(t,s)^*\mathcal{Q}(t)^*(t)f\|_{L^2_{-a}}
\le Me^{a^3(t-s)}\|f\|_{L^2_{-a}},
 \\ & \|e^{-ax}\mathcal{Q}(s)^*U(t,s)^*\mathcal{Q}(t)^*(t)f\|_{H^l}
\le M'(t-s)^{-\frac{l}{2}}e^{b(t-s)}\|f\|_{L^2_{-a}},
\end{align*}
since $\mathcal{Q}(t)^*$ is a projection to $Y_N(t,\mgamma^N)$ associated with \eqref{eq:wN}.
By a standard duality argument,
\begin{align*}
 & \|U(t,s)\mathcal{Q}(s)f\|_{L^2_a} \le Me^{a^3(t-s)}\|f\|_{L^2_a},
\\ & \|U(t,s)\mathcal{Q}(s)f\|_{L^2_a} \le M'e^{b(t-s)}(t-s)^{-\frac{l}{2}}\|e^{ax}f\|_{H^{-l}}.
 \end{align*}
Thus we prove Theorem \ref{thm:linearizedKdV}.
\end{proof}
\bigskip

\appendix
\section{Size of $u_c$ and $\rho_c$}
\label{sec:size}
\begin{claim}
  \label{cl:ucsize}
Let $c=1+\frac{1}{6}\eps^2$, $a\in(\frac14\eps,\frac74\eps)$ and
let $i$ and $j$ be nonnegative integers. Then
\begin{align*}
& \|\pd_x^i\pd_c^ju_c\|_{l^2_a\cap l^2_{-a}}=O(\eps^{\frac32+i-2j}),
\quad \|J^{-1}\pd_x^i\pd_c^ju_c\|_{l^2_{-a}}=O(\eps^{\frac12+i-2j}),
\\ &
\|\pd_x^i\pd_c^ju_c\|_{l^\infty_{_a}\cap l^\infty_{-a}}=O(\eps^{2+i-2j}),
\quad \|J^{-1}\pd_x^i\pd_c^ju_c\|_{l^\infty\cap l^\infty_{-a}}=O(\eps^{1+i-2j}).
\end{align*}
\end{claim}
To estimate $l^2$-norm of $u_c$, we need the following.
\begin{claim}
  \label{cl:31}
Let $f\in H^1(\R)$. Then  $\sum_{n\in\Z}f(n)^2\le 2\|f\|_{H^1}^2.$
\end{claim}
\begin{proof}
Since $f(n)^2\le  2\int_n^{n+1}(f(x)^2+f'(x)^2)dx$ for any $n\in\Z$,
we have
\begin{align*}
\sum_{n\in\Z}f(n)^2 \le 2\sum_{n\in\Z} \int_n^{n+1}(f(x)^2+f'(x)^2)dx
=2\|f\|_{H^1(\R)}^2.
\end{align*}
\end{proof}
\begin{proof}[Proof of Claim \ref{cl:ucsize}]
Claim \ref{cl:ucsize} follows from (P4), Claim \ref{cl:31}
and the fact that $\|J^{-1}\|_{B(l^2_{-a})}=O(\eps^{-1})$.
\end{proof}
\begin{claim}
  \label{cl:intsize}
Let $0<k_1<k_2$ and $a\in[0,\frac74\eps)$.
Then there exists an $\eps_*>0$ such
that if $\eps\in(0,\eps_*)$ and $c_i=1+\frac{k_i^2\eps^2}6$ for $i=1$, $2$,
\begin{align*}
& \|\pd_x^{\alpha_1}\pd_c^{\beta_1}u_{c_1}(\cdot-x_1)
\pd_x^{\alpha_2}\pd_c^{\beta_2}u_{c_1}(\cdot-x_2)\|_{l^\infty}
=O(\eps^{4+\alpha_1+\alpha_2-2(\beta_1+\beta_2)}
e^{-k_1a|x_2(t)-x_1(t)|}),
\\ &
\|\pd_x^{\alpha_1}\pd_c^{\beta_1}u_{c_1}(\cdot-x_1)
\pd_x^{\alpha_2}\pd_c^{\beta_2}u_{c_1}(\cdot-x_2)\|_{l^1}
=O(\eps^{3+\alpha_1+\alpha_2-2(\beta_1+\beta_2)}e^{-k_1a|x_2(t)-x_1(t)|}).
\end{align*}
\end{claim}
\begin{proof}
  Claim \ref{cl:intsize} follows from Claim \ref{cl:ucsize}.
\end{proof}
 \begin{claim}
   \label{cl:4}
Let $a_1,\cdots,a_N\in\R$ and
$I=\{\sum_{i=1}^N\theta_ia_i:
 0\le \theta_i\le 1\text{ for $1\le i\le N$}\}$.
Suppose $f\in C^2(\R)$ and $f(0)=0$. Then
$$\left|f(\sum_{1\le i\le N}a_i)-\sum_{1\le i\le N}f(a_i)\right|
\le \sup_{x\in I}|f''(x)|\sum_{i\ne j}|a_ia_j|.$$
 \end{claim}
 \begin{proof}
Let $b=\sum_{1\le i\le N}a_i$. By the mean value theorem,
\begin{align*}
\left|f(b)-\sum_{1\le i\le N}f(a_i)\right|
=& \left|\sum_{1\le i\le N} \int_0^1(f'(s_1 b)-f'(s_1 a_i))ds_1 a_i\right|
\\ =& 
\left|
\sum_{1\le i\le N} \int_0^1\int_0^1f''(s_1(s_2b+(1-s_2)a_i)ds_1ds_2
a_i(b-a_i)\right|
\\ \le &
\sup_{x\in I}|f''(x)|\sum_{i=1}^N|a_i||b-a_i|.
\end{align*}
Thus we prove Claim \ref{cl:4}.
 \end{proof}
Now we estimate size of $\rho_c$. 
\begin{claim}
Let $a\in[0,2k_1\eps)$. Then
  \label{cl:3}
\begin{gather*}
\|\pd_x^i\pd_c^j\rho_c\|_{l^2_a\cap l^2_{-a}}
+\|J^i\pd_c^j\rho_c\|_{l^2_a\cap l^2_{-a}}
=O(\eps^{\frac32+i-2j}),\\
\|\pd_x^i\pd_c^j\rho_c\|_{l^2_a\cap l^2_{-a}}
+\|J^i\pd_c^j\rho_c\|_{l^\infty_a\cap l^\infty_{-a}}
=O(\eps^{2+i-2j}).
\end{gather*}
\end{claim}
\begin{proof}
Noting that $(H''(u_c)-I)\pd_xu_c=O(r_c\pd_xr_c)$,
we see that Claim \ref{cl:3} follows from Claim \ref{cl:ucsize} and
Claim \ref{cl:cpdJ} below.
\end{proof}
\begin{claim}
  \label{cl:cpdJ}
Let $c=1+\frac{\eps^2}{6}$ and $a\in(0,2)$.
There exists a positive number $\eps_0$ such that
$$\sup_{\eps\in(0,\eps_0)}\eps^2
\|\pd_x(c\pd_x+J)^{-1}\|_{B(L^2_{a\eps}\cap L^2_{-a\eps})}
<\infty.$$
\end{claim}
\begin{proof}
Since
$$\mathcal{F}\pd_x(c\pd_x+J)^{-1}=
\frac{i\xi}{c^2\xi^2-4\sin^2\frac\xi2}
\begin{pmatrix}
  -ci\xi & e^{i\xi}-1 \\ 1-e^{-i\xi} & -ci\xi
\end{pmatrix},$$
we have 
$$
\|\pd_x(c\pd_x+J)^{-1}\|_{B(L^2_a)}\le \sup_{\xi\in\R}|m(\xi+ia\eps)|,
$$ where
$m(\xi)=\xi^2(c^2\xi^2-4\sin^2\frac{\xi}{2})^{-1}$.
\par
Using 
$$c^2-\frac{4\sin^2\frac{\xi}{2}}{\xi^2}
=\frac1{12}(\xi^2+4\eps^2)+O(\xi^4+\eps^4),$$
we have
$\sup_{\eps\in(0,\eps_0)}\eps^2
\sup_{\xi\in (-\eps^{\frac23},\eps^{\frac23})}
|m(\xi+ia\eps)|<\infty.$
Suppose  $|\xi|\ge \eps^{\frac23}$. 
Obviously, 
$$\inf_{\eps\in(0,\eps_0)}\inf_{|\xi|\ge \eps^{\frac23}}
\left|c+\frac{2\sin\frac{\xi+ia\eps}{2}}{\xi+ia\eps}\right|>0,$$
and since $0\le \cosh\frac{a\eps}{2}-1=O(\eps^2)$ and
$1-\frac{2\sin{\frac\xi2}}{\xi}\gtrsim\eps^{\frac43}$,
\begin{align*}
\left|c(\xi+ia)-2\sin\frac{\xi+ia\eps}{2}\right|
 \ge & |\xi|
\left(c-\cosh\frac{a\eps}{2}\frac{2\sin\frac\xi2}{|\xi|}\right)
\\ \ge & (c-1)|\xi| \\ \gtrsim & \eps^2|\xi+ia\eps|.
\end{align*}
Combining the above, we conclude Claim \ref{cl:cpdJ}.
\end{proof}

To prove Lemma \ref{lem:FPUprb}, we need the following:
\begin{claim}
  \label{cl:J-1u}
Let $a$ be a positive number, $u=(u_1,u_2)\in l^2_a\cap l^2_{-a}$ and
$v=(v_1,v_2)\in l^2_a\cap l^2_{-a}$. Then 
\begin{equation}
  \label{eq:clj11}
\la u,J^{-1}v\ra=\la u_1,\sum_{k=-\infty}^0 e^{k\pd}v_2\ra
+\la v_1,\sum^{\infty}_{k=1} e^{k\pd}u_2\ra.  
\end{equation}
Especially, $\la u,J^{-1}u\ra=\la u_1,1\ra\la u_2,1\ra$, and as $l\to\infty$,
\begin{gather*}
\la u,J^{-1}e^{l\pd}v\ra
=O(a^{-1}e^{-la}\|u\|_{l^2_a\cap l^2_{-a}}\|v\|_{l^2_a\cap l^2_{-a}}),
\\
\la u,J^{-1}e^{l\pd}v\ra=\la u_1,1\ra\la v_2,1\ra+\la u_2,1\ra\la v_1,1\ra
+O(a^{-1}e^{la}\|u\|_{l^2_a\cap l^2_{-a}}\|v\|_{l^2_a\cap l^2_{-a}}).
\end{gather*}
\end{claim}
\begin{proof}
Eq. \eqref{eq:clj11} follows from \eqref{eq:J-1} and the others follows immediately from
\eqref{eq:clj11}.
\end{proof}

\section{Proof of Lemma \ref{lem:FT-FS}}
\label{sec:apa}

\begin{proof}[Proof of Lemma \ref{lem:FT-FS}]
Let $a(n)=(2\pi)^{-\frac12}\int_\T g(\xi)e^{in\xi}d\xi$.
By  Parseval's identity,
\begin{align*}
 \left\|\int_\T\widetilde{f}(\xi)g(\xi-\xi_1)d\xi_1\right\|_{L^2(\T)}
=& \|f(n)a(n)\|_{l^2}
\\ \lesssim & \|f\|_{L^\infty(\R)}\|g\|_{L^2}.
\end{align*}
\par
Next we prove (ii). By \cite{GGKM}, there exist positive constants
$A_{i_1,\cdots,i_n}$ such that
$$\det(1+C_N)=1+\sum_{n=1}^N \sum_{1\le i_1\le \cdots\le i_n\le N}
A_{i_1,\cdots,i_n}e^{-2(\theta_{i_1}+\cdots+\theta_{i_n})}.$$
Hence $\varphi_{N}(t,z;\mk,\mgamma)$ is analytic on
$\{z\in\C: |\Im z|\le\delta\}$ and
$\sup_{|y|\le \delta}\|\varphi_{N}(t,\cdot+iy;\mk,\mgamma)
\|_{L^1(\R)}<\infty$. 
By the Paley-Wiener theorem \cite[Theorem 9.14]{RS2},
\begin{equation}
  \label{eq:PWd}
  \widehat{r_{N,\eps}}(t,\xi;\mk,\mgamma)
=\eps \widehat{r_{N,1}}(t,\eps^{-1}\xi;\mk,\mgamma)
=O(e^{-\delta|\xi|/\eps}).
\end{equation}
Making use of \eqref{eq:PWd} and the Poisson summation formula, we have
\begin{align*}
\left|\widetilde{r}_{N,\eps}(t,\xi_1,\mgamma)
-\widehat{r}_{N,\eps}(t,\xi,\mgamma))\right|=&
\left|\sum_{n\ne0}\widehat{r}_{N,\eps}(t,\xi+2n\pi,\mgamma))\right|
\\ \lesssim & 
\sum_{n\ge1}e^{-n\pi\delta/\eps}
\lesssim  e^{-\pi\delta/\eps}
\quad\text{for $\xi\in[-\pi,\pi]$.}
\end{align*}
\end{proof}

\section{Relation between secular term conditions of FPU and KdV}
\label{sec:apb}

A multi-soliton solution resolves into a train of $1$-solitons as
$t\to\infty$ (\cite{GGKM}). In fact, we have the following.
\begin{lemma}
\label{lem:KdV-Nsol}
Let $0<k_1<\cdots<k_n$ and $\gamma_i\in\R$ for $1\le i\le N$.
Then
$$
\varphi_N(t,x;\mk,\mgamma)=\sum_{1\le j\le N}
k_j^2\sech^2\tilde{\theta}_j+2\frac{d^2}{dx^2}\log(1+R),$$
where  $\tilde{\theta}_j=k_j(x-4k_j^2t-\tilde{\gamma_j})$ and 
\begin{gather*}
\tilde{\gamma}_N=\gamma_N-\frac{1}{2k_N}\log(2k_N),\quad
\\ \tilde{\gamma}_i=\gamma_i-\frac{1}{2k_i}\log(2k_i)
-\frac{1}{2k_i}\sum_{j=i+1}^N\log\left(\frac{k_j+k_i}{k_j-k_i}\right)
\quad\text{for $1\le i\le N-1$,}\end{gather*}
and there exist positive numbers $a$, $b$ and $\delta$ such that
\begin{equation}
 \label{eq:resol1}
 \sum_{\substack{1\le i\le N\\ \alpha_1,\alpha_2,\alpha_3\ge0}}\sup_{x\in\R}|\cosh(ax)
 \pd_x^{\alpha_1}\pd_{k_i}^{\alpha_2}\pd_{\gamma_i}^{\alpha_3}R(t,x)|\le \delta e^{-bt}
\quad \text{for $t\ge0$,}
\end{equation}
where $\delta$ is chosen as a function of $L:=\inf_{1\le j\le N-1}(\gamma_{j+1}-\gamma_j)$
satisfying $\delta(L)\to0$ as $L\to\infty$.
 Moreover, for any $a\in[0,2)$, there exists a positive number $b'>0$
such that
$$
\sum_{\substack{1\le i\le N\\ \alpha_1,\alpha_2,\alpha_3\ge0}}
\|e^{-a\theta_1}\pd_x^{\alpha_1}\pd_{k_i}^{\alpha_2}
\pd_{\gamma_i}^{\alpha_3}R\|_{L^2}\le \delta e^{-b't}
\quad \text{for $t\ge0$.}$$
\end{lemma}
\begin{proof}
The former part of Lemma \ref{lem:KdV-Nsol} is a slight modification of Theorem 2.1 in Haragus-Sattinger \cite{HS} and can be seen easily from their proof.
The latter part also follows immediately from their proof.
In fact, \cite{HS} tells us that
 $$\left| \pd_x^{\alpha_1}\pd_{k_i}^{\alpha_2}\pd_{\gamma_i}^{\alpha_3}R \right|
 \lesssim \sum_{2\le m\le N}\frac{1}{1+e^{-2\theta_m}},$$
and
 \begin{align*}
\frac{1}{1+e^{-2\theta_m}}=\frac{1}{1+\exp(-\frac{2k_m}{k_1}\theta_1)
\exp\{8k_m(k_m^2-k_1^2)t+4k_m(\gamma_m-\gamma)\}}.
 \end{align*}
 Thus we have Lemma \ref{lem:KdV-Nsol}.
\end{proof}
\par
Now we are in position to prove Lemma \ref{lem:KdVPb}.
\begin{proof}[Proof of Lemma \ref{lem:KdVPb}]
 For $i=1,\cdots,N$, let
 \begin{align*}
& \xi_i^1(\tau)=\pd_{\gamma_i}\varphi_N(\tau,x;\mk,\mgamma),\quad
\xi_i^2(\tau)=\pd_{k_i}\varphi_N(\tau,x;\mk,\mgamma),
\\ & \eta_i^1(\tau)=\int_{-\infty}^x\pd_{\gamma_i}
\varphi_N(\tau,y;\mk,\mgamma)dy,\quad
\eta_i^2(\tau)=\int_{-\infty}^x\pd_{\gamma_i}\varphi_N(\tau,y;\mk,\mgamma)dy,
 \end{align*}
and let  
$$
\mathcal{A}_{KdV}=\begin{pmatrix}\mathcal{A}_{KdV}^{i\,j}
\end{pmatrix}_{\substack{i=1,\cdots,N\rightarrow,\\j=1,\cdots,N\downarrow}},
\quad
\mathcal{A}_{KdV}^{i\,j}=
\begin{pmatrix}\la \xi_i^1,\eta_j^1\ra & \la \xi_i^2,\eta_j^1\ra
\\ \la \xi_i^1,\eta_j^2\ra & \la \xi_i^2,\eta_j^2\ra
\end{pmatrix}.$$
Then we have
$$\mathcal{P}(\tau)f
=\sum_{i=1}^N(\alpha_i\xi_i^1(\tau)+\beta_i\xi_i^2(\tau)),$$
where $\alpha_i$ and $\beta_i$ are given by
$$
\mathcal{A}_{KdV}
\begin{pmatrix}\alpha_i\\ \beta_i
\end{pmatrix}_{i=1,\cdots,N\downarrow}
=\begin{pmatrix}\la f,\eta_j^1(\tau)\ra \\ \la f,\eta_j^2(\tau)\ra
\end{pmatrix}_{j=1,\cdots,N\downarrow}.
$$
  Since $\xi_i^k$ $(1\le i\le N,\,k=1,2)$ are solutions of \eqref{eq:LKdV}
 and $\eta_j^l$ are solutions of the adjoint equation of \eqref{eq:LKdV},
 $\la \xi_i^k,\eta_j^l\ra$ are independent of $t$. Let $\phi_k(x)=k^2\sech^2kx$.
 By Lemma \ref{lem:KdV-Nsol}, 
\begin{align}
\label{eq:etaj1}
\eta_j^1=& -\phi_{k_j}(x-4k_j^2t-\tilde{\gamma}_k)+R_{1,j},
\\
\label{eq:etaj2}
\eta_j^2=& \int_{-\infty}^x \pd_k\phi_{k_j}(y-4k_j^2t-\tilde{\gamma}_j)dy
-\sum_{m<j}\frac{\pd\tilde{\gamma}_m}{\pd k_j}\phi_{k_m}(y-4k_i^2t-\tilde{\gamma}_i)+R_{2,j},
\end{align}
where $R_{1,j}=2\pd_{k_j}\pd_x\log(1+R)$ and $R_{2,j}=2\pd_{\gamma_j}\pd_x\log(1+R)$.
Observing limit as $t\to\infty$, we have
$\la \xi_i^k,\eta_j^l\ra=0$ if $i\ne j$ and $(k,l)\ne (2,2)$, and
 \begin{align*}
 & \la \xi_i^1,\eta_i^1\ra=0,\quad
 \la \xi_i^1,\eta_i^2\ra=-\la \xi_i^2,\eta_i^1\ra=\frac12\frac{d}{dk_i}\|\phi_{k_i}\|_{L^2}^2\ne0
 \quad\text{for $i=1,\cdots,N$.} 
 \end{align*}
If $i<j$,
\begin{align*}
 & \la \xi_i^2,\eta_j^2\ra\\ =&  \lim_{t\to\infty}
\left\la  \pd_{k_i}\phi_{k_i}
-\sum_{l=1}^{i-1}\frac{\pd\tilde{\gamma}_l}{\pd k_i}\pd_x\phi_{k_l},
 \int_{-\infty}^x\pd_{k_j}\phi_{k_j}dy
-\sum_{m=1}^{j-1}\frac{\pd\tilde{\gamma}_m}{\pd k_j}\phi_{k_m}
 \right\ra \\ =&0.
\end{align*}
It follows from above that $\mathcal{A}_{KdV}^{i\,j}=O$ if $i<j$,
that $\mathcal{A}_{KdV}$ is invertible, and that
\begin{equation}
  \label{eq:lemPbPc}
  \begin{split}
\|\mathcal{P}(\tau)f\|_{L^2_a}  \lesssim &
\sum_{l,m}\sum_{i\le j}|\la f,\eta_j^m\ra| \|\xi_i^l\|_{L^2_a}
 \\ \lesssim & \sum_{l, m}\sum_{i\le j} 
 e^{-a\{(4(k_j^2-k_i^2)t+\tilde{\gamma}_j-\tilde{\gamma}_i\}}
 \|e^{-a(\cdot-4k_j^2t-\tilde{\gamma}_j)}\eta_j^m\|_{L^2}
\\ \quad & \times
\|e^{a(\cdot-4k_i^2t-\tilde{\gamma}_i)}\xi_i^l\|_{L^2}\|f\|_{L^2_a}
\\ \lesssim & \|f\|_{L^2_a}.
  \end{split}
\end{equation}
Thus we complete the proof of Lemma \ref{lem:KdVPb}.
\end{proof}
Next we prove Lemma \ref{lem:projerr}.
\begin{proof}[Proof of Lemma \ref{lem:projerr}]
By \eqref{eq:orth3} and Parseval's identity,
  \begin{align*}
& \left|\left\la w(t),J^{-1}\pd_{\gamma_i}u_{N,\eps}\right\ra\right|
\\=&
\left|\left\la f(t,\xi), e^{ic_{1,\eps}t\xi}
P(\xi)^*\widehat{J}^{-1}
\mathcal{F}_n\pd_{\gamma_i}u_{N,\eps}(t,\xi,\mgamma)\right\ra\right|
\\ =&
\frac{1}{2} \left|\left\la \tau_{ik_1\eps}f(t,\xi),
\tau_{-ik_1\eps}\left\{e^{ic_{1,\eps}t\xi}
(\sin\tfrac{\xi}{2})^{-1}\sigma_3P(\xi)^*\right\}
\mathcal{F}_n\pd_{\gamma_i}u_{N,\eps}(t,\xi,\mgamma)
\right\ra\right|
\\= & \le \eps^{\frac12}\delta_2e^{-k_1\gamma_1}\|\tau_{ik_1\eps}f(t)\|_{L^2}.
  \end{align*}
As in the proof of Lemma \ref{lem:FT-FS}, we see that
$$
\|\mathcal{F}_n\pd_{\gamma_i}u_{N,\eps}(t,\xi-ik_1\eps,\mgamma)
-\mathcal{F}_x\pd_{\gamma_i}u_{N,\eps}(t,\xi-ik_1\eps,\mgamma)
\|_{L^2(-\pi,\pi)}=O(e^{-c/\eps})$$
for a $c>0$.
Combining the above with $P(0)^*\pd_{\gamma_i}u_{N,\eps}
={}^t(\sqrt{2}r_{N,\eps},0)$ and the facts that
$$
|P(\xi-ik_1\eps)^*-P^*(0)|+
\left|\frac{1}{\sin\tfrac{\xi-ik_1\eps}{2}}-\frac{2}{\xi-ik_1\eps}
\right|\lesssim |\xi-ik_1\eps|
\quad\text{for $\xi\in[-\pi,\pi]$,}
$$
and that $\|e^{-k_1\eps(\cdot-c_{1,\eps}t-\eps^{-1}\gamma_1)}\pd_x\pd_{\gamma_i}
r_{N,\eps}(t,\cdot;\mk,\mgamma)\|_{l^2}=O(\eps^{\frac52})$,
we have
\begin{align*}
& \left\la \tau_{ik_1\eps}f_+(t),
\tau_{-ik_1\eps}\left\{
e^{ic_{1,\eps}t\xi}
\xi^{-1}\widehat{\pd_{\gamma_i}r_{N,\eps}}(t,\xi;\mk,\mgamma)
\right\} \right\ra
\\ =&O(\eps^{\frac12}(\delta_2+\eps^2)e^{-k_1\gamma_1} \|\tau_{ik_1\eps}f(t)\|_{L^2}).
\end{align*}
Let $h_2$, $h_3\in L^2(\R)$ such that
\begin{align*}
  & h_1(\tau,y)+h_2(\tau,y)=
\frac{1}{\sqrt{2\pi}}\int_{-\pi\eps^{-1}}^{\pi\eps^{-1}}
e^{ic_{1,\eps}\eps t(\eta+ik_1)}
f_{\#}(t,\eps(\eta+ik_1))e^{iy\eta}dy,
\\ & h_3(\tau,y)=\frac{1}{\sqrt{2\pi}}\int_{-\pi\eps^{-1}}^{\pi\eps^{-1}}
(f_{2,+}(t,\eps\eta)+f_{3,+}(t,\eps\eta))e^{iy\eta}dy.
\end{align*}
Then
\begin{align*}
& \left\la \tau_{ik_1\eps}f_+(t),
\tau_{-ik_1\eps}\left\{ e^{ic_{1,\eps}t\xi}
\xi^{-1}\widehat{\pd_{\gamma_i}r_{N,\eps}}(t,\xi;\mk,\mgamma)
\right\} \right\ra
\\=& \eps\left\la \widehat{h_1}+\widehat{h_3},
\tau_{-ik_1}\left\{\eta^{-1}e^{4ik_1^2\tau\eta}
\widehat{\pd_{\gamma_i}\varphi_N}(\tau,\eta;\mk,\mgamma)
\right\} \right\ra
\\=& 
\eps\left\la h_1+h_3, e^{-k_1y}\int_{-\infty}^y
\pd_{\gamma_i}\varphi_N(\tau,y_1+4k_1^2\tau;\mk,\mgamma)dy_1
\right\ra.
\end{align*}
Since $\widehat{h_3}(\tau,\eta)=0$ for $\eta\in[-K,K]$,
it follows from Lemma \ref{lem:KdV-Nsol} that
\begin{align*}
& \left|\left\la h_3, e^{-k_1y}\int_{-\infty}^y
\pd_{\gamma_i}\varphi_N(\tau,y_1+4k_1^2\tau;\mk,\mgamma) dy_1\right\ra\right|
\\ \lesssim &
\|h_3\|_{H^{-2}} \left\|e^{-k_1y}\int_{-\infty}^y
\pd_{\gamma_i}\varphi_N(\tau,y_1+4k_1^2\tau;\mk,\mgamma)
dy_1\right\|_{H^2}
\\ \lesssim & K^{-2}e^{-k_1\{4(k_i^2-k_1^2)\tau+\gamma_i\}}\|h_3\|_{L^2}.
\end{align*}
Combining the above, we have
\begin{equation}
  \label{eq:projerr1}
  \begin{split}
& \eps^{\frac12}\left|\left\la h_1, e^{-k_1\{y-4(k_i^2-k_1^2)\tau-\gamma_i\}}
\int_{-\infty}^y \pd_{\gamma_i}\varphi_N(\tau,y_1+4k_1^2\tau;\mk,\mgamma)dy_1
\right\ra\right|
\\ \lesssim &
\eps^{\frac12}K^{-2}\|h_3\|_{L^2}+(\eps^2+\delta_2)\|\tau_{ik_1\eps }f\|_{L^2}
\\ \lesssim & (K^{-2}+\eps^2+\delta_2)\|\tau_{ik_1\eps} f\|_{L^2}.
  \end{split}
\end{equation}
Similarly,
\begin{equation}
  \label{eq:projerr2}
  \begin{split}
& \eps^{\frac12}\left|\left\la h_1, e^{-k_1\{y-4(k_i^2-k_1^2)\tau-\gamma_i\}}
\int_{-\infty}^y
\pd_{k_i}\varphi_N(\tau,y_1+4k_1^2\tau;\mk,\mgamma)dy_1
\right\ra\right|
\\ \lesssim &
(K^{-2}+\eps^2+\delta_2)\|\tau_{ik_1\eps} f\|_{L^2}.
  \end{split}
\end{equation}
By \eqref{eq:projerr1}, \eqref{eq:projerr2} and\eqref{eq:lemPbPc},
we have
\begin{align*}
& \|\mathcal{P}_1(\tau)h_1(\tau)\|_{L^2_a} \\\lesssim &
\sum_{l,m} \sum_{i\le j}
|\la h_1(\tau),e^{-k_y}\eta_j^m(\tau,\cdot+4k_1^2\tau)\ra|
\|e^{k_1y}\xi_i^l(\tau,\cdot+4k_1^2\tau)\|_{L^2}
\\ \lesssim & \sum_{l, m}\sum_{i\le j} 
 e^{-a\{(4(k_j^2-k_i^2)\tau+\gamma_j-\gamma_i\}}
|\la h_1(\tau),e^{-k_1\{y-4(k_j^2-k_1^2)\tau-\gamma_j\}}
\eta_j^m(\tau,\cdot+4k_1^2\tau)\ra|
 \\ \lesssim & \eps^{-\frac12}(K^{-2}+\eps^2+\delta_2)
\|\tau_{ik_1\eps} f\|_{L^2}.  
\end{align*}
Thus we complete the proof of Lemma \ref{lem:projerr}.
\end{proof}
\section{Proof of Lemma \ref{lem:linearstability}}
\label{sec:lemlinearstability}
To begin with, we compare spectral projection associated with a solitary wave
solution of FPU and that associated with KdV $1$-soliton.
\begin{lemma}
  \label{lem:approx-adjef}
Let $\eps>0$, $a\in(\eps/8,2\eps)$ and $c=1+\eps^2/6$.
Then
\begin{gather*}
\left\|J^{-1}\pd_x{u}_c+\phi_\eps\begin{pmatrix}1\\-1\end{pmatrix}\right\|_{l^2_{-a}}=O(\eps^{\frac52}),
\\
\left\|J^{-1} \pd_c u_\eps +\int_{-\infty}^n\pd_c\phi_\eps
\begin{pmatrix}1\\-1\end{pmatrix}\right\|_{l^2_{-a}}
=O(\eps^{-\frac12}).
\end{gather*}
\end{lemma}
To prove Lemma \ref{lem:approx-adjef}, we need the following:
\begin{claim}
  \label{cl:32}
Suppose $a\in(0,1)$ and $f\in C_0^\infty(\R)$. Then
  \begin{align*}
& \|(e^\pd-1)^{-1}\pd_xf\|_{L^2_a}\lesssim 
\|f\|_{L^2_a}+a^{-1}\|\pd_xf\|_{L^2_a},
\\ &  \|(e^\pd-1)^{-1}\pd_xf-f\|_{L^2_a}\lesssim 
a\|f\|_{L^2_a}+a^{-1}\|\pd_x^2f\|_{L^2_a},
\\ &    
\|(e^\pd-2+e^{-\pd})^{-1}\pd_x^2f\|_{L^2_a}\lesssim 
\|f\|_{L^2_a}+a^{-2}\|\pd_x^2f\|_{L^2_a},
\\ & \|(e^\pd-2+e^{-\pd})^{-1}\pd_x^2f-f\|_{L^2_a}\lesssim 
a^2\|f\|_{L^2_a}+a^{-2}\|\pd_x^4f\|_{L^2_a}.
  \end{align*}
\end{claim}
  \begin{proof}
Let $g(x)=e^{ax}f(x)$. Using $|e^{i\xi-a}-1|\ge 1-e^{-a}\gtrsim a$ and
$$|e^{i\xi-a}-i\xi+a-1|\lesssim a^2+|\xi|^2,$$ we have
 \begin{align*}
\|(e^\pd-1)^{-1}\pd_xf\|_{L^2_a}
=  \left\|\frac{i\xi-a}{e^{i\xi-a}-1}\hat{g}\right\|_{L^2}
\lesssim  \|f\|_{L^2_a}+a^{-1}\|\pd_xf\|_{L^2_a},
 \end{align*}
and 
 \begin{align*}
 \|(e^\pd-1)^{-1}\pd_xf-f\|_{L^2_a}
=& \left\|\frac{e^{i\xi-a}-i\xi+a-1}{e^{i\xi-a}-1}\hat{g}\right\|_{L^2}
\\ \lesssim &
a\|\hat{g}\|_{L^2}+a^{-1}\|\xi^2\hat{g}\|_{L^2}
\\ \lesssim &
a\|f\|_{L^2_a}+a^{-1}\|\pd_x^2f\|_{L^2_a}.
\end{align*}
Similarly, by using $|e^{i\xi-a}+e^{-i\xi+a}-2|\ge 4\sinh^2(a/2)$ and
$$|e^{i\xi-a}+e^{-\xi+a}-2-(i\xi-a)^2|
\lesssim \xi^4+a^4,$$
we have 
 \begin{align*}
\|(e^\pd-2+e^{-\pd})^{-1}\pd_x^2f\|_{L^2_a}
=&\left\|\frac{(i\xi-a)^2}{e^{i\xi-a}-2+e^{-i\xi+a}}\hat{g}\right\|_{L^2}
\\ \lesssim  & \|f\|_{L^2_a}+a^{-2}\|\pd_x^2f\|_{L^2_a},
 \end{align*}
and 
 \begin{align*}
\|(e^\pd-2+e^{-\pd})^{-1}\pd_x^2f-f\|_{L^2_a}
=&\left\|\frac{(i\xi-a)^2}{e^{i\xi-a}-2+e^{-i\xi+a}}\hat{g}
-\hat{g}\right\|_{L^2}
\\ \lesssim  & 
a^2\|\hat{g}\|_{L^2}+a^{-2}\|\xi^4\hat{g}\|_{L^2}
\\ \lesssim  & 
a^2\|f\|_{L^2_a}+a^{-2}\|\pd_x^4f\|_{L^2_a}.
\end{align*}
  \end{proof}
  \begin{claim}
\label{cl:33}
Let $a\in\R$ and $f\in H^1(\R)$. Then
$$\left\|f(x)-\int_x^{x\pm1}f(y)dy\right\|_{L^2_a(\R)}
\le \max(1,e^{-a})\|f'\|_{L^2_a(\R)}.$$
  \end{claim}
  \begin{proof}
Since
    \begin{align*}
  \left|f(x)-\int_x^{x+1}f(y)dy\right|
=& \left|\int_x^{x+1}\int_y^xf'(t)dtdy\right|
\le \left(\int_x^{x+1}f'(t)^2dt\right)^{\frac12},
    \end{align*}
we have
    \begin{align*}
\left\|f(x)-\int_x^{x+1}f(y)dy\right\|_{L^2_a}^2
\le &  \int_\R\left(e^{2ax}\int_x^{x+1}f'(t)^2dt\right)dx
\\ \le & \max(1,e^{-2a})\|f'\|_{L^2_a}^2.
    \end{align*}
  \end{proof}

\begin{proof}[Proof of Lemma \ref{lem:approx-adjef}]
By the definition of $u_c$, we have
\begin{equation}
  \label{eq:pcrc}
p_c=-c(e^\pd-1)^{-1}\pd_x r_c, \quad
J^{-1}\pd_x u_c=(-c(e^\pd-2+e^{-\pd})^{-1}\pd_x^2r_c, (e^\pd-1)^{-1}\pd_xr_c).
\end{equation}
Thus by Claims \ref{cl:31} and \ref{cl:32},
\begin{align*}
& \left\|J^{-1}\pd_xu_c+\phi_\eps \begin{pmatrix}1\\-1\end{pmatrix}\right\|_{l^2_{-a}}
\\ \le & \left\|J^{-1}\pd_xu_c+\phi_\eps \begin{pmatrix}1\\-1\end{pmatrix}\right\|_{H^1_{-a}}  
\\ \le & \left\|\begin{pmatrix}
c(e^\pd-2+e^{-\pd})^{-1}\pd_x^2(r_c-\phi_\eps)\\
-(e^{\pd}-1)^{-1}\pd_x(r_c-\phi_\eps) \end{pmatrix}\right\|_{H^1_{-a}}
+\left\|\begin{pmatrix}
(c(e^\pd-2+e^{-\pd})^{-1}\pd_x^2-1)\phi_\eps\\
(-(e^{\pd}-1)^{-1}\pd_x+1)\phi_\eps \end{pmatrix}\right\|_{H^1_{-a}}
\\ \lesssim &
\|r_c-\phi_\eps\|_{H^1_{-a}}+a^{-2}\|\pd_x^2(r_c-\phi_\eps)\|_{H^1_{-a}}
+\eps^2(\|\phi_\eps\|_{H^1_{-a}}+a^{-2}\|\pd_x^2\phi_\eps\|_{H^1_{-a}})
\\ & +a^2\|\phi_\eps\|_{H^1_{-a}}+a^{-2}\|\pd_x^4\phi_\eps\|_{H^1_{-a}}
 +a\|\phi_\eps\|_{H^1_{-a}}+a^{-1}\|\pd_x^2\phi_\eps\|_{H^1_{-a}}
\\ 
\lesssim & (\eps^{\frac72}+a\eps^{\frac32})\left(1+\frac{\eps}{a}\right)^2
+a^2\eps^{\frac32}\left(1+\frac{\eps^2}{a^2}\right)^2
=O(\eps^{\frac52}).
\end{align*}
\par
Since $\|J^{-1}\|_{B(l^2_{-a}\times l^2_{-a})}\lesssim a^{-1}$,
\begin{align*}
& \left\|J^{-1}\pd_cu_c +\int_{-\infty}^n\pd_c\phi_\eps
\begin{pmatrix}1\\-1\end{pmatrix}\right\|_{l^2_{-a}}
\\ \lesssim &
a^{-1}\left\|\pd_cu_c+J\int_{-\infty}^n\pd_c\phi_\eps
\begin{pmatrix}1\\-1\end{pmatrix}\right\|_{l^2_{-a}}
\\ \lesssim &
a^{-1}\left\|\pd_cr_c-\int_x^{x+1}\pd_c\phi_\eps\right\|_{H^1_{-a}}
+a^{-1}\left\|\pd_cp_c+\int_{x-1}^{x}\pd_c\phi_\eps\right\|_{H^1_{-a}}.
\end{align*}
By \eqref{eq:pcrc} and Claim \ref{cl:32},
\begin{align*}
& \|\pd_cp_c+\pd_cr_c\|_{l^2_{-a}}\\ \le & \|(e^\pd-1)^{-1}\pd_xr_c\|_{l^2_{-a}}
+\|\{c(e^\pd-1)^{-1}\pd_x-1\}\pd_cr_c\|_{l^2_{-a}}
\\ \lesssim &
\|r_c\|_{H^1_{-a}}+a^{-1}\|\pd_xr_c\|_{H^1_{-a}}
+a\|\pd_cr_c\|_{H^1_{-a}}+a^{-1}\|\pd_x^2\pd_cr_c\|_{H^1_{-a}}+\eps^2\|\pd_cr_c\|_{H^1_{-a}}
\\ \lesssim & \eps^{\frac32}(1+a^{-1}\eps)+a\eps^{-\frac12}(1+a^{-2}\eps^2)
=O(\eps^{\frac12}).
\end{align*}
Combining the above with (P4) and Claim \ref{cl:33},
we have
\begin{align*}
& \left\|J^{-1} \begin{pmatrix}\pd_cr_c\\ \pd_cp_c\end{pmatrix}
+\int_{-\infty}^n\pd_c\phi_\eps\begin{pmatrix}1\\-1\end{pmatrix}\right\|_{l^2_{-a}}
\\ \lesssim &
a^{-1}(\|\pd_cr_c-\pd_c\phi_\eps\|_{H^1_{-a}}+
\|\pd_cp_c+\pd_c\phi_\eps\|_{H^1_{-a}}+\|\pd_x\pd_c\phi_\eps\|_{H^1_{-a}})
\\ \lesssim & a^{-1}\eps^{\frac12}=O(\eps^{-\frac12}).
\end{align*}
\end{proof}
Finally, we will prove Lemma \ref{lem:linearstability}
\begin{proof}[Proof of Lemma \ref{lem:linearstability}]
We assume that $k=N$. The other cases can be shown in the same way.
  By (P4) and Lemma \ref{lem:KdV-Nsol}, we can choose $\mk$ and $\mgamma$
so that 
$$\sum_{i=0,1}\sup_{t\ge0,\,x\in\R}
\left|\pd_x^i(\widetilde{U}_N(t)-u_{N,\eps}(t,x,\mgamma)\right|
\le \delta(L)\eps^{2+i}+O(\eps^4).$$
Combining Lemmas \ref{lem:KdV-Nsol} and \ref{lem:approx-adjef}
with \eqref{eq:etaj1} and \eqref{eq:etaj2}, we obtain
\eqref{eq:orth3} from \eqref{eq:orth3'}.
Thus we prove Lemma \ref{lem:linearstability}.
\end{proof}

\end{document}